%% file: PaperQRACs.tex
\pdfoutput=1

\documentclass[letterpaper]{article}
\usepackage[hmargin=1.15in,vmargin=1in]{geometry}

\usepackage{amsmath,amssymb,amsthm}
\usepackage{graphicx} 
\usepackage{ifthen}   
\usepackage{multirow} 
\usepackage{caption}  

\usepackage{xcolor}

\usepackage[
  pdftitle={Quantum Random Access Codes with Shared Randomness},
  pdfauthor={Andris Ambainis, Debbie Leung, Laura Mancinska, Maris Ozols},
  colorlinks=true,
  linkcolor=blue, citecolor=blue, urlcolor=blue,
  bookmarksnumbered=true
]{hyperref}

\newcommand{\vc}[1]{{\boldsymbol #1}}
\newcommand{\abs}[1]{\left|#1\right|}
\newcommand{\set}[1]{\left\{#1\right\}}

\newcommand{\floor}[1]{\left\lfloor #1 \right\rfloor}
\newcommand{\norm}[1]{\left\|#1\right\|}
\newcommand{\bra}[1]{\left\langle#1\right|}
\newcommand{\ket}[1]{\left|#1\right\rangle}
\newcommand{\braket}[2]{\left\langle#1|#2\right\rangle}
\newcommand{\C}{\mathbb{C}}
\newcommand{\R}{\mathbb{R}}
\newcommand{\NOT}{\textstyle\mathop{\mathrm{NOT}}}
\newcommand{\XOR}{\oplus}
\newcommand{\V}{\checkmark}
\DeclareMathOperator*{\E}{\mathbb{E}}
\DeclareMathOperator*{\tr}{Tr}

\newcommand{\mx}[1]{\begin{pmatrix}#1\end{pmatrix}}
\newcommand{\smx}[1]{\bigl(\begin{smallmatrix}#1\end{smallmatrix}\bigr)}

\newcommand{\rac} [1]{\texorpdfstring{$#1 \mapsto 1$}{#1 --> 1}}                     
\newcommand{\racm}[2]{\texorpdfstring{$#1 \overset{p}{\mapsto} #2$}{#1 -p-> #2}}     
\newcommand{\racp}[1]{\racm{#1}{1}}                                                  
\newcommand{\racx}[2]{\texorpdfstring{$#1 \overset{#2}{\longmapsto} 1$}{#1 -#2-> 1}} 

\newcommand{\pure}{\mathcal{P}}             
\newcommand{\rand}{\mathcal{S}}             
\newcommand{\uniform}{\eta}                 
\newcommand{\vis}{{\set{\vc{v}_i}}}         
\newcommand{\rxs}{{\set{\vc{r}_x}}}         
 \newcommand{\Svr}{S_{\vc{v},\vc{r}}}
 \newcommand{\Sv}{S_{\vc{v}}}
\newcommand{\base}{\mathcal{B}}

\newcommand{\spacer}{\ }

\newcommand{\p}[1]
{\ifthenelse{\equal{#1}{2}}      {0.853553} 
{\ifthenelse{\equal{#1}{3}}      {0.788675} 
{\ifthenelse{\equal{#1}{4}}      {0.741481} 
{\ifthenelse{\equal{#1}{4[sym]}} {0.733253} 
{\ifthenelse{\equal{#1}{5}}      {0.713578} 
{\ifthenelse{\equal{#1}{6}}      {0.694046} 
{\ifthenelse{\equal{#1}{6[sym]}} {0.694042} 
{\ifthenelse{\equal{#1}{6[ort]}} {0.686973} 
{\ifthenelse{\equal{#1}{9}}      {0.656893} 
{\ifthenelse{\equal{#1}{9[sym]}} {0.656393} 
{\ifthenelse{\equal{#1}{15}}     {0.620355} 
{\ifthenelse{\equal{#1}{15[sym]}}{0.620183} 
{}}}}}}}}}}}}}

\newcommand{\vertskip}{\vspace{1.3ex}}

\newcommand{\subsubsubsection}[1]{{\vertskip\noindent\indent\textbf{#1}\nopagebreak[1]}}

\newcommand{\image}[5]{
  \begin{figure}
    \begin{center}
         \includegraphics[width=#1\textwidth]{#2}
      \caption[#3]{#4}
      \label{#5}
    \end{center}
  \end{figure}
}

\newcommand{\doubleimage}[8]{
  \begin{figure}
    \centering
    \begin{minipage}[t]{0.4\linewidth} 
      \centering
      \includegraphics[width=\textwidth]{#1}
      \caption[#2]{#3}
      \label{#4}
    \end{minipage}
    \hspace{0.06\linewidth} 
    \begin{minipage}[t]{0.4\linewidth}
      \centering
      \includegraphics[width=\textwidth]{#5}
      \caption[#6]{#7}
      \label{#8}
    \end{minipage}
  \end{figure}
}

\newtheorem{theorem}{Theorem}
\newtheorem{lemma}{Lemma}
\newtheorem*{problem}{Problem}
\newtheorem*{conjecture}{Conjecture}
\theoremstyle{definition}
\newtheorem*{definition}{Definition}
\newtheorem*{example}{Example}

\begin{document}


\input{FrontPagePaper} 

\section{Introduction}

\subsection{Random access codes}


In general \emph{random access code} (or simply RAC) stands for ``encoding a long message into fewer bits with the ability to recover (decode) any one of the initial bits (with some probability of success)''. A random access code can be characterized by the symbol ``\racm{n}{m}'' meaning that $n$ bits are encoded into $m$ and any one of the initial bits can be recovered with probability at least $p$. We require that $p > 1/2$ since $p = 1/2$ can be achieved by guessing. In this paper we consider only the case when $m = 1$. So we have the following problem:

\begin{problem}[Classical]
There are two parties---Alice and Bob. Alice is asked to encode some classical \mbox{$n$-bit} string into $1$ bit and send this bit to Bob. We want Bob to be able to recover any one of the $n$ initial bits with high success probability.
\end{problem}

Note that Alice does not know in advance which bit Bob will need to recover, so she cannot send only that bit. If they share a quantum channel then we have the quantum version of the previous problem:

\begin{problem}[Quantum]
Alice must encode her classical \mbox{$n$-bit} message into $1$ \emph{qubit} (quantum bit) and send it to Bob. He performs some measurement on the received qubit to extract the required bit (the measurement that is used depends on which bit is needed).
\end{problem}

Both problems look similar, however the quantum version has an important feature. In the classical case the fact that Bob can recover any one of the initial bits implies that he can actually recover all of them---each with high probability of success. Surprisingly in the quantum case this is not true, because after the first measurement the state of the qubit will be disturbed and further attempts to extract more information can fail.

\subsection{History and applications} \label{sect:History}

As noted in \cite{Galvao,Severini}, the idea behind \emph{quantum random access codes} or QRACs is very old (relative to quantum information standards). It first appeared in a paper by Stephen Wiesner \cite{Wiesner} published in 1983 and was called \textit{conjugate coding}. Later these codes were \mbox{re-discovered} by Ambainis et al. in \cite{DenseCoding1,DenseCoding2}. They show that there exists \racx{2}{0.85} QRAC and mention its immediate generalization to \racx{3}{0.79} QRAC due to Chuang (see also \cite{No41} and \cite{Severini} for more details). However, Hayashi et al. \cite{No41} show that it is impossible to construct a \racp{4} QRAC with $p>1/2$. We will discuss these results more in Sect.~\ref{sect:KnownQRACs}.

There has also been work on \racm{n}{m} codes with $m>1$, see \cite{DenseCoding1,DenseCoding2,Nayak}. Ambainis et al. \cite{DenseCoding1} show that if a \racm{n}{m} QRAC with $p>1/2$ exists, then $m = \Omega(n / \log n)$, which was later improved by Nayak \cite{Nayak,DenseCoding2} to $m \geq (1-H(p)) n$, where $H(p) = -p \log p - (1-p) \log (1-p)$ is the \emph{binary entropy function}. Other generalizations include: considering \mbox{$d$-valued} bits instead of qubits \cite{Galvao,Severini} and recovering several rather than a single bit \cite{Hypercontractive}.

Originally quantum random access codes were studied in the context of quantum finite automata \cite{DenseCoding1,DenseCoding2,Nayak}. However, they also have applications in quantum communication complexity \cite{Galvao,Klauck,Aaronson,Gavinsky}, in particular for \emph{network coding} \cite{No41,NetworkCoding} and \emph{locally decodable codes} \cite{Kerenidis,KerenidisDeWolf,Wehner,Hypercontractive}. Recently results on quantum random access codes have been applied for quantum state learning \cite{AaronsonLearnability}.

Experimental feasibility of QRACs and their relation to \emph{contextuality} and \emph{non-locality} has been discussed in \cite[Chapter 7]{Galvao}. Recently a similar protocol called \emph{parity-oblivious multiplexing} has been considered in \cite{Spekkens}. It has an additional cryptographic constraint that Alice is not allowed to transmit any information about the parity of the input string. In addition \cite{Spekkens} also discuss the first experimental demonstration of \rac{2} and \rac{3} QRACs.

We want to emphasize the setting in which the impossibility of \racp{4} QRAC with $p>1/2$ was proved in \cite{No41}: Alice is allowed to perform a locally randomized encoding of the given string into a one-qubit state and Bob is allowed to perform different \emph{positive operator-valued measure} (POVM) measurements to recover different bits. This is the most general setting when information is encoded into a one-qubit state and both parties are allowed to use randomized strategies, but only have access to local coins. However, we can consider an even more general setting---when both parties share a common coin. This means that Alice and Bob are allowed to cooperate by using some shared source of randomness to agree on which strategy to use. We will refer to this source as a \emph{shared random string} or \emph{shared randomness} (SR). Note that shared randomness is a more powerful resource than local randomness, since parts of the shared random string can be exclusively used only by Alice or Bob to simulate local coins. It turns out that in this new setting \racp{4} QRAC is possible with $p>1/2$. In fact, \racp{n} QRACs with $p > 1/2$ can be constructed for all $n \geq 1$ (see Sect.~\ref{sect:LowerBounds}).

\subsection{Outline of results}

In Sect.~\ref{sect:ClassicalRACs} we study classical \rac{n} random access codes with shared randomness. In Sect.~\ref{sect:Yao} we introduce Yao's principle that is useful for understanding both classical and quantum codes. A classical code that is optimal for all $n$ is presented in Sect.~\ref{sect:OptimalClassical} and the asymptotic behavior of its success probability is considered in Sect.~\ref{sect:ClassicalBound}.

In Sect.~\ref{sect:QuantumRACs} we study quantum random access codes with shared randomness. In Sect.~\ref{sect:KnownQRACs} we discuss what is known in the case when shared randomness is not allowed, i.e., \rac{2} and \rac{3} QRACs and the impossibility of \rac{4} QRAC. In Sect.~\ref{sect:UpperBound} we give an upper bound of success probability of QRACs with SR and generalize it in Sect.~\ref{sect:GeneralUpperBound} for POVM measurements. In Sect.~\ref{sect:LowerBounds} we give two constructions of \racp{n} QRAC with SR and $p > 1/2$ for all $n \geq 2$ that provide a lower bound for success probability.

In Sect.~\ref{sect:Constructions} we try to find optimal QRACs with SR for several small values of $n$. In particular, in Sect.~\ref{sect:Numerical} we discuss QRACs obtained by numerical optimization, and in Sect.~\ref{sect:Symmetric} we consider symmetric constructions.

Finally, we conclude in Sect.~\ref{sect:Conclusion} with a summary of the obtained results (Sect. \ref{sect:Summary}), a list of open problems (Sect.~\ref{sect:OpenProblems}) and possible generalizations (Sect.~\ref{sect:Generalizations}).

\section{Classical random access codes} \label{sect:ClassicalRACs}

\subsection{Types of classical encoding-decoding strategies}

As a synonym for random access code we will use the term \emph{strategy} to refer to the joint \emph{encoding-decoding scheme} used by Alice and Bob. Two measures of how good the strategy is will be used: the \emph{worst case success probability} and the \emph{average success probability}. Both probabilities must be calculated over all possible pairs $(x,i)$ where $x\in\set{0,1}^n$ is the input and $i\in\set{1,\dotsc,n}$ indicates which bit must be recovered. We are interested in the worst case success probability, but in our case according to Yao's principle (introduced in Sect.~\ref{sect:Yao}) the average success probability can be used to estimate it.

Depending on the computational model considered, different types of strategies are allowed. The simplest type corresponds to Alice and Bob acting deterministically and independently.

\begin{definition}
A \emph{pure classical \rac{n} encoding-decoding strategy} is an ordered tuple $(E,D_1,\dotsc,D_n)$ that consists of an \emph{encoding function} $E: \set{0,1}^n \mapsto \set{0,1}$ and $n$ \emph{decoding functions} $D_i: \set{0,1} \mapsto \set{0,1}$.
\end{definition}

These limited strategies yield RACs with poor performance. This is because Bob can recover all bits correctly for no more than two input strings, since he receives either $0$ or $1$ and acts deterministically in each case. For all other strings at least one bit will definitely be recovered incorrectly, therefore the worst case success probability is $0$. If we allow Alice and Bob to act probabilistically but without cooperation, then we get mixed strategies.

\begin{definition}
A \emph{mixed classical \rac{n} encoding-decoding strategy} is an ordered tuple $(P_E,P_{D_1},\dotsc,P_{D_n})$ of probability distributions. $P_E$ is a distribution over encoding functions and $P_{D_i}$ over decoding functions.
\end{definition}

It is obvious that in this setting the worst case probability is at least $1/2$. This is obtained by guessing---we output either $0$ or $1$ with probability $1/2$ regardless of the input. Formally this means that for each $i$, $P_{D_i}$ is a uniform distribution over two constant decoding functions $0$ and $1$. It has been shown that in this setting for \rac{2} case one cannot do better than guessing, i.e., there is no \racp{2} RAC with worst case success probability $p > 1/2$ \cite{DenseCoding1,DenseCoding2}.

However, we can allow cooperation between Alice and Bob---they can use a shared random string to agree on some joint strategy. 

\begin{definition}
A \emph{classical \rac{n} encoding-decoding strategy with shared randomness} is a probability distribution over pure classical strategies.
\end{definition}

Note that this is the most general randomized setting, since both randomized cooperation and local randomization are possible. This is demonstrated in the following example.

\begin{example}
Consider the following strategy: randomly agree on $i\in\set{1,\dotsc,n}$ and send the \mbox{$i$th} bit; if the \mbox{$i$th} bit is requested, output the received bit, otherwise guess. This strategy can formally be specified as follows: uniformly choose a pure strategy from the set
\begin{equation}
  \bigcup_{i \in \set{1,\dotsc,n}}
  \bigl\{
    (e_i, c_1, \dotsc, c_{i-1}, d, g_1, \dotsc, g_{n-i}) \mid
    c \in \set{d_0,d_1}^{i-1},
    g \in \set{d_0,d_1}^{n-i}
  \bigr\},
\end{equation}
where the encoding function $e_i$ is given by $e_i(x) = x_i$ and decoding functions $d_0$, $d_1$, and $d$ are given by $d_0(b) = 0$, $d_1(b) = 1$, and $d(b) = b$, where $b$ is the received bit. The total amount of required randomness is $n - 1 + \log n$ bits, because one out of $n \cdot 2^{n-1}$ pure strategies must be selected. Note that only $\log n$ of these bits must be shared among Alice and Bob, so that they can agree on the value of $i$. The remaining $n - 1$ random bits are needed only by Bob for choosing random decoding functions $c \in \set{d_0,d_1}^{i-1}$ and $g \in \set{d_0,d_1}^{n-i}$.
\end{example}

Note that the amount of randomness used in the above example can be reduced. Since only one bit must be recovered, there is no need to choose each of the decoding functions independently. Thus Bob needs only one random bit that he will output whenever some bit other than the \mbox{$i$th} bit is requested. This is illustrated in the next example.

\begin{example}
Alice and Bob uniformly sample a pure strategy from the following set:
\begin{equation}
  \bigl\{
    (e_i, \underbrace{c, \dotsc, c}_{i-1}, d, \underbrace{c, \dotsc, c}_{n-i}) \mid
    1 \leq i \leq n,
    c \in \set{d_0,d_1}
  \bigr\}.
\end{equation}
This requires $\log n$ random bits to be shared among Alice and Bob and $1$ private random bit for Bob, i.e., $1 + \log n$ random bits in total.
\end{example}

We are interested in classical strategies with SR, because they provide a classical analogue of QRACs with SR. However, in this setting finding the optimal strategy seems to be hard, therefore we will turn to Yao's principle for help.

\subsection{Yao's principle} \label{sect:Yao}

When dealing with randomized algorithms, it is hard to draw general conclusions (like proving optimality of a certain randomized algorithm) because the possible algorithms may form a continuum. In such situations it is very helpful to apply Yao's principle \cite{Yao}. This allows us to shift the randomness in the algorithm to the input and consider only deterministic algorithms.

Let $\rand$ be a classical strategy with SR. One can think of it as a stochastic process consisting of applying the encoding map $E$ to the input $x$, followed by applying the decoding map $D_i$ to the \mbox{$i$th} bit. Both of these maps depend on the value of the shared random string. The result of $\rand$ is $\rand(x,i) = D_i(E(x))$, which is a stochastic variable over the set $\set{0,1}$. Let $\Pr[\rand(x,i) = x_i]$ denote the probability that the stochastic variable $\rand(x,i)$ takes value $x_i$. Then the worst case success probability of the optimal classical strategy with SR is given by
\begin{equation}
  \max_{\rand} \min_{x,i} \Pr[\rand(x,i) = x_i].
  \label{eq:YaoRandomized}
\end{equation}

Let $\mu$ be some distribution over the input set $\set{0,1}^n \times \set{1,\dotsc,n}$ and let ${\textstyle\Pr_{\mu}[\pure(x,i) = x_i]}$ denote the expected success probability of a pure (deterministic) strategy $\pure$. If the ``hardest'' input distribution is chosen as $\mu$, then the expected success probability of the best pure strategy for this distribution is
\begin{equation}
  \min_{\mu} \max_{\pure} {\textstyle\Pr_{\mu}[\pure(x,i) = x_i]}.
  \label{eq:YaoPure}
\end{equation}

\emph{Yao's principle} states that the quantities given in (\ref{eq:YaoRandomized}) and (\ref{eq:YaoPure}) are equal \cite{Yao}:
\begin{equation}
  \max_{\rand} \min_{x,i} \Pr[\rand(x,i) = x_i] =
  \min_{\mu} \max_{\pure} {\textstyle\Pr_{\mu}[\pure(x,i) = x_i]}.
  \label{eq:Yao}
\end{equation}
Thus Yao's principle provides us with an upper bound for the worst case probability (\ref{eq:YaoRandomized}). All we have to do is to choose an arbitrary input distribution $\mu_0$ and find the best pure strategy $\pure_0$ for it. Then according to Yao's principle we have
\begin{equation}
  {\textstyle\Pr_{\mu_0}[\pure_0(x,i) = x_i]} \geq \max_{\rand} \min_{x,i} \Pr[\rand(x,i) = x_i],
  \label{eq:YaoBound}
\end{equation}
with equality if and only if $\mu_0$ is the ``hardest'' distribution. It turns out that for random access codes the uniform distribution $\uniform$ is the ``hardest''. To prove it, we must first consider the randomization lemma.

\begin{lemma}
$\forall \pure \exists \rand: \min_{x,i} \Pr[\rand(x,i) = x_i] = {\textstyle\Pr_{\uniform}[\pure(x,i) = x_i]},$ where $\uniform$ is the uniform distribution. In other words: the worst case success probability of $\rand$ is the same as the average case success probability of $\pure$ with uniformly distributed input.
\label{lem:Randomization}
\end{lemma}

\begin{proof}
This can be achieved by randomizing the input with the help of the shared random string. Alice's input can be randomized by \mbox{XOR-ing} it with an \mbox{$n$-bit} random string $r$. But Bob's input can be randomized by adding (modulo $n$) a random number $d\in\set{0,\dotsc,n-1}$ to it (assume for now that bits are numbered from $0$ to $n-1$). To obtain a consistent strategy, these actions must be identically performed on both sides, thus a shared random string of $n + \log n$ bits\footnote{We will not worry about how Bob obtains a uniformly distributed $d$ from a string of random bits when $n \neq 2^k$.} is required. Assume that $E$ and $D_i$ are the encoding and decoding functions of the pure strategy $\pure$; then the new strategy $\rand$ is
\begin{align}
  E'(x) &= E(\mathrm{Shift}_d(x \oplus r)), \\
  D'_i(b) &= D_{i + d \bmod n}(b) \oplus r_{i},
\end{align}
where $\mathrm{Shift}_d(s)$ substitutes $s_{i + d \bmod n}$ by $s_i$ in string $s$. Due to input randomization, this strategy has the same success probability for all inputs $(x,i)$, namely
\begin{equation}
  \Pr[\rand(x,i) = x_i] \spacer
    = \sum_{y\in\set{0,1}^n} \sum_{j=0}^{n-1} \frac{1}{2^n \cdot n} \Pr[\pure(y,j) = y_j]
    = {\textstyle\Pr_{\uniform}[\pure(y,j) = y_j]},
\end{equation}
coinciding with the average success probability of the pure strategy $\pure$.
\end{proof}

Now we will show that inequality (\ref{eq:YaoBound}) becomes an equality when $\mu_0 = \uniform$, meaning that the uniform distribution $\uniform$ is the ``hardest''.

\begin{lemma}
The minimum of (\ref{eq:YaoPure}) is reached at the uniform distribution $\uniform$, i.e.,
\begin{equation}
  \min_{\mu} \max_{\pure} {\textstyle\Pr_{\mu}[\pure(x,i) = x_i]} =
  \max_{\pure} {\textstyle\Pr_{\uniform}[\pure(x,i) = x_i]}.
  \label{eq:HardestDistribution}
\end{equation}
\end{lemma}

\begin{proof}
From the previous Lemma we know that there exists a strategy with SR $\rand_0$ such that
\begin{equation}
  \min_{x,i} \Pr[\rand_0(x,i) = x_i] = \max_{\pure} {\textstyle\Pr_{\uniform}[\pure(x,i) = x_i]}
\end{equation}
($\rand_0$ is obtained from the best pure strategy by prepending it with input randomization). However, among all strategies with SR there might be one that is better than $\rand_0$, thus
\begin{equation}
  \max_{\rand} \min_{x,i} \Pr[\rand(x,i) = x_i] \geq  \max_{\pure} {\textstyle\Pr_{\uniform}[\pure(x,i) = x_i]}.
  \label{eq:FromRandomization}
\end{equation}
But if we put $\mu_0=\uniform$ into inequality (\ref{eq:YaoBound}), we obtain
\begin{equation}
  \max_{\pure} {\textstyle\Pr_{\uniform}[\pure(x,i) = x_i]} \geq \max_{\rand} \min_{x,i} \Pr[\rand(x,i) = x_i],
\end{equation}
which is the same as (\ref{eq:FromRandomization}), but with reversed sign. This means that both sides are actually equal:
\begin{equation}
  \max_{\pure} {\textstyle\Pr_{\uniform}[\pure(x,i) = x_i]} = \max_{\rand} \min_{x,i} \Pr[\rand(x,i) = x_i].
  \label{eq:TightBound}
\end{equation}
Applying Yao's principle to the right hand side of (\ref{eq:TightBound}) we obtain the desired equation (\ref{eq:HardestDistribution}).
\end{proof}

\begin{theorem}
For any pure strategy $\pure$
\begin{equation}
  {\textstyle\Pr_{\uniform}[\pure(x,i) = x_i]} \leq \max_{\rand} \min_{x,i} \Pr[\rand(x,i) = x_i],
\end{equation}
with equality if and only if $\pure$ is optimal for the uniform distribution $\uniform$.
\label{thm:SRtoPure}
\end{theorem}

\begin{proof}
To obtain the required inequality, do not maximize the left hand side of equation (\ref{eq:TightBound}), but put an arbitrary $\pure$. It is obvious that we will obtain equality if and only if $\pure$ is optimal.
\end{proof}

This theorem has important consequences---it allows us to consider pure strategies with uniformly distributed input rather than strategies with SR. If we manage to find the optimal pure strategy, then we can also construct an optimal strategy with SR using input randomization\footnote{If the encoding function depends only on the Hamming weight of the input string $x$ (e.g., majority function) and the decoding function does not depend on $i$, there is no need to randomize over $i$, so $n$ instead of $n + \log n$ shared random bits are enough.}. If the pure strategy is not optimal, then we get a lower bound for the strategy with SR.

\subsection{Classical \rac{n} RAC}

Before considering \rac{n} QRACs with shared randomness, we will find an optimal classical \rac{n} RAC with shared randomness and derive bounds for it. 

\subsubsection{Optimal strategy} \label{sect:OptimalClassical}

According to Theorem~\ref{thm:SRtoPure} we can consider only pure strategies. As a pure strategy is deterministic, for each input it gives either a correct or a wrong answer. To maximize the average success probability we must find a pure strategy that gives the correct answer for as many of the $n \cdot 2^n$ inputs as possible---such a strategy we will call an \emph{optimal pure strategy}.

Let us first consider the problem of finding an optimal decoding strategy, when the encoding strategy is fixed. An encoding function $E: \set{0,1}^n \mapsto \set{0,1}$ divides the set of all strings into two parts:
\begin{equation}
  \begin{aligned}
    X_0 &= \set{x\in\set{0,1}^n \mid E(x) = 0}, \\
    X_1 &= \set{x\in\set{0,1}^n \mid E(x) = 1}.
  \end{aligned}
\end{equation}
If Bob receives bit $b$, he knows that the initial string was definitely from the set $X_b$, but there is no way for him to tell exactly which string it was. However, if he must recover only the \mbox{$i$th} bit, he can check whether there are more zeros or ones among the \mbox{$i$th} bits of strings from set $X_b$. More formally, we can introduce the symbol $N_i^b(k)$ that denotes the number of strings from set $X_b$ that have the bit $k$ in \mbox{$i$th} position:
\begin{equation}
  N_i^b(k) = \abs{\set{x \in X_b \mid x_i = k}},
\end{equation}
Therefore the optimal decoding strategy $D_i: \set{0,1} \mapsto \set{0,1}$ for the \mbox{$i$th} bit is
\begin{equation}
  D_i(b) =
  \begin{cases}
    0& \text{if $N_i^b(0) \geq N_i^b(1)$,}\\
    1& \text{otherwise}.
  \end{cases}
  \label{eq:OptimalClassicalDecoding}
\end{equation}
Of course, if $N_i^b(0) = N_i^b(1)$, Bob can output $1$ as well. For pure strategies there are only $4$ possible decoding functions for each bit: $0$, $1$, $b$, or $\NOT b$. But this is still quite a lot so we will consider the two following lemmas. The first lemma will rule out the \emph{constant decoding functions} $0$ and $1$.

\begin{lemma}
For any $n$ there exists an optimal pure classical \rac{n} RAC that does not use constant decoding functions $0$ and $1$ for any bits.
\label{lm:NoConst}
\end{lemma}

\begin{proof}
We will show that if there exists an optimal strategy that contains constant decoding functions for some bits, then there also exists an optimal strategy that does not. Let us assume that there is an optimal strategy with constant decoding function $0$ for the \mbox{$i$th} bit (the same argument goes through for $1$ as well). Then according to equation (\ref{eq:OptimalClassicalDecoding}) we have $N_i^0(0) \geq N_i^0(1)$ and $N_i^1(0) \geq N_i^1(1)$. Note that $N_i^0(0) + N_i^1(0) = N_i^0(1) + N_i^1(1) = 2^{n-1}$, because $x_i=0$ in exactly half of all $2^n$ strings. This means that actually $N_i^0(0) = N_i^0(1)$ and $N_i^1(0) = N_i^1(1)$. If we take a look at (\ref{eq:OptimalClassicalDecoding}) again, we see that in such situation any decoding strategy is optimal and we can use any non-constant strategy instead.
\end{proof}

\begin{lemma}
For any $n$ there exists an optimal pure classical \rac{n} RAC that does not use decoding function $\NOT b$ for any bits.
\label{lm:NoNot}
\end{lemma}

\begin{proof}
We will show that for each pure strategy $\pure$ that uses negation as the decoding function for the \mbox{$i$th} bit, there exists a pure strategy $\pure'$ with the same average case success probability that does not. If $\pure$ consists of encoding function $E$ and decoding functions $D_j$, then $\pure'$ can be obtained from $\pure$ by inverting the \mbox{$i$th} bit before encoding and after decoding:
\begin{align}
  E'(x) &= E(\NOT_i x), \\
  D_j'(b) &= 
    \begin{cases}
      \NOT D_j(b)& \text{if $j=i$,} \\
      D_j(b)& \text{otherwise,}
    \end{cases}
\end{align}
where $\NOT_i$ inverts the \mbox{$i$th} bit of string. It is obvious that $\pure$ and $\pure'$ have the same average success probabilities, because if $\pure$ gives the correct answer for input $(x,i)$ then $\pure'$ gives the correct answer for input $(\NOT_i x, i)$. The same holds for wrong answers.
\end{proof}

\begin{theorem}
The pure classical \rac{n} RAC with identity decoding functions and majority encoding function is optimal.
\label{thm:OptimalClassical}
\end{theorem}

\begin{proof}
According to Lemma~\ref{lm:NoConst} and Lemma~\ref{lm:NoNot}, there exists an optimal pure classical \rac{n} RAC with identity decoding function for all bits. Now we must consider the other part---finding an optimal encoding given a particular (identity) decoding function. It is obvious that in our case optimal encoding must return the majority of bits:
\begin{equation}
  E'(x) =
    \begin{cases}
      0& \text{if $\abs{x} < n/2$,} \\
      1& \text{otherwise,}
    \end{cases}
\end{equation}
where $\abs{x}$ is the Hamming weight of string $x$ (the number of ones in it).
\end{proof}

\subsubsection{Asymptotic bounds} \label{sect:ClassicalBound}

Let us find the exact value of the average success probability for the optimal pure RAC suggested in Theorem~\ref{thm:OptimalClassical}. We will separately consider the even and odd cases.

In the odd case ($n=2m+1$) the average success probability is given by
\begin{equation}
  p(2m+1) = \frac{1}{(2m+1) \cdot 2^{2m+1}}
  \Biggl( 2 \sum_{i=m+1}^{2m+1} i \binom{2m+1}{i} \Biggr),
  \label{eq:Odd}
\end{equation}
where the factor $2$ stands for either zeros or ones being the majority, and $\binom{2m+1}{i}$ stands for the number of strings where the given symbol dominates and appears exactly $i$ times.

In the even case ($n=2m$) there are a lot of strings with the same number of zeros and ones. These strings are bad, because with majority encoding and identity decoding it is not possible to give the correct answer for more than half of all bits. The corresponding average success probability is given by
\begin{equation}
  p(2m) = \frac{1}{2m \cdot 2^{2m}}
  \Biggl( 2 \sum_{i=m+1}^{2m} i \binom{2m}{i} + m \binom{2m}{m} \Biggr),
  \label{eq:Even}
\end{equation}
where the last term stands for the bad strings.

In Appendix~\ref{app:MagicFormulas} we give a combinatorial interpretation of the sums in (\ref{eq:Odd}) and (\ref{eq:Even}). Equations (\ref{eq:Magic1}) and (\ref{eq:Magic2}) derived in Appendix~\ref{app:MagicFormulas} can be used to simplify $p(2m+1)$ and $p(2m)$, respectively. It turns out that both probabilities are equal:
\begin{equation}
  p(2m) = p(2m+1) = \frac{1}{2} + \frac{1}{2^{2m+1}} \binom{2m}{m}.
  \label{eq:Exact}
\end{equation}
These two expressions can be combined as follows:
\begin{equation}
  p(n) = \frac{1}{2} + \frac{1}{2^n} \binom{n-1}{\floor{\frac{n-1}{2}}}.
\end{equation}

\image{0.9}{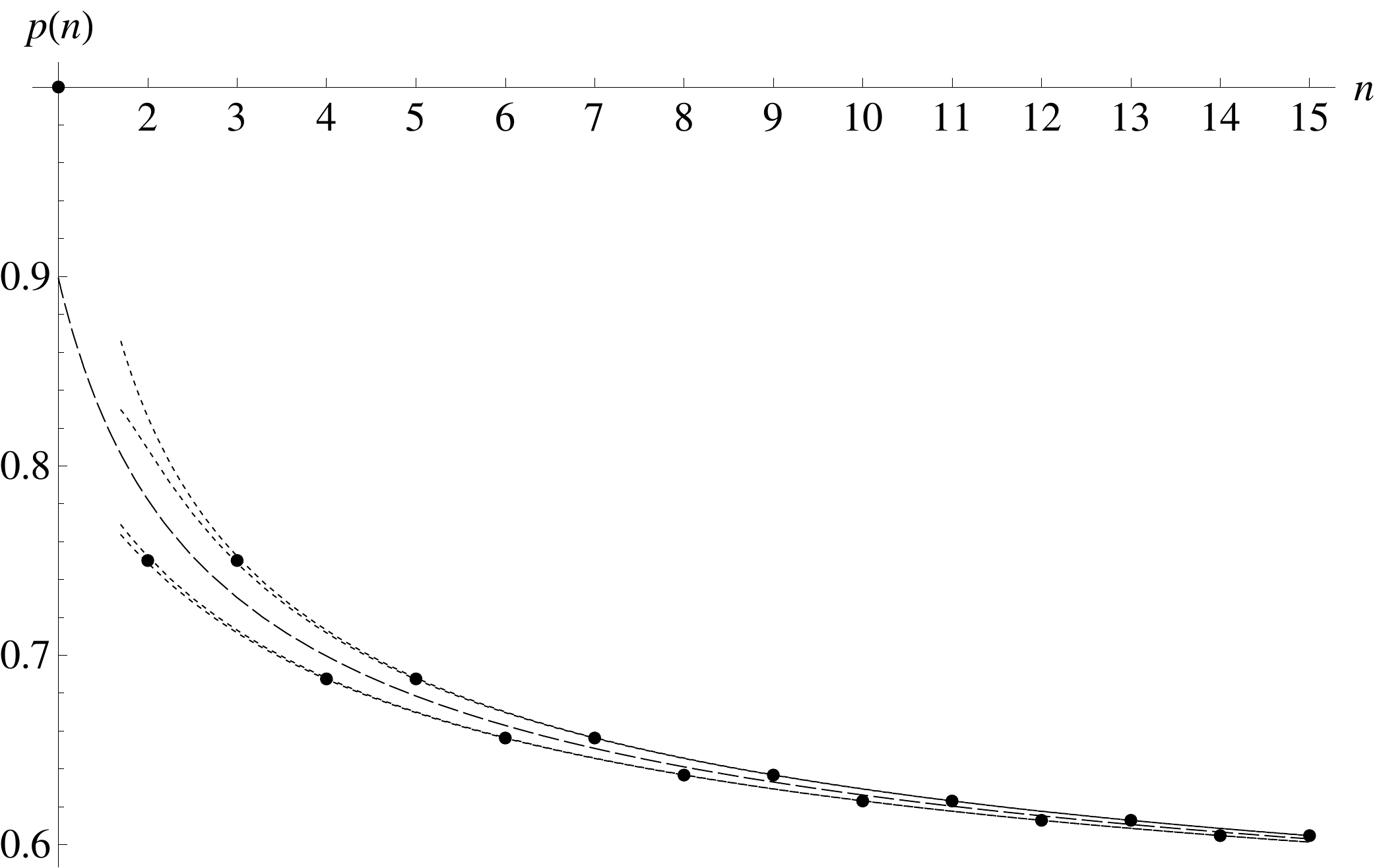}{Success probability for optimal pure classical \rac{n} RAC}{Exact probability of success $p(n)$ for optimal pure classical \rac{n} RAC (black dots) according to (\ref{eq:Exact}) and its approximate value (dashed line) according to (\ref{eq:Approx}). Dotted lines show upper and lower bounds of $p(n)$ for odd and even $n$ according to inequalities (\ref{eq:ApproxOdd}) and (\ref{eq:ApproxEven}).}{fig:Classical}

We can apply Stirling's approximation \cite{Stirling} $m! \approx \left(\frac{m}{e}\right)^m \sqrt{2 \pi m}$ to (\ref{eq:Exact}) and obtain
\begin{equation}
  p(2m) = p(2m+1) \approx \frac{1}{2} + \frac{1}{2 \sqrt{\pi m}}.
  \label{eq:ApproxEvenOdd}
\end{equation}
If we put $m \approx \frac{n}{2}$, then (\ref{eq:ApproxEvenOdd}) turns to
\begin{equation}
  p(n) \approx \frac{1}{2} + \frac{1}{\sqrt{2 \pi n}}.
  \label{eq:Approx}
\end{equation}

We see that the value of (\ref{eq:Approx}) approaches $1/2$ as $n$ increases. Thus the obtained codes are not very good for large $n$, since $p = 1/2$ can be obtained by guessing. We will observe a similar (but slightly better) behavior also in the quantum case. The exact probability (\ref{eq:Exact}) and its approximation (\ref{eq:Approx}) are shown in Fig.~\ref{fig:Classical}.

For odd and even cases asymptotic upper and lower bounds on $p(n)$ can be obtained using the following inequality \cite{Stirling}:
\begin{equation}
  \sqrt{2 \pi n} \left(\frac{n}{e}\right)^n e^\frac{1}{12n+1} < n! <
  \sqrt{2 \pi n} \left(\frac{n}{e}\right)^n e^\frac{1}{12n}.
\end{equation}
For the odd case we have
\begin{equation}
  \frac{\exp\left(\frac{1}{12n-11}-\frac{2}{6n-6}\right)}{\sqrt{2 \pi (n-1)}}
  < p(n) - \frac{1}{2} <
  \frac{\exp\left(\frac{1}{12n-12}-\frac{2}{6n-5}\right)}{\sqrt{2 \pi (n-1)}},
  \label{eq:ApproxOdd}
\end{equation}
but for the even case
\begin{equation}
  \frac{\exp\left(\frac{1}{12n}-\frac{2}{6n+1}\right)}{\sqrt{2 \pi n}}
  < p(n) - \frac{1}{2} <
  \frac{\exp\left(\frac{1}{12n+1}-\frac{2}{6n}\right)}{\sqrt{2 \pi n}}.
  \label{eq:ApproxEven}
\end{equation}
All four bounds are shown in Fig.~\ref{fig:Classical}.

\section{Quantum random access codes} \label{sect:QuantumRACs}

\subsection{Visualizing a qubit}

When dealing with quantum random access codes (at least in the qubit case), it is a good idea to try to visualize them. We provide two ways.

\subsubsection{Bloch sphere representation} \label{sect:BlochSphere}

A \emph{pure qubit state} is a column vector $\ket{\psi} \in \C^2$. It can be expressed as a linear combination $\ket{\psi} = \alpha \ket{0} + \beta \ket{1}$, where $\ket{0} = \smx{1\\0}$ and $\ket{1} = \smx{0\\1}$. The coefficients $\alpha, \beta \in \C$ must satisfy $|\alpha|^2 + |\beta|^2 = 1$. Since the physical state is not affected by the \emph{phase factor} (i.e., $\ket{\psi}$ and $e^{i\phi}\ket{\psi}$ are the same states for any $\phi\in\R$), without the loss of generality one can write
\begin{equation}
  \ket{\psi} =
  \mx{
    \cos{\frac{\theta}{2}} \\
    e^{i\varphi}\sin{\frac{\theta}{2}}
  },
  \label{eq:Psi}
\end{equation}
where $0\leq\theta\leq\pi$ and $0\leq\varphi<2\pi$ (the factor $1/2$ for $\theta$ in (\ref{eq:Psi}) is chosen so that these ranges resemble the ones for \emph{spherical coordinates} in $\R^3$).

For almost all states $\ket{\psi}$ there is a unique way to assign the parameters $\theta$ and $\varphi$. The only exceptions are states $\ket{0}$ and $\ket{1}$, that correspond to $\theta=0$ and $\theta=\pi$, respectively. In both cases $\varphi$ does not affect the physical state. Note that the spherical coordinates with \emph{latitude} $\theta$ and \emph{longitude} $\varphi$ have the same property, namely---the longitude is not defined at poles. This suggests that the state space of a single qubit is topologically a sphere.

Indeed, there is a \mbox{one-to-one} correspondence between pure qubit states and the points on a unit sphere in $\R^3$. This is called the \emph{Bloch sphere representation} of a qubit state. The \emph{Bloch vector} for state (\ref{eq:Psi}) is $\vc{r} = (x,y,z)$, where the coordinates (see Fig.~\ref{fig:UnitVector}) are given by
\begin{equation}
  \left\{
  \begin{aligned}
    x &= \sin{\theta}\cos{\varphi}, \\
    y &= \sin{\theta}\sin{\varphi}, \\
    z &= \cos{\theta}.
  \end{aligned}
  \right.
  \label{eq:UnitVector}
\end{equation}
Given the Bloch vector $\vc{r} = (x,y,z)$, the coefficients of the corresponding state $\ket{\psi} = \alpha \ket{0} + \beta \ket{1}$ can be found as follows \cite[pp.~102]{PrinciplesOfQ}:
\begin{equation}
  \alpha = \sqrt{\frac{z+1}{2}}, \quad \beta = \frac{x+iy}{\sqrt{2(z+1)}}
  \label{eq:AlphaBeta}
\end{equation}
with the convention that $(0,0,-1)$ corresponds to $\alpha=0$ and $\beta=1$.

\doubleimage
{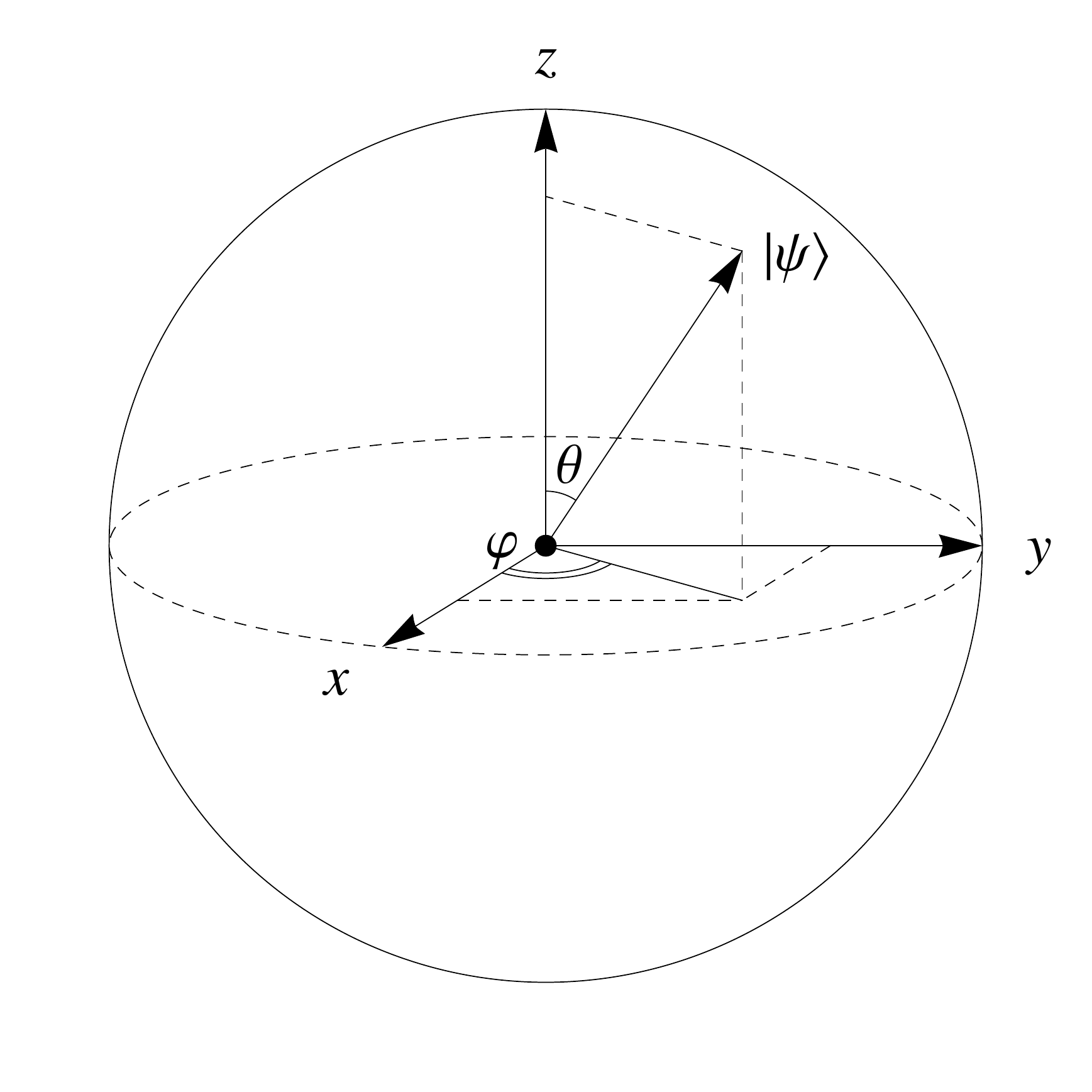}{Angles $\theta$ and $\varphi$ of the Bloch vector}{Angles $\theta$ and $\varphi$ of the Bloch vector corresponding to state $\ket{\psi}$.}{fig:UnitVector}
{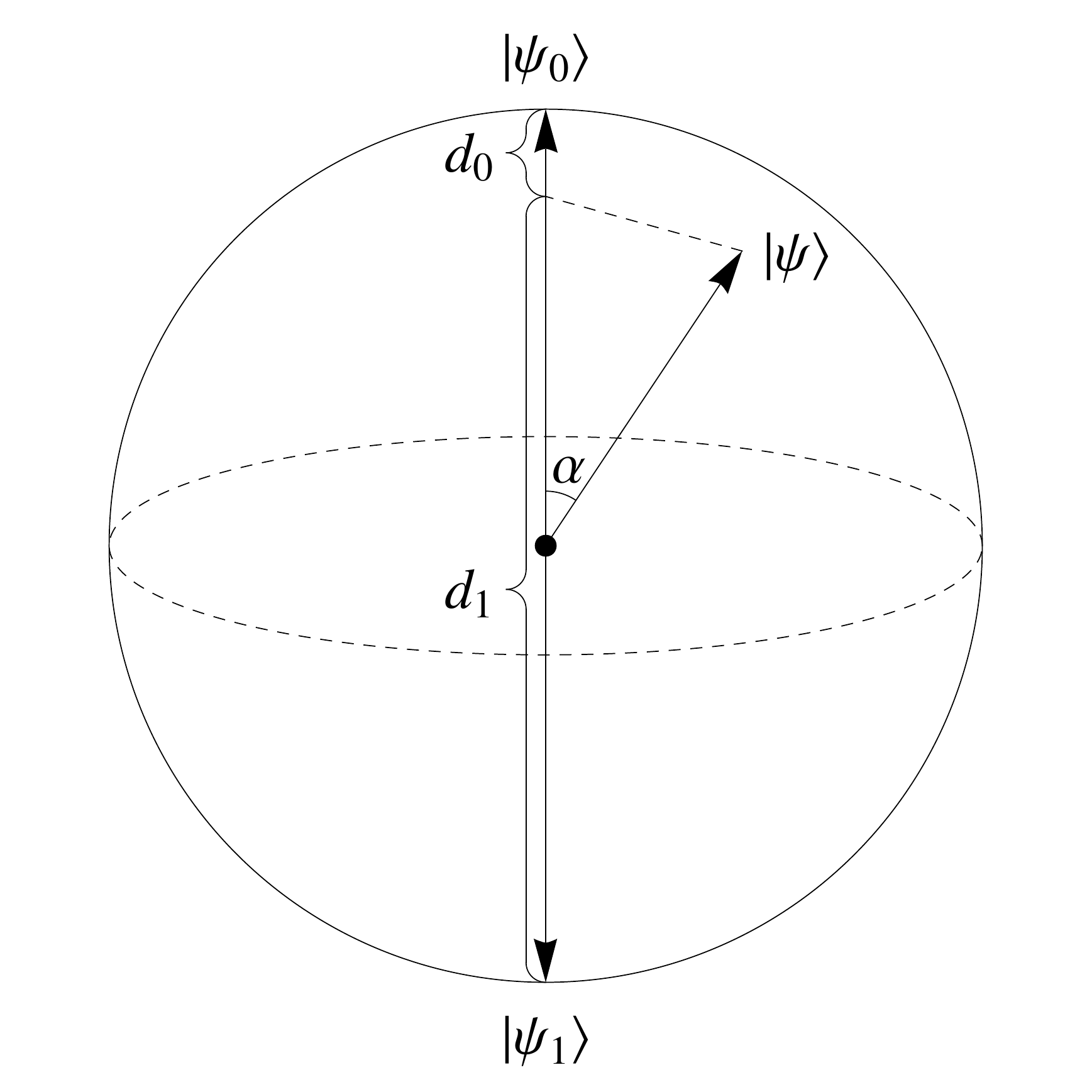}{Geometric interpretation of orthogonal measurement}{Geometric interpretation of orthogonal measurement.}{fig:BlochSphere}

The \emph{density matrix} of a pure state $\ket{\psi}$ is defined as $\rho = \ket{\psi}\!\bra{\psi}$. For the state $\ket{\psi}$ in (\ref{eq:Psi}) we have
\begin{equation}
  \rho
  = \frac{1}{2} \mx{
      1 + \cos{\theta} & e^{-i \varphi} \sin{\theta} \\
      e^{ i \varphi} \sin{\theta} & 1 - \cos{\theta}}
  = \frac{1}{2} \left( I + x \sigma_x + y \sigma_y + z \sigma_z \right),
  \label{eq:LongRho}
\end{equation}
where $(x,y,z)$ are the coordinates of the Bloch vector $\vc{r}$ given in (\ref{eq:UnitVector}) and
\begin{equation}
         I = \mx{1& 0\\0& 1}, \quad
  \sigma_x = \mx{0& 1\\1& 0}, \quad
  \sigma_y = \mx{0&-i\\i& 0}, \quad
  \sigma_z = \mx{1& 0\\0&-1}
  \label{eq:Pauli}
\end{equation}
are called \emph{Pauli matrices}. We can write (\ref{eq:LongRho}) more concisely as
\begin{equation}
  \rho = \frac{1}{2} \left( I + \vc{r} \cdot \vc{\sigma} \right)
  \label{eq:Rho}
\end{equation}
where $\vc{r} = (x, y, z)$ and $\vc{\sigma} = (\sigma_x, \sigma_y, \sigma_z)$.

If $\vc{r}_1$ and $\vc{r}_2$ are the Bloch vectors of two pure states $\ket{\psi_1}$ and $\ket{\psi_2}$, then
\begin{equation}
  \abs{\braket{\psi_1}{\psi_2}}^2
    = \tr (\rho_1 \rho_2)
    = \frac{1}{2} (1 + \vc{r}_1 \cdot \vc{r}_2).
  \label{eq:ScalarProduct}
\end{equation}
This relates the inner product in $\C^2$ to the one in $\R^3$. Since $\vc{r}_1$ and $\vc{r}_2$ are unit vectors, $\vc{r}_1 \cdot \vc{r}_2 = \cos \alpha$, where $\alpha$ is the angle between $\vc{r}_1$ and $\vc{r}_2$.

An \emph{orthogonal measurement} $M$ on a qubit can be specified by a set of two orthonormal states: $M = \set{\ket{\psi_0}, \ket{\psi_1}}$. Orthonormality means that \mbox{$\braket{\psi_i}{\psi_j} = \delta_{ij}$}. If we measure a qubit that is in state $\ket{\psi}$ with measurement $M$ then the \emph{outcome} will be either $0$ or $1$ and the state will ``collapse'' to $\ket{\psi_0}$ or $\ket{\psi_1}$ with probabilities $\abs{\braket{\psi_0}{\psi}}^2$ and $\abs{\braket{\psi_1}{\psi}}^2$, respectively. Observe that for orthogonal states equation (\ref{eq:ScalarProduct}) implies $\vc{r}_1 \cdot \vc{r}_2 = -1$, therefore they correspond to antipodal points on the Bloch sphere. If we denote the angle between the Bloch vectors of $\ket{\psi}$ and $\ket{\psi_0}$ by $\alpha$, then according to (\ref{eq:ScalarProduct}) the probabilities of the outcomes are
\begin{equation}
  \left\{
  \begin{aligned}
    p_0 &= \frac{1}{2}(1 + \cos{\alpha}), \\
    p_1 &= \frac{1}{2}(1 - \cos{\alpha}).
  \end{aligned}
  \right.
  \label{eq:Projections}
\end{equation}
There is a nice geometrical interpretation of these probabilities. If we project the Bloch vector corresponding to $\ket{\psi}$ on the axes spanned by the Bloch vectors of $\ket{\psi_0}$ and $\ket{\psi_1}$ (see Fig.~\ref{fig:BlochSphere}), then $p_0 = d_1 / 2$ and $p_1 = d_0 / 2$ (note the different indices), where $d_0$ is the distance between the projection and $\ket{\psi_0}$, but $d_1$ is the distance between the projection and $\ket{\psi_1}$. Observe that vectors on the upper hemisphere have greater probability to collapse to $\ket{\psi_0}$, but on lower hemisphere, to $\ket{\psi_1}$. On the equator both probabilities are equal to $\frac{1}{2}$.

\subsubsection{Unit disk representation} \label{sect:UnitDisk}

\image{0.9}{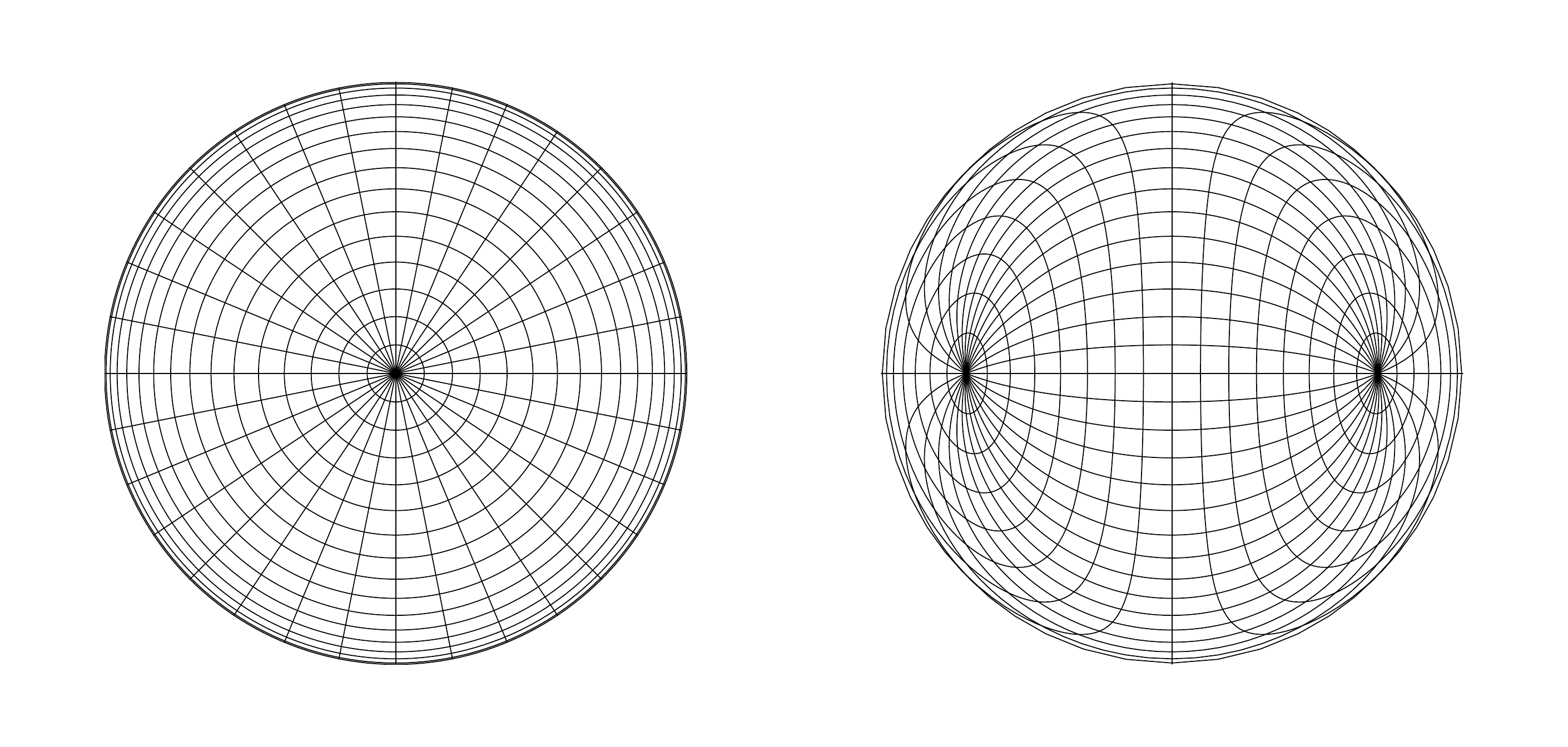}{Hadamard transformation in the unit disk representation}{Curves of constant $\theta$ and $\varphi$ before (on the left) and after the Hadamard transformation (on the right). Initially the curves of constant $\theta$ are concentric circles, but after the transformation they appear as deformed circles around both poles. The curves of constant $\varphi$ transform form radial rays to ``field lines'' connecting both poles. The image on the left appears to have only the North pole $\ket{0}$, since the Bloch sphere is punctured at the South pole $\ket{1}$ which must be identified with the boundary of the unit disk. The ``left pole'' and ``right pole'' in the image on the right correspond to the states $\ket{1}$ and $\ket{0}$, respectively.}{fig:Beta}

There is another way of visualizing a qubit. Unlike the Bloch sphere representation, this way of representing a qubit is not known to have appeared elsewhere. The idea is to use only one complex number to specify a pure qubit state $\ket{\psi} = \smx{\alpha \\ \beta} \in \C^2$. This is possible since $\ket{\psi}$ can be written in the form (\ref{eq:Psi}), which is completely determined by its second component
\begin{equation}
  \beta = e^{i\varphi}\sin{\frac{\theta}{2}}.
\end{equation}
The first component is just $\sqrt{1 - \abs{\beta}^2} = \alpha$. As $\abs{\beta} \leq 1$, the set of all possible qubit states can be identified with a unit disk in the complex plane (the polar coordinates assigned to $\ket{\psi}$ are $(r,\varphi)$, where $r = \sin{\frac{\theta}{2}}$). The origin $\beta = 0$ corresponds to $\ket{\psi}=\ket{0}$, and all points on the unit circle $\abs{\beta}=1$ are identified with $\ket{\psi}=\ket{1}$, since $e^{i \varphi} \ket{1}$ corresponds to the same quantum state for all $\varphi \in \R$.

The relation between the unit disk representation and the Bloch sphere representation can be visualized as follows:
\begin{itemize}
  \item the unit disk is obtained by puncturing the Bloch sphere at its South pole and flattening it,
  \item the Bloch sphere is obtained by gluing together the boundary of the unit disk.
\end{itemize}

It is much harder to visualize how a unitary transformation acts in the unit disk representation. Let us consider a simple example.
\begin{example}
Let us consider the action of the \emph{Hadamard gate} $H=\frac{1}{\sqrt{2}}\smx{1&1\\1&-1}$ in the unit disk representation. Note that $H^2 = I$ thus $H$ is an involution (self-inverse). It acts on the standard basis states as follows:
\begin{align}
  H \ket{0} &= \tfrac{1}{\sqrt{2}} \ket{0} + \tfrac{1}{\sqrt{2}} \ket{1} = \ket{+}, \label{eq:H0} \\
  H \ket{1} &= \tfrac{1}{\sqrt{2}} \ket{0} - \tfrac{1}{\sqrt{2}} \ket{1} = \ket{-}. \label{eq:H1}
\end{align}
The way $H$ transforms the curves of constant $\theta$ and $\varphi$ is shown in Fig.~\ref{fig:Beta}. From equation (\ref{eq:H0}) we see that the origin $\beta = 0$ corresponding to $\ket{0}$ is mapped to the ``right pole'' $\beta = \frac{1}{\sqrt{2}}$ corresponding to $\ket{+}$ (and vice versa). Recall that all points on the boundary of the unit disk in Fig.~\ref{fig:Beta} (on the left) are identified with $\ket{1}$. Thus equation (\ref{eq:H1}) tells us that the unit circle $\abs{\beta} = 1$ is mapped to the ``left pole'' $\beta = -\frac{1}{\sqrt{2}}$ in Fig.~\ref{fig:Beta} (on the right) corresponding to $\ket{-}$ (and vice versa). This means that $\ket{-}$ is mapped to the boundary of the unit disk in Fig.~\ref{fig:Beta} (on the right).
\end{example}

Since we use only one complex number $\beta$ to represent a quantum state, a finite set of quantum states $\set{\beta_1, \beta_2, \dotsc, \beta_n}$ can be represented by a polynomial
\begin{equation}
  c \, (\beta - \beta_1) (\beta - \beta_2) \dotsb (\beta - \beta_n)
\end{equation}
whose roots are $\beta_i$ (here $c \neq 0$ is arbitrary). We will use this representation in Sects.~\ref{sect:KnownQRACs} and \ref{sect:Numerical} to describe the qubit states whose Bloch vectors are the vertices of certain polyhedra. It is surprising that for those states the values of $c$ can be chosen so that the resulting polynomials have integer coefficients.

\subsection{Types of quantum encoding-decoding strategies} \label{sect:TypesOfQRACs}

Let us now consider the quantum analogue of a pure strategy.

\begin{definition}
A \emph{pure quantum \rac{n} encoding-decoding strategy} is an ordered tuple $(E,M_1,\dotsc,M_n)$ that consists of encoding function $E: \set{0,1}^n \mapsto \C^2$ and $n$ orthogonal measurements: $M_i = \set{\ket{\psi^i_0}, \ket{\psi^i_1}}$.
\end{definition}

If Alice encodes the string $x$ with function $E$, she obtains a pure qubit state $\ket{\psi} = E(x)$. When Bob receives $\ket{\psi}$ and is asked to recover the \mbox{$i$th} bit of $x$, he performs the measurement $M_i$. The probability that Bob recovers $x_i$ correctly is equal to
\begin{equation}
  p(x,i) = \abs{\braket{\psi^i_{x_i}\big}{\psi}}^2.
  \label{eq:pxi}
\end{equation}

As in the classical setting, we can allow Alice and Bob to have probabilistic quantum strategies without cooperation. Though we will not need it, mixed quantum strategies can be defined in complete analogy with mixed classical strategies.

\begin{definition}
A \emph{mixed quantum \rac{n} encoding-decoding strategy} is an ordered tuple $(P_E,P_{M_1},\dotsc,P_{M_n})$ of probability distributions. $P_E$ is a distribution over encoding functions $E$ and $P_{M_i}$ are probability distributions over orthogonal measurements of qubit.
\end{definition}

The main objects of our research are quantum strategies with cooperation, i.e., with shared randomness. They are defined in complete analogy with the classical ones.

\begin{definition}
A \emph{quantum \rac{n} encoding-decoding strategy with shared randomness} is a probability distribution over pure quantum strategies.
\end{definition}

We would like to point out two very important things about quantum strategies with shared randomness. The first thing is that all statements about classical strategies with SR in Sect.~\ref{sect:Yao} are valid for quantum strategies as well (the only difference is that \emph{``pure strategy''} now means \emph{``pure quantum strategy''} instead of \emph{``pure classical strategy''} and \emph{``strategy with SR''} means \emph{``quantum strategy with SR''} instead of \emph{``classical strategy with SR''}). The most important consequence of this observation is that Theorem~\ref{thm:SRtoPure} is valid also for quantum strategies with SR. This means that the same technique of obtaining the upper bound can be used in the quantum case, i.e., we can consider the average success probability of a pure quantum strategy instead of the worst case success probability of the quantum strategy with SR.

The second thing is that the quantum strategy with SR is the most powerful quantum encoding-decoding strategy when both kinds of classical randomness (local and shared) is allowed. However, it is not the most general strategy, since it cannot be used to simulate certain classical strategies, e.g., the ones with fixed output. However, it turns out that the ability to simulate such strategies does not give any advantage (see Sect.~\ref{sect:GeneralUpperBound} and Appendix~\ref{app:POVMs}).

\subsection{Known quantum RACs} \label{sect:KnownQRACs}

In \cite{DenseCoding1,DenseCoding2} it has been shown that for \rac{2} classical RACs in the mixed setting the decoding party cannot do better than guessing, i.e., the worst case success probability cannot exceed $1/2$. However, if quantum states can be transmitted, there are pure quantum \rac{2} and \rac{3} schemes \cite{DenseCoding1,DenseCoding2}. This clearly indicates the advantages of quantum RACs. On the other hand, a quantum \rac{4} scheme cannot exist \cite{No41}. We will review these results in the next three sections.

\doubleimage
{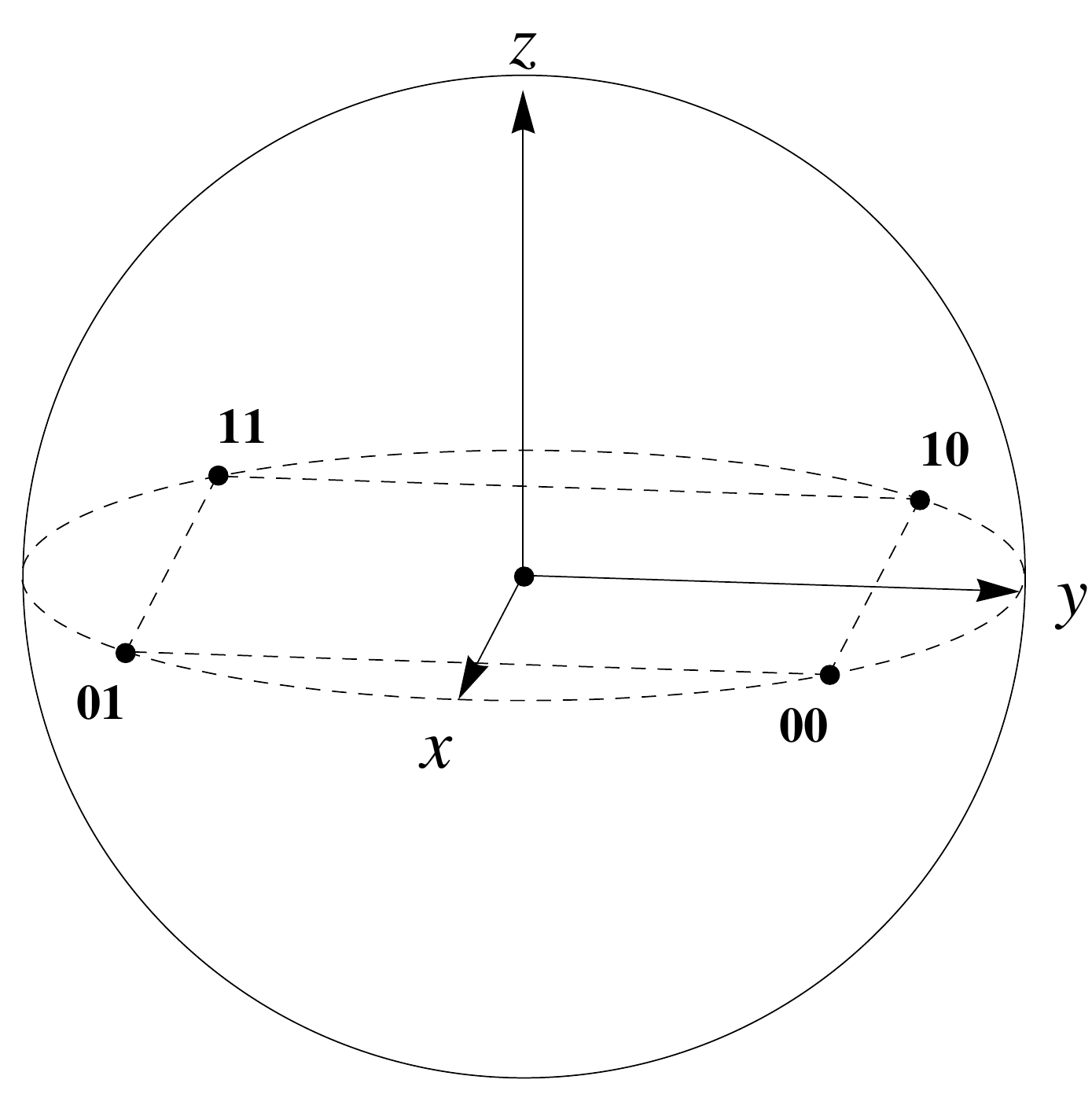}{Bloch sphere representation of \rac{2} QRAC}{Bloch sphere representation of encoding for \rac{2} quantum random access code.}{fig:KnownQRAC2}
{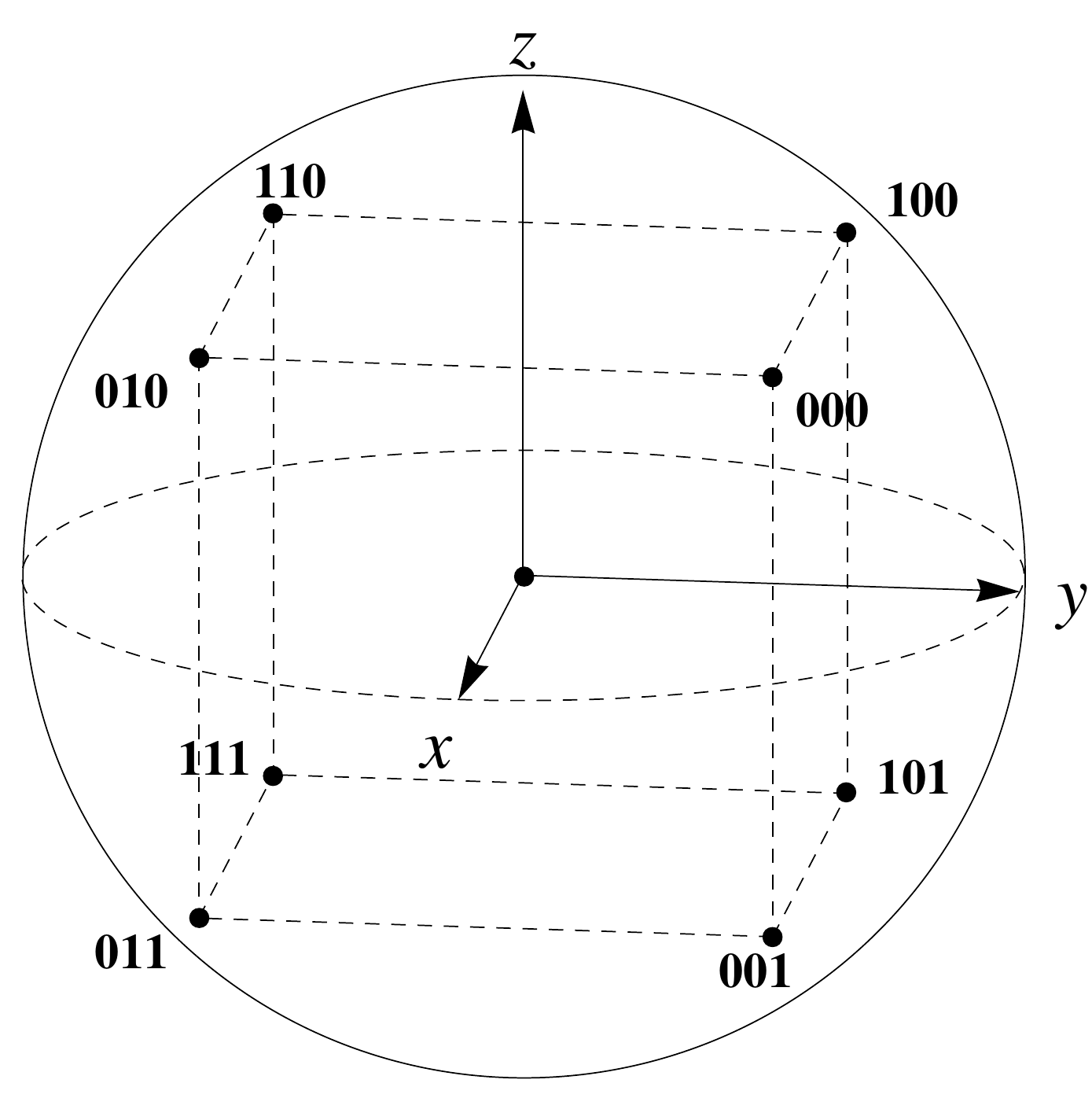}{Bloch sphere representation of \rac{3} QRAC}{Bloch sphere representation of encoding for \rac{3} quantum random access code.}{fig:KnownQRAC3}

\subsubsection{The \rac{2} QRAC} \label{sect:KnownQRAC2}

The \rac{2} QRAC is described in \cite{DenseCoding1,DenseCoding2,No41}. The main idea is to use two mutually orthogonal pairs of antipodal Bloch vectors for measurement bases. For example, let $M_1$ and $M_2$ be the measurements along the $x$ and $y$ axes, respectively. The corresponding Bloch vectors are $\vc{v}_1=(\pm 1,0,0)$ and $\vc{v}_2=(0,\pm 1,0)$. The measurement bases are
\begin{align}
  M_1&=\set{\frac{1}{\sqrt{2}}\mx{1\\1},\frac{1}{\sqrt{2}}\mx{1\\-1}}, \label{eq:M1} \\
  M_2&=\set{\frac{1}{\sqrt{2}}\mx{1\\i},\frac{1}{\sqrt{2}}\mx{1\\-i}}. \label{eq:M2}
\end{align}
The planes orthogonal to the $x$ and $y$ axes cut the Bloch sphere into four parts. Note that in each part only one definite string can be encoded (otherwise the worst case success probability will be less than $\frac{1}{2}$). According to (\ref{eq:Projections}), all encoding points must be as far from both planes as possible in order to maximize the worst case success probability (recall the geometrical interpretation of the measurement shown in Fig.~\ref{fig:BlochSphere}). In our case the best encoding states are the vertices of a square $\frac{1}{\sqrt{2}}(\pm 1,\pm 1,0)$ inscribed in the unit circle on the $xy$ plane (see Fig.~\ref{fig:KnownQRAC2}). Given a string $x = x_1 x_2$, the Bloch vector of the encoding state can be found as follows:
\begin{equation}
  \vc{r}(x) = \frac{1}{\sqrt{2}}
  \mx{
    (-1)^{x_1} \\
    (-1)^{x_2} \\
    0
  }.
\end{equation}
The corresponding encoding function is
\begin{equation}
  E(x_1,x_2) = \frac{1}{\sqrt{2}} \ket{0} + \frac{(-1)^{x_1}+i(-1)^{x_2}}{2} \ket{1}.
  \label{eq:QRAC2Encoding}
\end{equation}
The success probability is the same for all input strings and all bits to be recovered:
\begin{equation}
  p = \frac{1}{2}\left(1+\cos{\frac{\pi}{4}}\right) = \frac{1}{2}+\frac{1}{2\sqrt{2}} \approx \p{2}.
  \label{eq:p2}
\end{equation}

\subsubsection{The \rac{3} QRAC} \label{sect:KnownQRAC3}

It is not hard to generalize the \rac{2} QRAC to a \rac{3} code---just take three mutually orthogonal pairs of antipodal Bloch vectors, i.e., the vertices of an \emph{octahedron} \cite{No41,Severini}. The third pair is $\vc{v}_3=(0,0,\pm 1)$ and the corresponding measurement basis is
\begin{equation}
  M_3=\set{\mx{1\\0},\mx{0\\1}}.
  \label{eq:M3}
\end{equation}
In this case we have three orthogonal planes that cut the sphere into eight parts and only one string can be encoded into each part. In this case the optimal encoding states correspond to the vertices of a cube $\frac{1}{\sqrt{3}}(\pm 1,\pm 1,\pm 1)$ inscribed in the Bloch sphere (see Fig.~\ref{fig:KnownQRAC3}). The Bloch vector of the encoding state of string $x = x_1 x_2 x_3$ is
\begin{equation}
  \vc{r}(x) = \frac{1}{\sqrt{3}}
  \mx{
    (-1)^{x_1} \\
    (-1)^{x_2} \\
    (-1)^{x_3}
  }.
\end{equation}
The corresponding encoding function is $E(x_1,x_2,x_3)=\alpha\ket{0}+\beta\ket{1}$ with coefficients $\alpha$ and $\beta$ explicitly given by
\begin{equation}
  \left\{
  \begin{aligned}
    \alpha &= \sqrt{\frac{1}{2}+\frac{(-1)^{x_3}}{2\sqrt{3}}}, \\
    \beta  &= \frac{(-1)^{x_1}+i(-1)^{x_2}}{\sqrt{6+2\sqrt{3}(-1)^{x_3}}}.
  \end{aligned}
  \right.
\end{equation}
In fact, the coefficients $\beta$ are exactly the eight roots of the polynomial\footnote{The unit disk representation of a quantum state and the representation of a finite set of quantum states using a polynomial was discussed in Sect.~\ref{sect:UnitDisk}.}
\begin{equation}
  36 \beta^8 + 24 \beta^4 + 1
\end{equation}
This code also has the same success probability in all cases:
\begin{equation}
  p = \frac{1}{2}+\frac{1}{2\sqrt{3}} \approx \p{3}.
  \label{eq:p3}
\end{equation}

\subsubsection{Impossibility of the \rac{4} QRAC} \label{sect:noQRAC4}

Hayashi et al. \cite{No41} have shown that \rac{2} and \rac{3} codes discussed above cannot be generalized for $4$ (and hence more) encoded bits. The reason is simple---it is not possible to cut the Bloch sphere into $16$ parts with $4$ great circles (see the proof below). Thus the number of strings will exceed the number of parts, hence at least two strings must be encoded in the same part. This makes the worst case success probability drop below $\frac{1}{2}$.

\image{0.6}{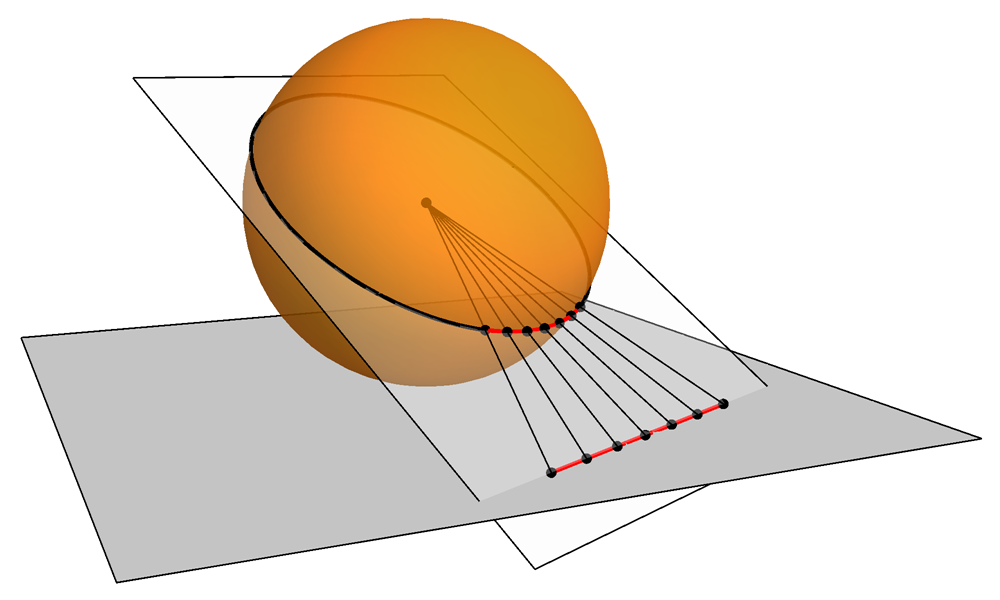}{Gnomonic projection}{Gnomonic projection transforms great circles to lines and vice versa.}{fig:Gnomonic}

Let us consider how many parts can be obtained by cutting a sphere with $4$ great circles. Without loss of generality we can assume that the first great circle coincides with the equator. We use the \emph{gnomonic projection} (from the center of the sphere) to project the remaining three circles to a plane tangent to the South pole. Note that great circles are transformed into lines and vice versa (see Fig.~\ref{fig:Gnomonic}), thus we will obtain three lines. Also note that each region in the plane corresponds to two (diametrically opposite) regions on the sphere. It is simple to verify that three lines cannot cut the plane into more than $7$ parts (see Fig.~\ref{fig:CuttingPlane}). Thus the sphere cannot be cut into more than $14$ parts with four great circles.\footnote{In general, if we have $n$ great circles on the sphere, the maximal number of parts we can obtain is twice what we can obtain by cutting the plane with $n-1$ lines. If each line we draw intersects all previous lines and no three lines intersect at the same point, the sphere is cut into $n(n-1)+2$ parts after the inverse gnomonic projection.} An example achieving the upper bound is shown in see Fig.~\ref{fig:CuttingSphere} (see also Figs.~\ref{fig:QRAC4[sym]} and \ref{fig:Tetrahedron}). Using essentially the same argument for generalized Bloch vectors Hayashi et al. \cite{No41} show that \racm{2^{2m}}{m} QRACs with $p>1/2$ do not exist for all $m \geq 1$. The generalized Bloch vector will be briefly introduced in Sect.~\ref{sect:Generalizations}.

\doubleimage
{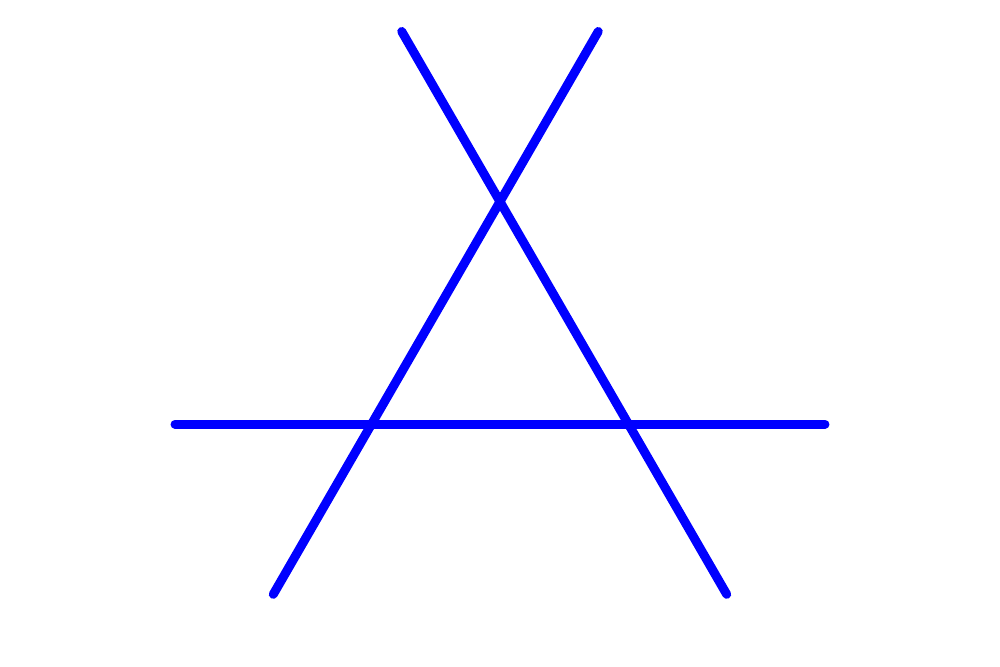}{Cutting the plane with $3$ lines into $7$ parts}{Cutting the plane with $3$ lines into $7$ parts.}{fig:CuttingPlane}
{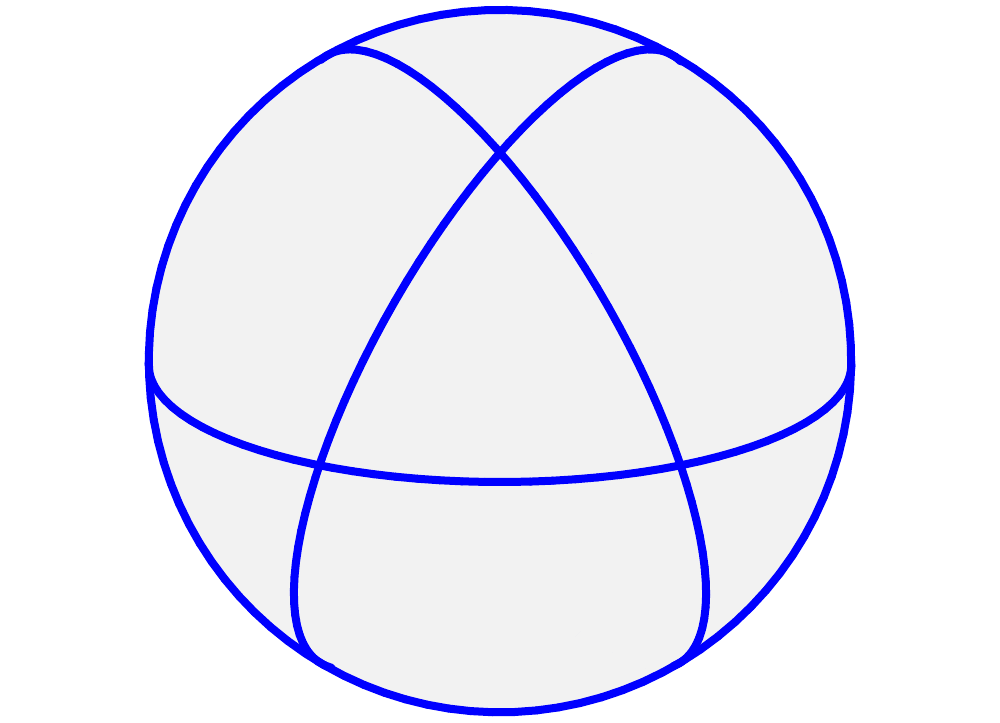}{Cutting the sphere with $4$ great circles into $14$ parts}{Cutting the sphere with $4$ great circles into $14$ parts (seven diametrically opposite parts are equal).}{fig:CuttingSphere}

\subsection{Optimal encoding for given decoding strategy} \label{sect:OptimalEncoding}

We just reviewed the known results on pure \rac{n} quantum random access codes. From now on we will consider QRACs with shared randomness. In this section we will show how to find the optimal encoding strategy for a given decoding strategy. More precisely, we will show that the measurement directions of a QRAC with SR determine the corresponding optimal encoding states in a simple way.

An orthogonal measurement for the \mbox{$i$th} bit is specified by antipodal points on the Bloch sphere: $M_i = \set{\vc{v}_i, -\vc{v}_i}$. Let $\vc{r}_x$ be the Bloch vector that corresponds to the quantum state in which string $x \in \set{0,1}^n$ is encoded. According to equations in (\ref{eq:Projections}) the success probability for input $(x,i)$ is
\begin{equation}
  p(x,i) = \frac{1}{2} \bigl(1 + (-1)^{x_i} \vc{v}_i \cdot \vc{r}_x \bigr)
\end{equation}
and the average success probability is given by
\begin{equation}
  \begin{split}
    p &= \frac{1}{2^n \cdot n} \sum_{x \in \set{0,1}^n} \sum_{i=1}^n
         \frac{1}{2} \bigl(1 + (-1)^{x_i} \vc{v}_i \cdot \vc{r}_x \bigr) \\
      &= \frac{1}{2} \Biggl(1 + \frac{1}{2^n \cdot n}
         \underbrace{\sum_{x \in \set{0,1}^n} \vc{r}_x \cdot \sum_{i=1}^n (-1)^{x_i} \vc{v}_i}_{\Svr} \Biggr).
  \end{split}
  \label{eq:AverageProbability}
\end{equation}
In order to maximize the probability $p$, we only need to maximize $\Svr$ in equation (\ref{eq:AverageProbability}) over all possible measurements $\vc{v}_i$ and encodings $\vc{r}_x$ (in total $n + 2^n$ unit vectors in $\R^3$). We will denote the maximum of $\Svr$ by $S(n)$:
\begin{equation}
  S(n) \spacer
    = \max_{\vis, \rxs} \Svr \spacer
    = \spacer \max_\vis \sum_{x \in \set{0,1}^n}
      \max_{\vc{r}_x} \spacer \vc{r}_x \cdot \sum_{i=1}^n (-1)^{x_i} \vc{v}_i.
  \label{eq:rxvx}
\end{equation}
If we define
\begin{equation}
  \vc{v}_x = \sum_{i=1}^n (-1)^{x_i} \vc{v}_i,
  \label{eq:vx}
\end{equation}
then it is obvious that the scalar product $\vc{r}_x \cdot \vc{v}_x$ in (\ref{eq:rxvx}) will be maximized when $\vc{r}_x$ is chosen along the same direction as $\vc{v}_x$, i.e. $\vc{r}_x = \vc{v}_x / \norm{\vc{v}_x}$ when $\norm{\vc{v}_x} \neq 0$. In this case we have $\vc{r}_x \cdot \vc{v}_x = \norm{\vc{v}_x}$ and
\begin{equation}
  S(n) = \max_\vis \sum_{x \in \set{0,1}^n} \norm{\sum_{i=1}^n (-1)^{x_i} \vc{v}_i}.
  \label{eq:Sn}
\end{equation}
Therefore we only need to maximize over all possible measurements succinctly represented by $n$ unit vectors $\vc{v}_i \in \R^3$, because the optimal encoding is already determined by measurements (see Sect.~\ref{sect:Numerical} for some numerical results obtained in this way). When the value of $S(n)$ is found, then according to (\ref{eq:AverageProbability}) the corresponding probability is
\begin{equation}
  p(n) = \frac{1}{2} \left(1 + \frac{S(n)}{2^n \cdot n}\right).
  \label{eq:ProbabilitySn}
\end{equation}

We can observe a connection between quantum and classical RACs with SR. Assume that Marge and Homer\footnote{In this scenario it is more convenient to replace Alice and Bob with Marge and Homer from \textit{The Simpsons}.} have to implement \rac{n} QRAC with SR and are deciding what strategies to use---Homer is responsible for choosing the measurements, but Marge has to choose how to encode the input string. Once they have decided, they have to follow the agreement and cannot cheat. Unfortunately, Homer is foolish and he proposes to measure all bits in the same basis. Luckily Marge is clever enough to choose the optimal encoding for Homer's measurements. According to the discussion above, she has to use the majority encoding function. Thus the obtained QRAC is as good as an optimal classical RAC discussed in Sect.~\ref{sect:OptimalClassical}, Theorem~\ref{thm:OptimalClassical}.

It looks plausible that using the same measurement for all bits is the worst decoding strategy. However, we have not proved this, so we leave it as a conjecture:

\begin{conjecture}
For any choice of measurements there is an encoding such that the resulting \rac{n} quantum RAC with SR is at least as good as the optimal \rac{n} classical one.
\end{conjecture}

\subsection{Relation to a random walk in \texorpdfstring{$\R^3$}{R3}} \label{sect:RandomWalk}

QRACs with shared randomness are related to random walks in $\R^3$. This relation can be seen by suitably interpreting equations (\ref{eq:Sn}) and (\ref{eq:ProbabilitySn}). Let us consider an \rac{n} QRAC with SR whose measurement directions are given by unit vectors $\set{\vc{v}_i}$ and let us assume that the corresponding optimal encoding for these measurements is used as described in the previous section. Then we can write the success probability $p(\vc{v}_1, \dotsc, \vc{v}_n)$ of this QRAC in the following suggestive form:
\begin{equation}
  p(\vc{v}_1, \dotsc, \vc{v}_n) = \frac{1}{2} \Bigl( 1 + \frac{1}{n} \; d(\vc{v}_1, \dotsc, \vc{v}_n) \Bigr),
  \label{eq:pv}
\end{equation}
where
\begin{equation}
  d(\vc{v}_1, \dotsc, \vc{v}_n) = \frac{1}{2^n} \sum_{a \in \set{1,-1}^n} \norm{\sum_{i=1}^n a_i \vc{v}_i}
  \label{eq:dv}
\end{equation}
is the average distance traveled by a random walk whose \mbox{$i$th} step is $\vc{v}_i$ or $-\vc{v}_i$, each with probability $1/2$. For example, $\vc{v}_1 = \vc{v}_2 = \dotsb = \vc{v}_n$ corresponds to a random walk on a line and $d(\vc{v}_1, \dotsc, \vc{v}_n)$ is the average distance traveled after $n$ steps of this walk. Recall from the previous section that this choice of $\set{\vc{v}_i}$ corresponds to the optimal classical RAC and we conjecture that this is the worst possible choice. Similarly, if we choose roughly one third of vectors $\set{\vc{v}_i}$ along each coordinate axis, we obtain a random walk in a cubic lattice and $d(\vc{v}_1, \dotsc, \vc{v}_n)$ is the average distance traveled when roughly $n/3$ steps are performed along each coordinate axis (see Sect.~\ref{sect:Lower2}).

In Sects.~\ref{sect:LowerBounds} we will use this relation between random access codes and random walks to prove a lower bound for the success probability of \rac{n} QRACs with SR.

\subsection{Upper bound} \label{sect:UpperBound}

In this section we will derive an upper bound for $S(n)$. For this purpose we rewrite the equation (\ref{eq:Sn}) in the following form:
\begin{equation}
  S(n) = \max_\vis \Sv
  \label{eq:MaxSv}
\end{equation}
where
\begin{equation}
  \Sv \spacer = \sum_{a \in \set{1,-1}^n} \norm{\sum_{i=1}^n a_i \vc{v}_i}
  \label{eq:Sv}
\end{equation}
(for convenience we take the sum over the set $\set{1,-1}^n$ instead of $\set{0,1}^n$).

\begin{lemma}
For any unit vectors $\vc{v}_1,\dotsc,\vc{v}_n$ we have
\begin{equation}
  \sum_{a_1, \dotsc, a_n \in \set{1,-1}} \norm{a_1 \vc{v}_1 + \dotsb + a_n \vc{v}_n}^2 = n \cdot 2^n.
  \label{eq:NormSquareSum}
\end{equation}
\label{lm:NormSquares}
\end{lemma}

\begin{proof}
For $n = 1$ we have
\begin{equation}
  \sum_{a_1 \in \set{1,-1}} \norm{a_1 \vc{v}_1}^2 = \norm{\vc{v}_1}^2 + \norm{-\vc{v}_1}^2 = 2.
\end{equation}
Let us assume that equation (\ref{eq:NormSquareSum}) holds for $n = k$. Then for $n = k + 1$ we have
\begin{equation}
    \sum_{a_1, \dotsc, a_{k}, a_{k+1} \in \set{1,-1}}
      \norm{a_1 \vc{v}_1 + \dotsb + a_k \vc{v}_k + a_{k+1} \vc{v}_{k+1}}^2.
\end{equation}
If we write out the sum over $a_{k+1}$ explicitly, we obtain
\begin{equation}
  \sum_{a_1, \dotsc, a_{k} \in \set{1,-1}}
    \left(
      \norm{a_1 \vc{v}_1 + \dotsb + a_k \vc{v}_k + \vc{v}_{k+1}}^2 +
      \norm{a_1 \vc{v}_1 + \dotsb + a_k \vc{v}_k - \vc{v}_{k+1}}^2
    \right).
\end{equation}
We can use the parallelogram identity
\begin{equation}
  \norm{\vc{u}_1 + \vc{u}_2}^2 + \norm{\vc{u}_1 - \vc{u}_2}^2 = 2 \left(\norm{\vc{u}_1}^2 + \norm{\vc{u}_2}^2\right),
\end{equation}
which holds for any two vectors $\vc{u}_1$ and $\vc{u}_2$, to simplify the sum as follows:
\begin{equation}
  \sum_{a_1, \dotsc, a_{k} \in \set{1,-1}}
    2 \left( \norm{a_1 \vc{v}_1 + \dotsb + a_k \vc{v}_k}^2 + \norm{\vc{v}_{k+1}}^2 \right).
  \label{eq:AlmostDone}
\end{equation}
We know that $\vc{v}_{k+1}$ is a unit vector and we have assumed that (\ref{eq:NormSquareSum}) holds for $n = k$; therefore (\ref{eq:AlmostDone}) simplifies to $2 \left( k \cdot 2^k + 2^k \right) = (k + 1) \cdot 2^{k + 1}$.
\end{proof}

We will use the previous lemma to obtain an upper bound for $\Sv^2$ defined in (\ref{eq:Sv}). According to (\ref{eq:MaxSv}) this will give us an upper bound for $S(n)$ as well.

\begin{lemma}
For any set of unit vectors $\set{\vc{v}_i}_{i=1}^n$, the inequality $\Sv \leq \sqrt{n} \cdot 2^n$ holds.
\label{lm:Sv}
\end{lemma}

\begin{proof}
We can interpret the first sum in equation (\ref{eq:Sv}) as an inner product with $(1, \dotsc, 1) \in \R^{2^n}$. Then the Cauchy-Schwarz inequality $\vc{x} \cdot \vc{y} \leq \norm{\vc{x}} \norm{\vc{y}}$ says that
\begin{equation}
  \Sv \leq \sqrt{2^n} \; \sqrt{\sum_{a \in \set{1,-1}^n} \norm{\sum_{i=1}^n a_i \vc{v}_i}^2}
         = \sqrt{2^n} \; \sqrt{n \cdot 2^n} = \sqrt{n} \cdot 2^n,
\end{equation}
where Lemma~\ref{lm:NormSquares} was used to obtain the first equality.
\end{proof}

\begin{theorem}
For any \racp{n} QRAC with shared randomness, $p \leq \frac{1}{2} + \frac{1}{2 \sqrt{n}}$.
\label{thm:UpperBound}
\end{theorem}

\begin{proof}
From Lemma~\ref{lm:Sv} we have $\Sv \leq \sqrt{n} \cdot 2^n$. From equation (\ref{eq:MaxSv}) we see that the same upper bound applies to $S(n)$. Putting this into (\ref{eq:ProbabilitySn}) we get
\begin{equation*}
  p \leq \frac{1}{2} + \frac{1}{2\sqrt{n}}. \qedhere
\end{equation*}
\end{proof}

In particular, this means that the known \rac{2} and \rac{3} QRACs discussed in Sect.~\ref{sect:KnownQRACs} cannot be improved even if shared randomness is allowed.

The intuition behind this upper bound is as follows. If instead of $\R^3$ the Bloch vector of a qubit state would be in $\R^n$, we could choose all $n$ measurements to be mutually orthogonal. For example, we could choose the vectors forming measurement bases to be the vertices of the \emph{cross polytope}, i.e., all permutations of $(\pm 1, 0, \dotsc, 0)$. The optimal encoding corresponding to this choice are the vertices of the \emph{hypercube}, i.e., points $(\pm 1, \pm 1, \dotsc, \pm 1)$, thus all terms in equation (\ref{eq:Sv}) are equal to $\sqrt{n}$ and sum to $2^n \sqrt{n}$, so the probability (\ref{eq:ProbabilitySn}) is $\frac{1}{2} (1 + \frac{2^n \sqrt{n}}{2^n n}) = \frac{1}{2} + \frac{1}{2\sqrt{n}}$. Since we have only three dimensions, the actual probability should not be larger.

\subsection{General upper bound} \label{sect:GeneralUpperBound}

Let us prove an analogue of Theorem~\ref{thm:UpperBound} for a more general model, because quantum mechanics allows us to consider more general quantum states and measurements. Namely, Alice can encode her message into a \emph{mixed state} instead of a pure state and Bob can use a POVM measurement instead of an orthogonal measurement to recover information. A mixed state is just a probability distribution over pure states, so it does not provide a more general encoding model. In contrast, a POVM measurement provides a more general decoding model. In fact, there is another reason to extend the model.

\begin{example}
It is not possible to construct a pure QRAC (as defined in Sect.~\ref{sect:TypesOfQRACs}) that simulates the following pure classical \rac{2} RAC:
\begin{itemize}
  \item \emph{encoding:} encode the first bit,
  \item \emph{decoding:} if the first bit is requested, output the received bit; if the second one is requested---output $0$ no matter what is received.
\end{itemize}
To recover the first bit with certainty, Alice and Bob have to agree on two antipodal points on the Bloch sphere, where the information is encoded. Unfortunately the second bit will cause a problem---it is not possible to choose an orthogonal measurement of a qubit in an unknown state, so that the result is always the same.
\end{example}

This example suggests that the model of pure quantum encoding-decoding strategies introduced in Sect.~\ref{sect:TypesOfQRACs} should be extended in one way or the other. It is obvious that a \emph{constant decoding function} ($0$ or $1$) can be implemented using a single-outcome POVM measurement. However, it turns out that in the qubit case a two-outcome POVM measurement can be replaced by a probability distribution over orthogonal measurements and constant decoding functions (see appendix~\ref{app:POVMs}). This means that both extensions are equivalent. For simplicity we choose to extend the model by allowing constant decoding functions, thus Bob can either perform an orthogonal measurement or use a constant decoding function. The goal of this section is to show that constant decoding functions do not give any advantage.

\begin{definition}
An \emph{enhanced orthogonal measurement} is either an orthogonal measurement or one that always gives the same answer.
\end{definition}

\begin{definition}
An \emph{enhanced pure quantum \rac{n} encoding-decoding strategy} is given by an ordered tuple $(E,M_1,\dotsc,M_n)$ consisting of encoding function $E: \set{0,1}^n \mapsto \C^2$ and $n$ decoding functions $M_i$ that are enhanced orthogonal measurements.
\end{definition}

\begin{definition}
An \emph{enhanced quantum encoding-decoding strategy with SR} is a probability distribution over enhanced pure quantum strategies.
\end{definition}

Now it is straightforward to construct a pure quantum RAC for the previous example. In fact, now any classical RAC (either pure, mixed or with SR) can be simulated by the corresponding type of a quantum RAC.

There is no need to further extend the model of enhanced QRACs with SR by adding other types of classical randomness. For example, a probabilistic combination of POVMs does not provide a more general measurement, because it can be simulated by a probabilistic combination of enhanced orthogonal measurements. The same holds for probabilistic \mbox{post-processing} of the measurement results (which can be simulated by a probabilistic combination of enhanced orthogonal measurements as shown in Appendix~\ref{app:POVMs}). Therefore enhanced QRACs with SR constitute the most general model when any kind of classical randomness is allowed.

One might suspect that the upper bound obtained in Theorem~\ref{thm:UpperBound} does not hold for this model, but this is not the case.

\begin{theorem}
For any \racp{n} enhanced QRAC with SR, $p \leq \frac{1}{2} + \frac{1}{2 \sqrt{n}}$.
\label{thm:GeneralUpperBound}
\end{theorem}

\begin{proof}
According to Yao's principle and Theorem~\ref{thm:SRtoPure}, we can consider the average success probability of pure enhanced QRACs instead. It suffices to rule out the constant decoding functions. More precisely, we have to show that QRACs having a constant decoding function for some bit give a smaller upper bound than those without it. In fact, we are proving a quantum analogue of Lemma~\ref{lm:NoConst} from Sect.~\ref{sect:OptimalClassical}.

We will use induction on $n$. The case $n=1$ is trivial---a pure enhanced QRAC with a constant decoding function has average success probability $\frac{1}{2} < 1$. Let us assume that for some $n = k - 1 \geq 1$ the constant decoding functions do not give any benefit. We now prove that the same holds for $n = k$. Let us assume that the constant decoding function $0$ is used for the \mbox{$k$th} bit. The average case success probability is
\begin{equation}
  p(k) = \frac{1}{2^k \cdot k} \sum_{x \in \set{0,1}^k}
         \left( \sum_{i=1}^{k-1} p(x,i) + \delta_{0,x_k} \right),
\end{equation}
where $p(x,i)$ is the success probability (\ref{eq:pxi}) for the input $(x,i)$ where $i \leq k-1$ and $\delta_{0,x_k}$ is the probability that the decoding function $0$ gives a correct answer for the \mbox{$k$th} bit. The last bit can be ignored during the encoding and decoding of other bits:
\begin{align}
  p(k) &= \left( \frac{1}{2^k \cdot k} \sum_{\spacer x \in \set{0,1}^{k-1}} 2 \sum_{i=1}^{k-1} p(x,i) \right)
          + \frac{1}{2k} \\
       &= \frac{k-1}{k} \left( \frac{1}{2^{k-1} \cdot (k-1)} \sum_{\spacer x \in \set{0,1}^{k-1}}
          \sum_{i=1}^{k-1} p(x,i) \right) + \frac{1}{2k}. \label{eq:pk}
\end{align}
Note that the bracketed expression in (\ref{eq:pk}) is the success probability $p(k-1)$ of a shorter QRAC. Therefore
\begin{equation}
  p(k) = \frac{k-1}{k} \cdot p(k-1) + \frac{1}{2k}.
\end{equation}
Now we can apply the inductive hypothesis:
\begin{equation}
  p(k) \leq \frac{k-1}{k} \left( \frac{1}{2} + \frac{1}{2\sqrt{k-1}} \right) + \frac{1}{2k}
          = \frac{1}{2} + \frac{\sqrt{k-1}}{2k}
          < \frac{1}{2} + \frac{1}{2\sqrt{k}},
  \label{eq:pk2}
\end{equation}
completing the proof. Thus the upper bound obtained in Theorem~\ref{thm:UpperBound} holds for the general model as well.
\end{proof}

Observe again that for $n=2$ and $n=3$ this upper bound matches equations (\ref{eq:p2}) and (\ref{eq:p3}), respectively. This means that the known \rac{2} and \rac{3} quantum random access codes with pure encoding-decoding strategies (see Sects.~\ref{sect:KnownQRAC2} and \ref{sect:KnownQRAC3}, respectively) are optimal even among enhanced strategies with SR. For $n=4$ we get $p\leq\frac{3}{4}$.

A similar upper bound was recently obtained by \mbox{Ben-Aroya} et al. \cite{Hypercontractive} for \racm{n}{m} QRACs, where $k$ bits must be recovered. They allow randomized strategies without shared randomness. In particular, they show that for any $\eta > 2 \ln 2$ there exists a constant $C_\eta$ such that for $n \gg k$
\begin{equation}
  p \leq C_\eta \left( \frac{1}{2} + \frac{1}{2} \sqrt{\frac{\eta m}{n}} \right)^k.
  \label{eq:HypercontractiveUpperBound}
\end{equation}
It might be possible to generalize our upper bound (\ref{eq:pk2}) to obtain something similar to (\ref{eq:HypercontractiveUpperBound}).

\subsection{Lower bounds} \label{sect:LowerBounds}

In the next two sections we will describe two constructions of \rac{n} QRAC with SR for all $n \geq 1$. These constructions provide a lower bound on the success probability. They use random and orthogonal measurements, respectively. In the first case it is hard to compute the exact success probability even for small values of $n$, but we will obtain an asymptotic expression. However, in the second case we do not know the asymptotic success probability, but can easily compute the exact success probability for small $n$.

\subsubsection{Lower bound by random measurements} \label{sect:Lower1}

We now turn to lower bound for $p$. A lower bound for QRACs with shared randomness can be obtained by randomized encoding. Alice and Bob can use the shared random string to agree on some random orthogonal measurement for each bit. Each of these measurement bases can be specified by antipodal points on the Bloch sphere (see Sect.~\ref{sect:BlochSphere}). These points can be sampled by using some sphere point picking method \cite{SpherePointPicking}, near uniformly given enough shared randomness. The chosen measurements determine the optimal encoding scheme (see Sect.~\ref{sect:OptimalEncoding}) which is known to both sides.

The expected success probability of randomized \rac{n} QRAC similarly to (\ref{eq:pv}) is given by
\begin{equation}
  \E(p) = \frac{1}{2} \left( 1 + \frac{1}{n} \E_\vis d(\vc{v}_1, \cdots, \vc{v}_n) \right)
  \label{eq:Eofp}
\end{equation}
where according to equation (\ref{eq:dv})
\begin{align}
  \E_\vis d(\vc{v}_1, \cdots, \vc{v}_n)
  &= \E_\vis \left( \frac{1}{2^n} \sum_{a \in \set{1,-1}^n} \norm{\sum_{i=1}^n a_i \vc{v}_i} \right) \\
  &= \frac{1}{2^n} \sum_{a \in \set{1,-1}^n} \E_\vis \norm{\sum_{i=1}^n a_i \vc{v}_i}. \label{eq:Eofav}
\end{align}
Each $a \in \set{1,-1}^n$ influences the direction of some vectors $\vc{v}_i$, but the resulting set $\set{a_i \vc{v}_i}$ is still uniformly distributed. Therefore the expected value in equation (\ref{eq:Eofav}) does not depend on $a$ and we have
\begin{equation}
  \E_\vis d(\vc{v}_1, \cdots, \vc{v}_n) = \E_\vis \norm{\sum_{i=1}^n \vc{v}_i}.
  \label{eq:Eofd}
\end{equation}
This expression has a very nice geometrical interpretation---it is the average distance traveled by a particle that performs $n$ steps of unit length each in a random direction. This distance can be found by evaluating the following integral:
\begin{equation}
  \frac{1}{(4\pi)^n}
  \int_{\theta_1=0}^{\pi}
  \int_{\varphi_1=0}^{2\pi}
  \dotsi
  \int_{\theta_n=0}^{\pi}
  \int_{\varphi_n=0}^{2\pi}
  \norm{
    \sum_{i=1}^n
    \mx{ \sin \theta_i \cos \varphi_i \\
         \sin \theta_i \sin \varphi_i \\
         \cos \theta_i }
  }
  \prod_{i=1}^n \sin \theta_i \, d \theta_i \, d \varphi_i.
\end{equation}
Unfortunately it is very hard to evaluate it even numerically, since the integrand is highly oscillatory. An alternative approach is to directly simulate a random walk by sampling points uniformly from the sphere \cite{SpherePointPicking}. For small values of $n$ the success probability (\ref{eq:Eofp}) averaged over $10^6$ simulations is given in Table~\ref{tab:LowerBounds}. Luckily, we have the following asymptotic result:

\begin{theorem}[Chandrasekhar \protect{\cite[pp.~14]{Chandrasekhar}}, Hughes \protect{\cite[pp.~91]{Hughes}}]
The probability density to arrive at point $\vc{R}$ after performing $n \gg 1$ steps of random walk is
\begin{equation}
  W(\vc{R}) \approx \left(\frac{3}{2 \pi n}\right)^{3/2} \exp{\left(-\frac{3\norm{\vc{R}}^2}{2n}\right)}.
  \label{eq:Chandrasekhar}
\end{equation}
\end{theorem}

\begin{theorem}
For every $n \gg 1$ there exists an \racp{n} QRAC with expected success probability $p \approx \frac{1}{2} + \sqrt{\frac{2}{3 \pi n}}$.
\label{thm:Lower}
\end{theorem}

\begin{proof}
Because of the spherical symmetry of the probability density in formula (\ref{eq:Chandrasekhar}), the average distance traveled after $n \gg 1$ steps of random walk is given by
\begin{equation}
  \E_\vis \norm{\sum_{i=1}^n \vc{v}_i} \approx
    \int_0^\infty R \cdot W(R) \cdot 4 \pi R^2 \, dR
  = 2 \sqrt{\frac{2n}{3\pi}}.
  \label{eq:ApproxDistance}
\end{equation}
From (\ref{eq:Eofd}) and (\ref{eq:Eofp}) we obtain
\begin{equation}
  \E(p) \approx \frac{1}{2} + \sqrt{\frac{2}{3 \pi n}},
  \label{eq:LowerBound}
\end{equation}
which gives the desired lower bound.
\end{proof}

Formally this lower bound holds only for large $n$. However, if one estimates the actual value of (\ref{eq:ApproxDistance}) by random sampling one can see that the asymptotic expression (\ref{eq:LowerBound}) is indeed smaller than the actual value (see Table.~\ref{tab:LowerBounds}).

\begin{table}[!hb]
  \centering
  \begin{tabular}{r|c|c|c|l|}
    & \multicolumn{2}{|c|}{Random measurements}
    & \multicolumn{2}{|c|}{Orthogonal measurements} \\ \cline{2-5}
    $n$ & Asymptotic & Sampling & Numerical  & \multicolumn{1}{|c|}{Exact} \\
    \hline
    $2$ & $0.825735$ & $0.8333$ & $\p{2}$    & $\frac{1}{2} + \frac{1}{2\sqrt{2}}$                                 \\
    $3$ & $0.765962$ & $0.7708$ & $\p{3}$    & $\frac{1}{2} + \frac{1}{2\sqrt{3}}$                                 \\
    $4$ & $0.730329$ & $0.7333$ & $0.741481$ & $\frac{1}{2} + \frac{1+\sqrt{3}}{8\sqrt{2}}$                        \\
    $5$ & $0.706013$ & $0.7082$ & $0.711803$ & $\frac{1}{2} + \frac{2+\sqrt{5}}{20}$                               \\
    $6$ & $0.688063$ & $0.6897$ & $0.686973$ & $\frac{1}{2} + \frac{1+\sqrt{3}+\sqrt{6}}{16\sqrt{3}}$              \\
    $7$ & $0.674113$ & $0.6754$ & $0.677458$ & $\frac{1}{2} + \frac{15+6\sqrt{5}+2\sqrt{13}+\sqrt{17}}{224}$       \\
    $8$ & $0.662868$ & $0.6638$ & $0.666270$ & $\frac{1}{2} + \frac{12+9\sqrt{3}+6\sqrt{5}+6\sqrt{7}+\sqrt{11}}{256\sqrt{2}}$ \\
    $9$ & $0.653553$ & $0.6544$ & $0.656893$ & $\frac{1}{2} + \frac{10\sqrt{3}+9\sqrt{11}+3\sqrt{19}}{384}$
  \end{tabular}
  \caption[Comparison of both lower bounds]{Comparison of \rac{n} QRACs with SR that use random and orthogonal measurements, respectively. For the first code we give the success probability according to the asymptotic expression (\ref{eq:LowerBound}) and a numerical value obtained by $10^6$ random samples. For the second code we give both the numerical and the exact value of the success probability according to equation (\ref{eq:CubicDistance}).}
  \label{tab:LowerBounds}
\end{table}

\subsubsection{Lower bound by orthogonal measurements} \label{sect:Lower2}

According to the upper bound obtained in Sect.~\ref{sect:UpperBound} the known \rac{2} and \rac{3} QRACs (see Sect.~\ref{sect:KnownQRACs}) are optimal. This suggests that orthogonal measurements can be used to construct good codes. Unfortunately this idea cannot be directly applied when $n > 3$, since in $\R^3$ there are only three mutually orthogonal directions. However, if we choose roughly one third of all measurements along each coordinate axis, we will get quite a lot of mutually orthogonal measurement pairs.

Let $\vc{v}_1 = (1,0,0)$, $\vc{v}_2 = (0,1,0)$, $\vc{v}_3 = (0,0,1)$, and $\forall i: \vc{v}_{i+3} \equiv \vc{v}_i$. According to equation (\ref{eq:pv}) in Sect.~\ref{sect:RandomWalk} the success probability of this \rac{n} QRAC with SR is related to the average distance (\ref{eq:dv}) traveled by a random walk. For our choice of measurement directions $\vc{v}_i$ the random walk takes place in a cubic lattice and consists of roughly $n/3$ steps along each coordinate axis. Thus we can simplify the equation (\ref{eq:dv}) for the average distance traveled to avoid having an exponential number of terms in it:
\begin{equation}
  d(\vc{v}_1, \dotsc, \vc{v}_n) =
  \frac{1}{2^n}
  \sum_{i=0}^{x}
  \sum_{j=0}^{y}
  \sum_{k=0}^{z}
  \binom{x}{i}
  \binom{y}{j}
  \binom{z}{k}
  \sqrt{(x - 2i)^2 + (y - 2j)^2 + (z - 2k)^2},
  \label{eq:CubicDistance}
\end{equation}
where $x + y + z = n$ and each of $x$, $y$, $z$ is roughly $n/3$. The corresponding success probability can be obtained by plugging this expression in equation (\ref{eq:pv}).

This lower bound is better than the one obtained in the previous section using random measurements and it also requires less shared randomness. The difference of both lower bounds is shown in Fig.~\ref{fig:OrthBound}. The periodic pattern of length $6$ in this picture can be explained as follows. When $n$ is a multiple of $3$, the same number of steps of a random walk is performed along each coordinate axis (this explains the factor $3$). To explain the factor $2$, let us consider a random walk on a line, i.e., one of the three coordinate axis. The distinction between odd an even number of steps of such a walk is that the probability distribution after an even number of steps is peaked at the origin, but this peak has no contribution whatsoever to the average distance traveled. This intuition suggests that it should be especially hard to beat this lower bound when $n$ is of the form $6k+3$.

\image{0.9}{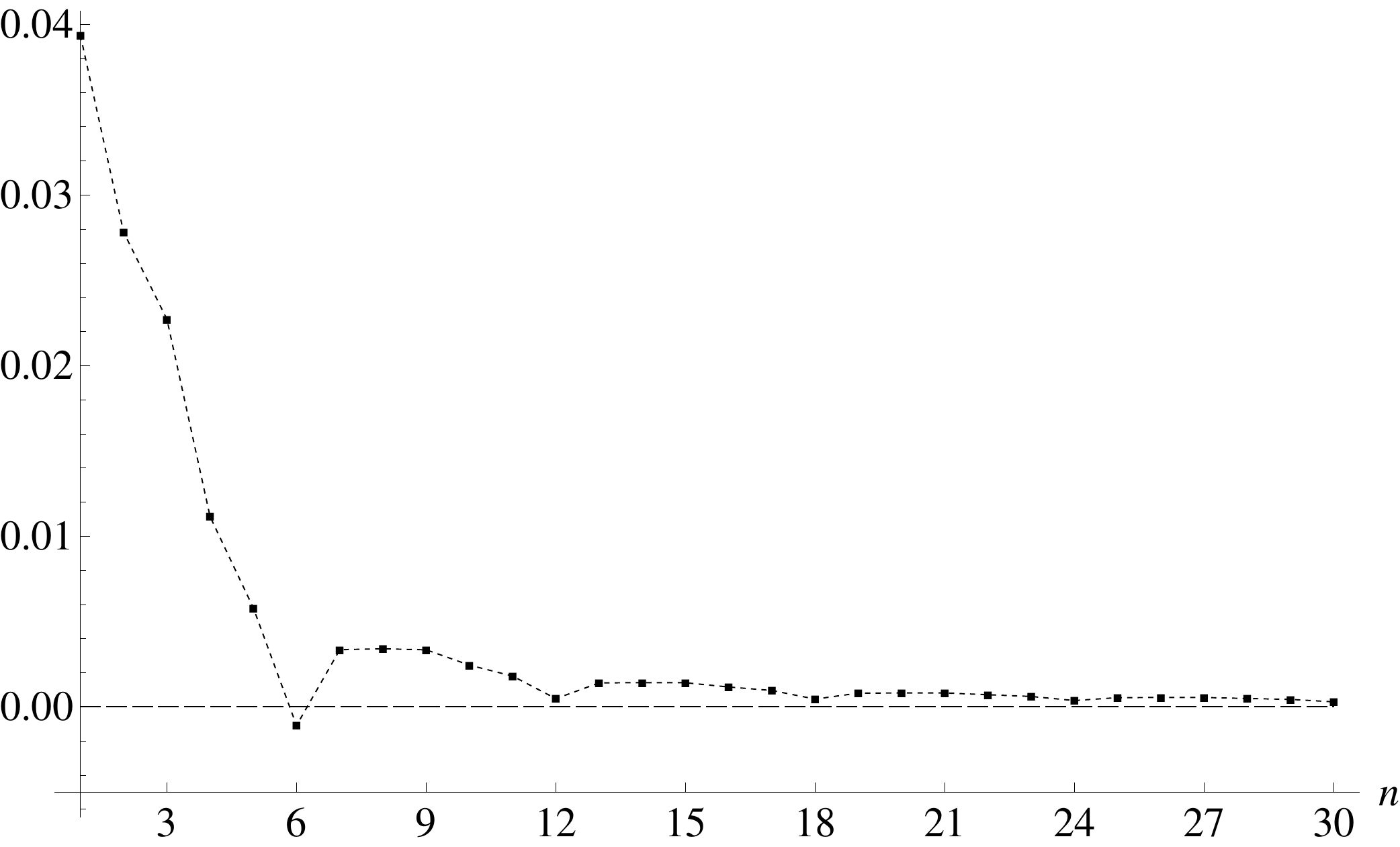}{The difference between two lower bounds for QRACs with SR}{The \emph{difference} of both lower bounds for QRACs with SR. Black squares correspond to the bound obtained using measurements along coordinate axes and the horizontal line corresponds to the asymptotic bound (\ref{eq:LowerBound}) using random measurements (see Sects.~\ref{sect:Lower1} and \ref{sect:Lower2}, respectively). The first bound is better, except for $n=6$ (notice a periodic pattern of length $6$).}{fig:OrthBound}

\section{Constructions of QRACs with SR} \label{sect:Constructions}

It is plausible that one can do better than the lower bound obtained above, which used random measurements. In this section we will consider several constructions of quantum random access codes with shared randomness for some particular values of $n$. First, in Sect.~\ref{sect:Numerical} we will describe numerically obtained QRACs. Then, in Sect.~\ref{sect:Symmetric} we will construct new QRACs with high degree of symmetry. In Sect.~\ref{sect:Discussion} we will compare both kinds of codes and draw some conclusions.

\subsection{Numerical results} \label{sect:Numerical}

\begin{table}[!hb]
  \centering
  \begin{tabular}{r|c|c}
    $n$ & Section                & Probability  \\
    \hline
    $2$ & \ref{sect:QRAC2,QRAC3} & $\p{2}$        \\
    $3$ & \ref{sect:QRAC2,QRAC3} & $\p{3}$        \\
    $4$ & \ref{sect:QRAC4}       & $\p{4}$        \\
    $5$ & \ref{sect:QRAC5}       & $\p{5}$        \\
    $6$ & \ref{sect:QRAC6}       & $\p{6}$        \\
    $7$ &                        & $0.678638$ \\ 
    $8$ &                        & $0.666633$ \\ 
    $9$ & \ref{sect:QRAC9}       & $\p{9}$        \\
   $10$ &                        & $0.648200$ \\ 
   $11$ &                        & $0.641051$ \\ 
   $12$ &                        & $0.634871$    
  \end{tabular}
  \caption[The success probabilities of numerical \rac{n} QRACs]{The success probabilities of numerical \rac{n} QRACs.}
  \label{tab:NumericalProbabilities}
\end{table}

In this section we will discuss some particular \rac{n} QRACs with shared randomness for several small values of $n$. These codes were obtained using numerical optimization. The optimization must be performed only over all possible measurements, because in Sect.~\ref{sect:OptimalEncoding} we showed that the choice of measurements determines the optimal encoding in a simple way. Each measurement is specified by a unit vector $\vc{v}_i\in\R^3$. For \rac{n} QRAC there are $n$ such vectors and one needs two angles to specify each of them. Without loss of generality we can assume that $\vc{v}_1=(0,0,1)$ due to the rotational symmetry of the Bloch sphere. Thus only $2(n-1)$ real parameters are required to specify all $\vc{v}_i$ and therefore an \rac{n} QRAC. To find the best configuration of measurements $\vc{v}_i$, one needs to maximize $\Sv$ given by (\ref{eq:Sv}). According to (\ref{eq:ProbabilitySn}) the success probability of the corresponding QRAC is given by
\begin{equation}
  p_\vc{v} = \frac{1}{2} \left(1 + \frac{\Sv}{2^n \cdot n}\right).
  \label{eq:SuccessProbability}
\end{equation}
This is not a convex optimization problem, since the feasible set (given by $\norm{\vc{v}_i}=1$ for all $1 \leq i \leq n$) is not convex. Note that it is not convex even if we relax equalities $\norm{\vc{v}_i}=1$ to inequalities $\norm{\vc{v}_i} \leq 1$. We used the \textit{Mathematica}'s general-purpose built-in function \texttt{NMaximize} to solve this problem.

Once the measurements $\vc{v}_i$ are found, one can easily obtain the Bloch vector $\vc{r}_x$ of the qubit state that must be used to optimally encode the string $x$. We showed (see Sect.~\ref{sect:OptimalEncoding}) that $\vc{r}_x$ is a unit vector in direction $\vc{v}_x$, where $\vc{v}_x$ is given by (\ref{eq:vx}). For almost all QRACs that we have found using numerical optimization, the points $\vc{r}_x$ form a symmetric pattern on the surface of the Bloch ball. Thus we were able to guess the exact values of $\vc{r}_x$ and $\vc{v}_i$. However, as in any numerical optimization, optimality of the resulting codes is not guaranteed.

In order to make the resulting codes more understandable, we depict them in three dimensions using the following conventions:
\begin{itemize}
  \item each \emph{red point} encodes the string indicated,
  \item each \emph{blue point} defines the axis of the measurement when the indicated bit is to be output, and
  \item for each measurement there is a corresponding (unlabeled) \emph{blue great circle} containing states yielding $0$ and $1$ equiprobably.
\end{itemize}
More precisely, the blue point with label $i$ defines the basis vector $\ket{\psi^i_0}$ corresponding to the outcome $0$ of the \mbox{$i$th} measurement (see Sect.~\ref{sect:TypesOfQRACs}). Note that the blue circles and blue points come in pairs---the vector $\ket{\psi^i_0}$ defined by the blue point is orthogonal to the corresponding circle. As a cautionary note, occasionally, the blue point for one measurement falls on the great circle of a different measurement (for example, blue points $1$ and $2$ in Fig.~\ref{fig:QRAC2} lie on one another's corresponding circles). If there are too many red points, we omit the string labels for clarity.

Usually the codes have some symmetry; for example, the encoding points may be the vertices of a polyhedron. In such cases we show the corresponding polyhedron instead of the Bloch sphere. We do not discuss \rac{7} and \rac{8} QRACs since the best numerical results have almost no discernible symmetry. We also do not discuss the numerical results for $n \geq 10$ (see Table~\ref{tab:NumericalProbabilities} for success probabilities). The numerically obtained \rac{10} code is symmetric and resembles \rac{6} code discussed in Sect.~\ref{sect:QRAC6}, but the \rac{11} and \rac{12} codes again have almost no discernible symmetry. Success probabilities of all numerical \rac{n} QRACs with SR are summarized in Table~\ref{tab:NumericalProbabilities} and Fig.~\ref{fig:NumericalResults}.

\image{0.9}{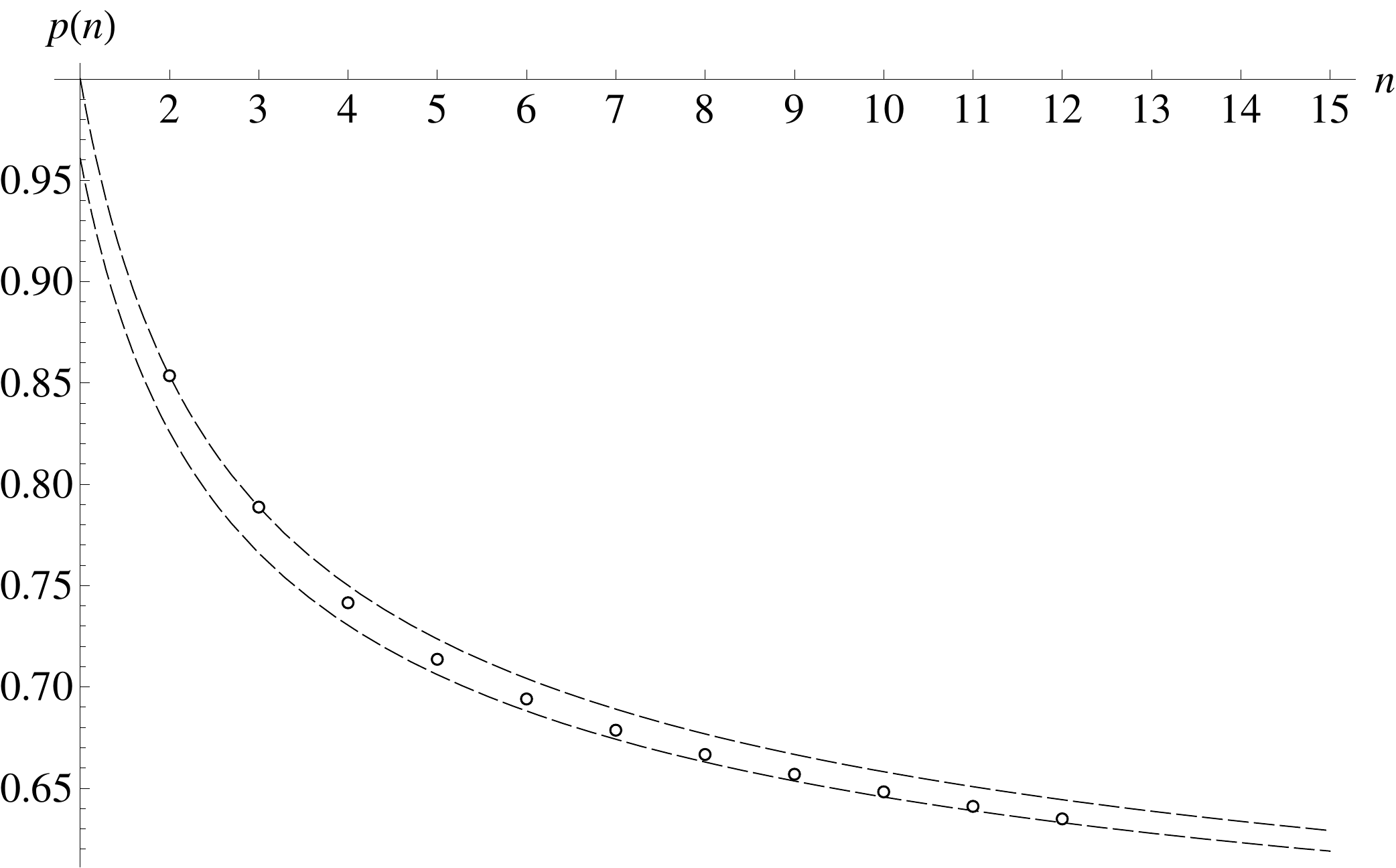}{Success probabilities of numerical \rac{n} QRACs with SR}{Success probabilities $p(n)$ of numerical \rac{n} QRACs with SR from Table~\ref{tab:NumericalProbabilities}. The upper bound $\frac{1}{2} + \frac{1}{2 \sqrt{n}}$ and the lower bound $\frac{1}{2} + \sqrt{\frac{2}{3 \pi n}}$ are indicated by dashed lines (see Sects.~\ref{sect:GeneralUpperBound} and \ref{sect:Lower1}, respectively).}{fig:NumericalResults}

\subsubsection{The \rac{2} and \rac{3} QRACs with SR} \label{sect:QRAC2,QRAC3}

\doubleimage
{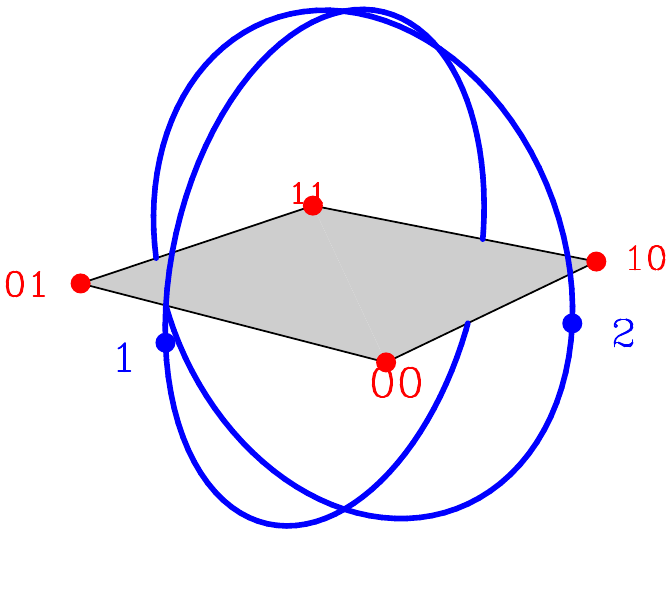}{The \rac{2} QRAC with SR}{The \rac{2} QRAC with SR.\protect\footnote{For those who are using a black-and-white printout: this is how \textbf{\textcolor{red}{red}} and \textbf{\textcolor{blue}{blue}} looks like.}}{fig:QRAC2}
{Exact3}{The \rac{3} QRAC with SR}{The \rac{3} QRAC with SR.}{fig:QRAC3}

We used numerical optimization as described above to find \rac{2} and \rac{3} QRACs with shared randomness and obtained the optimal codes discussed in Sects.~\ref{sect:KnownQRAC2} and \ref{sect:KnownQRAC3}.

The codes are shown in Fig.~\ref{fig:QRAC2} and \ref{fig:QRAC3}, respectively. In the first case the encoding points are the vertices of a square and the success probability is
\begin{equation}
  p = \frac{1}{2} + \frac{1}{2\sqrt{2}} \approx \p{2}.
\end{equation}
In the second case they are the vertices of a cube. The success probability is
\begin{equation}
  p = \frac{1}{2} + \frac{1}{2\sqrt{3}} \approx \p{3}.
\end{equation}

\subsubsection{The \rac{4} QRAC with SR} \label{sect:QRAC4}

\image{0.6}{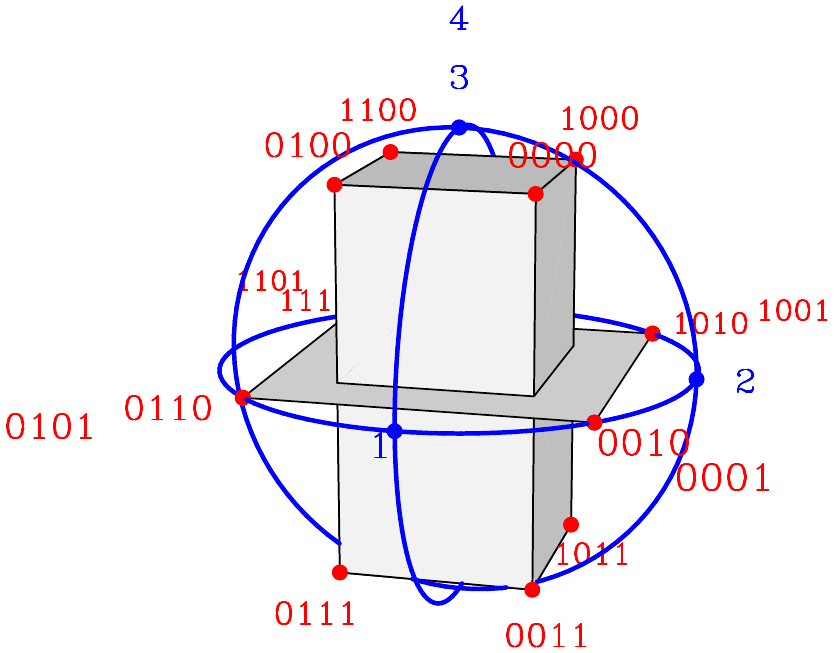}{The \rac{4} QRAC with SR}{The \rac{4} QRAC with SR.}{fig:QRAC4}

In Sect.~\ref{sect:noQRAC4} we discussed the impossibility of a \rac{4} QRAC when Alice and Bob are not allowed to cooperate. However, a \rac{4} QRAC can be obtained if they have shared randomness. The particular \rac{4} QRAC with SR discussed here was found by a numerical optimization. It is a hybrid of the \rac{2} and \rac{3} codes discussed in Sects.~\ref{sect:KnownQRAC2} and \ref{sect:KnownQRAC3}, respectively.

The measurements are performed in the bases $(M_1, M_2, M_3, M_3)$, where $M_1$, $M_2$, and $M_3$ are the same as in the \rac{3} case (note that the last two bits are measured in the same basis, namely $M_3$). These bases are given by (\ref{eq:M1}), (\ref{eq:M2}), and (\ref{eq:M3}), respectively. The points that correspond to an optimal encoding for these bases are the vertices of a regular square $\frac{1}{\sqrt{2}}(\pm 1,\pm 1,0)$ in the $xy$ plane and a cube $\frac{1}{\sqrt{6}}(\pm 1,\pm 1,\pm 2)$ that is stretched in the $z$ direction (see the Bloch sphere in Fig.~\ref{fig:QRAC4}). The Bloch vector for the string $x = x_1 x_2 x_3 x_4$ is explicitly given by
\begin{equation}
  \vc{r}(x) = \frac{1}{\sqrt{6}}
  \mx{
    (-1)^{x_1}\bigl(1-(1-\sqrt{3})\abs{x_3-x_4}\bigr) \\
    (-1)^{x_2}\bigl(1-(1-\sqrt{3})\abs{x_3-x_4}\bigr) \\
    (-1)^{x_3}+(-1)^{x_4}
  }.
\end{equation}

The encoding function can be described as follows:
\begin{itemize}
  \item if $x_3=x_4$, use the usual \rac{3} QRAC with an emphasis on $x_3$ to encode the string $x_1 x_2 x_3$,
  \item if $x_3 \neq x_4$---encode only $x_1 x_2$ using the usual \rac{2} QRAC.
\end{itemize}
In the \rac{3} scheme the probability to recover $x_3$ must be increased by stretching the cube along the $z$ axis, because $x_3$ equals $x_4$ and therefore it is of greater value than $x_1$ or $x_2$.

This \rac{4} QRAC can also be seen as a combination of two \rac{3} QRACs: the string $x_1 x_2 x_3$ is encoded into the vertices of a smaller cube inscribed in a half of the Bloch ball (the vertices that lie within the sphere are projected to its surface). The last bit $x_4$ indicates in which half the smaller cube lies (the upper and lower hemispheres correspond to $x_4 = 0$ and $1$, respectively).

The qubit state is explicitly given by $E(x_1,x_2,x_3,x_4)=\alpha\ket{0}+\beta\ket{1}$, where
\begin{equation}
  \left\{
  \begin{aligned}
    \alpha &= \sqrt{\frac{1}{2}+\frac{(-1)^{x_3}+(-1)^{x_4}}{2\sqrt{6}}}, \\
    \beta  &= \frac{(-1)^{x_1}+i(-1)^{x_2}}
                   {\sqrt{4\bigl(3-2\abs{x_3-x_4}\bigr)+2\sqrt{6}\bigl((-1)^{x_3}+(-1)^{x_4}\bigr)}}.
  \end{aligned}
  \right.
\end{equation}
The $16$ values for $\beta$ are exactly the sixteen roots of the polynomial (recall Sect.~\ref{sect:UnitDisk})
\begin{equation}
  2304 \beta^{16} + 3072 \beta^{12} + 1120 \beta^8 + 128 \beta^ 4 + 1.
\end{equation}

If a shared random string is not available, the worst case success probability of this QRAC is $\frac{1}{2}$. However, if shared randomness is available, input randomization (as in Lemma~\ref{lem:Randomization}) can be used and we will get the same success probability for all inputs, namely
\begin{equation}
  p = \frac{1}{2} + \frac{1+\sqrt{3}}{8\sqrt{2}} \approx \p{4}.
\end{equation}
We do not know if this \rac{4} QRAC with SR is optimal.

\subsubsection{The \rac{5} QRAC with SR} \label{sect:QRAC5}

\image{0.5}{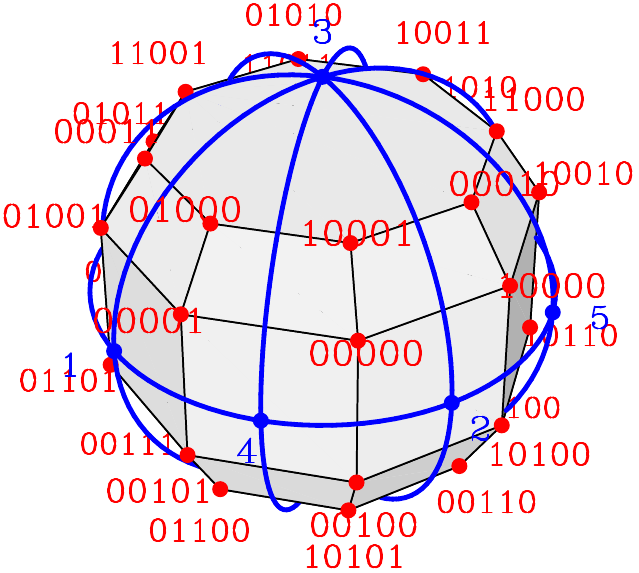}{The \rac{5} QRAC with SR}{The \rac{5} QRAC with SR.}{fig:QRAC5}

To obtain a \rac{5} QRAC, we take the bases $M_1$, $M_2$, and $M_3$, given by (\ref{eq:M1}), (\ref{eq:M2}), and (\ref{eq:M3}), respectively, and also
\begin{align}
  M_4&=\set{\frac{1}{2}\mx{ \sqrt{2}\\ i+1},
            \frac{1}{2}\mx{-\sqrt{2}\\ i+1}}, \label{eq:M4} \\
  M_5&=\set{\frac{1}{2}\mx{ \sqrt{2}\\ i-1},
            \frac{1}{2}\mx{-\sqrt{2}\\ i-1}}. \label{eq:M5}
\end{align}
The Bloch vectors $\vc{v}_3=(0,0,\pm1)$ for the basis $M_3$ are along the $z$ axis, but the Bloch vectors of the other four bases form a regular octagon in the $xy$ plane (shown in Fig.~\ref{fig:QRAC5}): $\vc{v}_1=(\pm1,0,0)$, $\vc{v}_2=(0,\pm1,0)$, $\vc{v}_4=\pm\frac{1}{\sqrt{2}}(1,1,0)$, $\vc{v}_5=\pm\frac{1}{\sqrt{2}}(-1,1,0)$. The Bloch vector encoding the string $x = x_1 x_2 x_3 x_4 x_5$ is
\begin{equation}
  \vc{r}(x) = \frac{1}{\sqrt{10+s(x)4\sqrt{2}}}
  \mx{
    \sqrt{2}(-1)^{x_1}+(-1)^{x_4}-(-1)^{x_5} \\
    \sqrt{2}(-1)^{x_2}+(-1)^{x_4}+(-1)^{x_5} \\
    \sqrt{2}(-1)^{x_3}
  },
\end{equation}
where $s(x) \in \set{-1,1}$ and is given by
\begin{equation}
  s(x) = \frac{(-1)^{x_1}+(-1)^{x_2}}{2} (-1)^{x_4} - \frac{(-1)^{x_1}-(-1)^{x_2}}{2} (-1)^{x_5}.
\end{equation}

The great circles with equiprobable outcomes of the measurements partition the Bloch sphere into $16$ equal spherical triangles. There are two strings encoded into each triangle. The idea for how to locate the correct point for the given string $x$ is as follows. Observe that the strings with $x_3=0$ and $x_3=1$ are encoded into the upper and lower hemisphere, respectively (this means that for all strings the probability that the measurement $M_3$ gives the correct value of $x_3$ is greater than $\frac{1}{2}$). Next observe that half of all strings have $s(x)=1$, but the other half have $s(x)=-1$ (in fact, the two strings in the same triangle have distinct values of $s$).

Let us first consider the case $s(x)=1$. We call such string \emph{compatible} with the measurements, because it can be encoded in such a way that every measurement gives the correct value of the corresponding bit with probability greater than $\frac{1}{2}$. For the \mbox{$i$th} bit of $x$ we can define the ``preferable region'' on the Bloch sphere as the hemisphere where $M_i$ recovers $x_i$ with probability greater than $\frac{1}{2}$. The intersection of these five regions is one sixteenth of the Bloch sphere---the triangle where $x$ must be encoded. The point with the smallest absolute value of the $z$ coordinate in this triangle must be chosen (it has smaller probability than other points in the triangle to recover $x_3$ correctly, but the probabilities for the other four bits are larger).

\image{0.3}{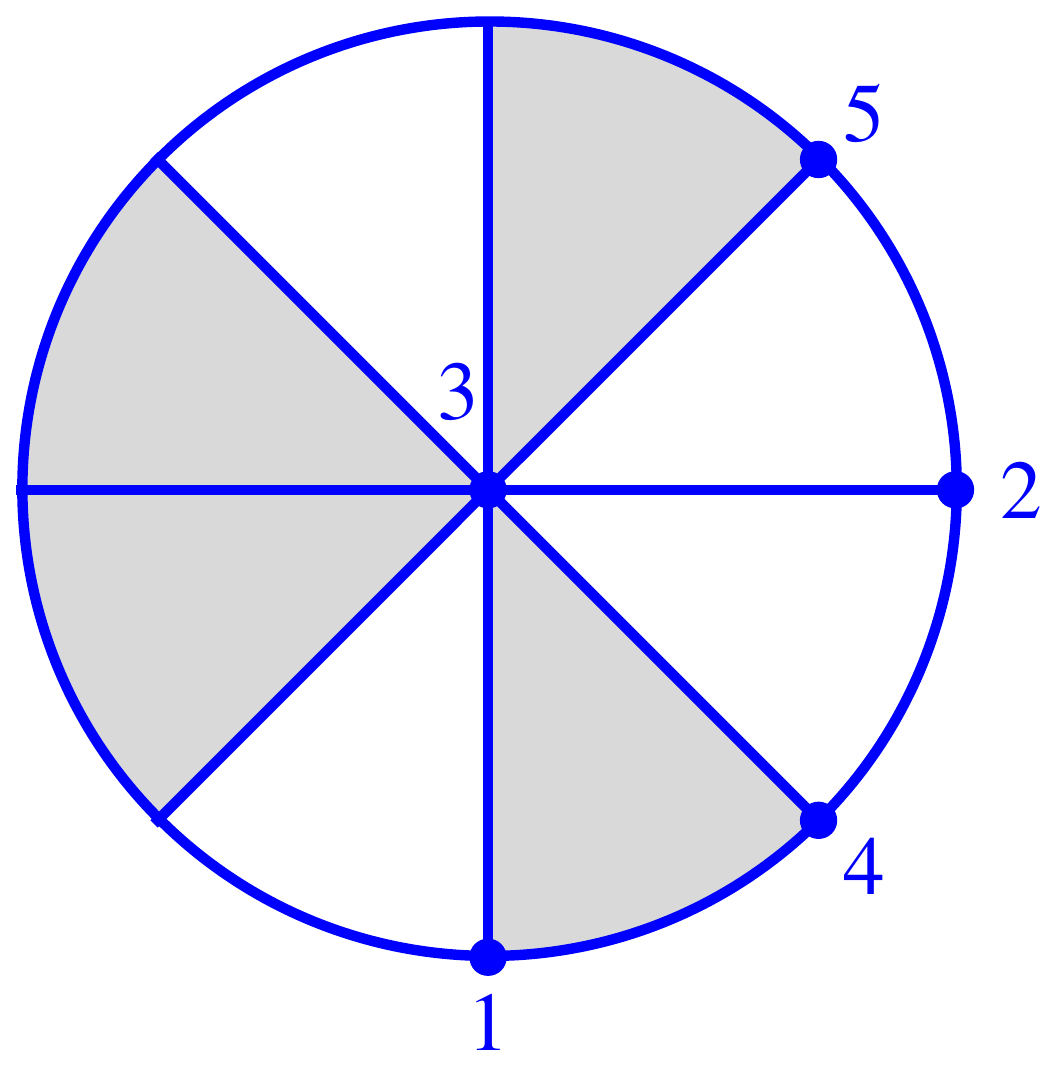}{The ``preferable regions'' of the measurement $M_2$}{The ``preferable regions'' of the measurement $M_2$ (only the upper hemisphere is shown, the other half is symmetric). For each of the measurements the direction of the Bloch vector $\ket{\psi_0}$ is indicated by the corresponding number. The white triangles correspond to $x_2=0$, but the gray ones to $x_2=1$.}{fig:Slices}

If $s(x)=-1$, the string $x$ is \emph{incompatible} with the measurements, because the intersection of the ``preferable regions'' is empty. Thus, no matter where the string is encoded, at least one bit will differ from the most probable outcome of the corresponding measurement. We can take this into account and modify the definition of the ``preferable region'' for the \mbox{$i$th} bit ($i \neq 3$). It is a union of eight triangles: four triangles where the most probable outcome of $M_i$ equals $x_i$, and four triangles where it does not equal $x_i$ (in either case the triangles with maximal probability of correct outcome of $M_i$ must be taken). For example, the ``preferable regions'' for $x_2$ are shown in Fig.~\ref{fig:Slices}. The regions for $x_3$ remain the same as in the previous case. The intersection of all five regions for the given string $x$ is the triangle where the string must be encoded. The point in the triangle with the largest absolute value of the $z$ coordinate must be chosen. As a result, three of the measurements will give the correct value of the corresponding bit of the string $x$ with probability greater than $\frac{1}{2}$.

The corresponding qubit state is given by $E(x_1,x_2,x_3,x_4,x_5)=\alpha\ket{0}+\beta\ket{1}$ with coefficients $\alpha$ and $\beta$ defined as follows:
\begin{equation}
  \left\{
  \begin{aligned}
    \alpha &= \sqrt{\frac{1}{2} + \frac{(-1)^{x_3}}{2 \sqrt{5 + s(x) 2 \sqrt{2}}}}, \\
    \beta  &= \frac{(-1)^{x_1} + i (-1)^{x_2} + \frac{i + 1}{\sqrt{2}} (-1)^{x_4} + \frac{i - 1}{\sqrt{2}} (-1)^{x_5}}
                   {\sqrt{10 + s(x) 4 \sqrt{2} + 2 (-1)^{x_3} \sqrt{5 + s(x) 2 \sqrt{2}}}}.
  \end{aligned}
  \right.
\end{equation}
The coefficients $\beta$ are the roots of the polynomial
\begin{equation}
  1336336 \beta^{32} + 961792 \beta^{24} + 151432 \beta^{16} + 1600 \beta^8 + 1.
\end{equation}

Again, using input randomization we obtain the same success probability for any input, namely
\begin{equation}
  p = \frac{1}{2} + \frac{1}{20} \sqrt{2(5+\sqrt{17})} \approx \p{5}.
\end{equation}

\subsubsection{The \rac{6} QRAC with SR} \label{sect:QRAC6}

\doubleimage
{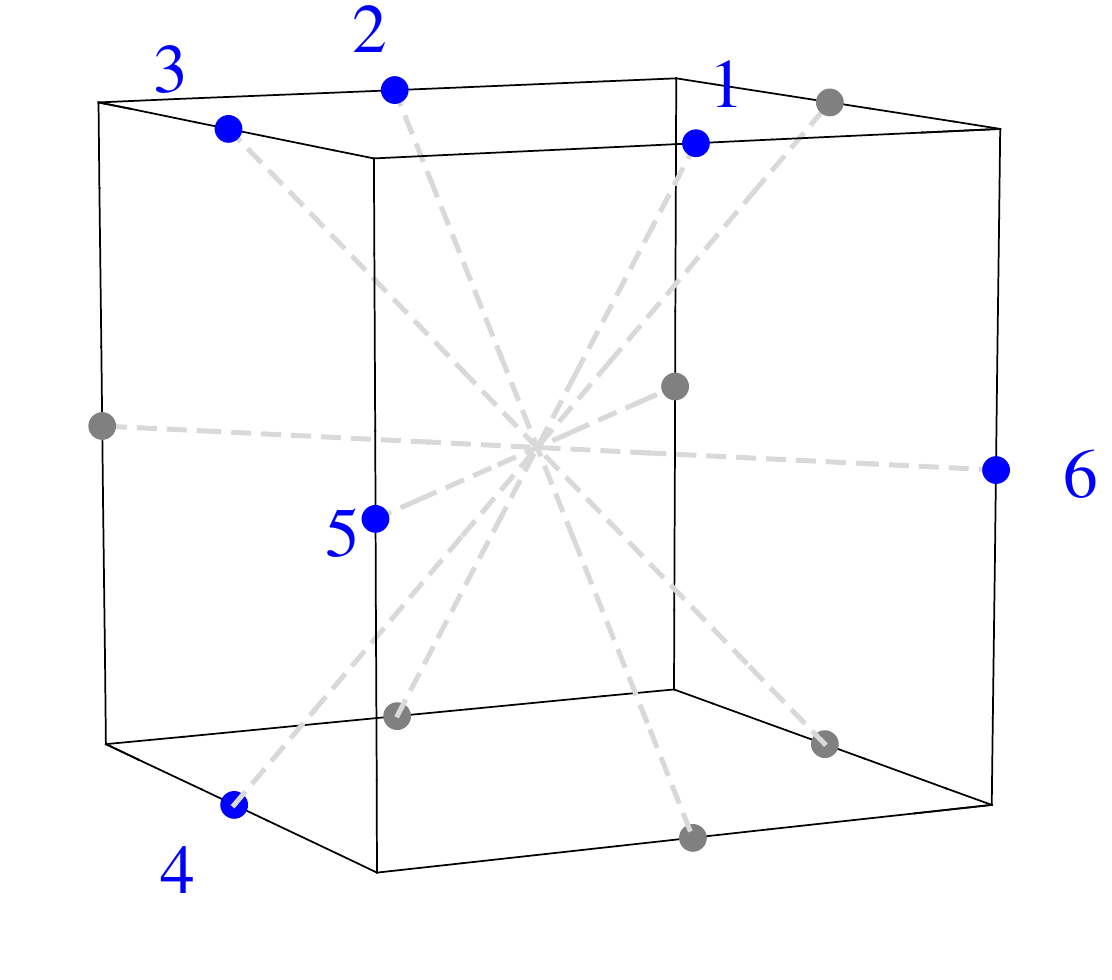}{The measurements for the \rac{6} QRAC}{The measurements for the \rac{6} QRAC shown on the right.}{fig:Cube}
{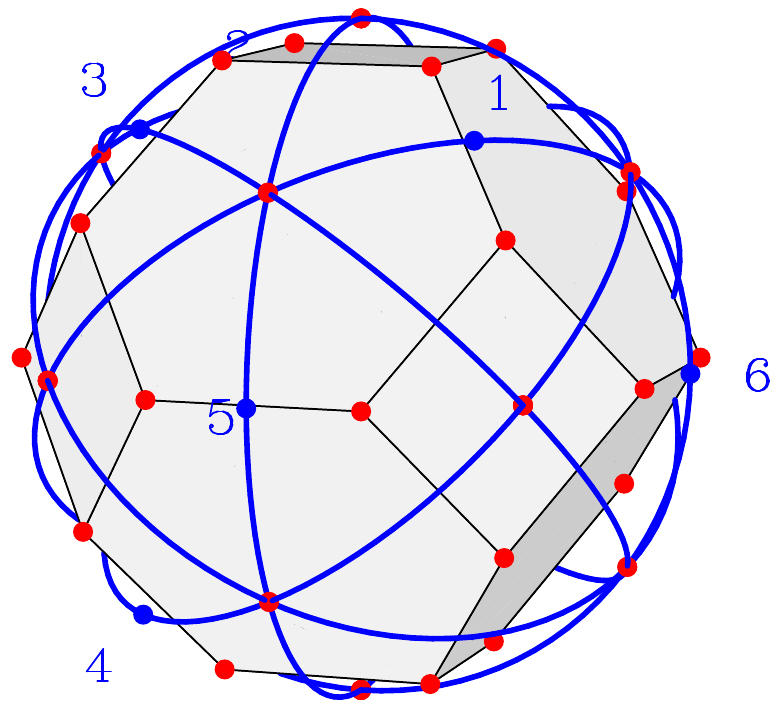}{The \rac{6} QRAC with SR}{The \rac{6} QRAC with SR.}{fig:QRAC6}

The Bloch vectors corresponding to the $6$ measurements are as follows:
\begin{equation}
  \begin{aligned}
    \vc{v}_1 &= \pm ( 0,+1,+1) / \sqrt{2}, \\
    \vc{v}_2 &= \pm ( 0,-1,+1) / \sqrt{2}, \\
    \vc{v}_3 &= \pm (+1, 0,+1) / \sqrt{2}, \\
    \vc{v}_4 &= \pm (+1, 0,-1) / \sqrt{2}, \\
    \vc{v}_5 &= \pm (+1,+1, 0) / \sqrt{2}, \\
    \vc{v}_6 &= \pm (-1,+1, 0) / \sqrt{2}.
  \end{aligned}
  \label{eq:v6}
\end{equation}
They correspond to the $12$ vertices of the \emph{cuboctahedron} (or the midpoints of the $12$ edges of the cube) and are shown in Fig.~\ref{fig:Cube}. The great circles orthogonal to these vectors form the projection of the edges of a normalized\footnote{The vertices of the \emph{tetrakis hexahedron} are not all at the same distance from the origin (the ones forming an octahedron are $2/\sqrt{3}$ times closer than those forming a cube). So the polyhedron has to be \emph{normalized} to fit inside the Bloch sphere (the vectors pointing to the vertices have to be rescaled to have a unit norm).} \emph{tetrakis hexahedron} and partition the Bloch sphere into $24$ parts (see Fig.~\ref{fig:QRAC6}). Each of these parts contains one vertex of a \emph{truncated octahedron}---the dual of tetrakis hexahedron. It is inscribed in the Bloch sphere shown in Fig.~\ref{fig:QRAC6}.

The measurement bases corresponding to $\vc{v}_i$ can be found using (\ref{eq:AlphaBeta}):
\begin{equation}
  \begin{aligned}
    M_1 &= \left\{ \frac{1}{2} \mx{ \sqrt{2+\sqrt{2}} \\  i \sqrt{2-\sqrt{2}} },
                   \frac{1}{2} \mx{ \sqrt{2-\sqrt{2}} \\ -i \sqrt{2+\sqrt{2}} } \right\}, \\
    M_2 &= \left\{ \frac{1}{2} \mx{ \sqrt{2+\sqrt{2}} \\ -i \sqrt{2-\sqrt{2}} },
                   \frac{1}{2} \mx{ \sqrt{2-\sqrt{2}} \\  i \sqrt{2+\sqrt{2}} } \right\}, \\
    M_3 &= \left\{ \frac{1}{2} \mx{ \sqrt{2+\sqrt{2}} \\    \sqrt{2-\sqrt{2}} },
                   \frac{1}{2} \mx{ \sqrt{2-\sqrt{2}} \\ -  \sqrt{2+\sqrt{2}} } \right\}, \\
    M_4 &= \left\{ \frac{1}{2} \mx{ \sqrt{2-\sqrt{2}} \\    \sqrt{2+\sqrt{2}} },
                   \frac{1}{2} \mx{ \sqrt{2+\sqrt{2}} \\ -  \sqrt{2-\sqrt{2}} } \right\}, \\
    M_5 &= \left\{ \frac{1}{2} \mx{ \sqrt{2} \\  i+1 },
                   \frac{1}{2} \mx{ \sqrt{2} \\ -i-1 } \right\}, \\
    M_6 &= \left\{ \frac{1}{2} \mx{ \sqrt{2} \\  i-1 },
                   \frac{1}{2} \mx{ \sqrt{2} \\ -i+1 } \right\}.
  \end{aligned}
\end{equation}
Note that $M_5$ and $M_6$ are the same as (\ref{eq:M4}) and (\ref{eq:M5}) for the \rac{5} QRAC described in the previous section. Another way to describe these $6$ bases is to consider the $\beta$ coefficients for the $12$ vectors that form them. It turns out that these coefficients are exactly the roots of the polynomial
\begin{equation}
  256 \beta^{12} - 128 \beta^8 - 44 \beta^4 + 1.
\end{equation}

Let us consider how to determine the point where a given string should be encoded. According to (\ref{eq:vx}) we have to find the sum of the vectors $\vc{v}_i$ defined in (\ref{eq:v6}), each taken with either a plus or a minus sign. These vectors correspond to six pairs of opposite edges of a cube and the signs determine which edge from each pair we are taking (see Fig.~\ref{fig:Cube}). There are only three distinct ways of doing this (see Fig.~\ref{fig:Cubes}). Regardless of which way it is, for each of the chosen edges there is exactly one other that shares a common face and is parallel to it. Thus we can partition the chosen edges into three pairs (in Fig.~\ref{fig:Cubes} such pairs are joined with a thick blue line). The sum of the vectors $\vc{v}_i$ for edges in a pair is always parallel to one of the axes and its direction is indicated with an arrow in Fig.~\ref{fig:Cubes}. From these arrows one can see where the encoding point should lie.

Now let us classify all $2^6 = 64$ strings of length $6$ into $3$ types according to the location of the encoding point on the Bloch sphere. Each type of string is encoded into a vertex of a specific polyhedron (see Fig.~\ref{fig:Polyhedra6}). These polyhedra are the \emph{cube}, the \emph{truncated octahedron}, and the \emph{octahedron} and the number of strings of each type are $16$, $24$, and $24$, respectively. Let us consider them case by case:

\image{1.0}{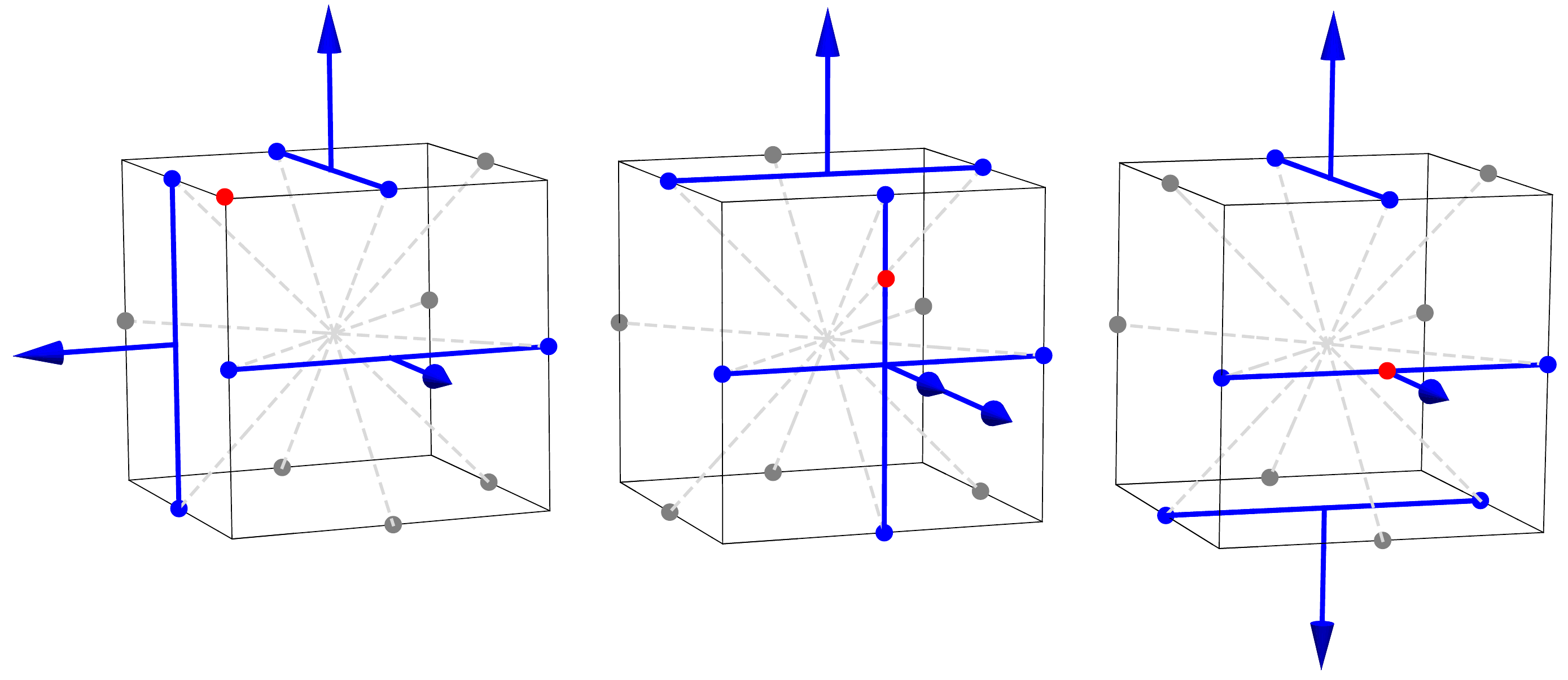}{Three distinct ways of choosing one edge from each pair of opposite edges of a cube}{Three distinct ways of choosing one edge from each pair of opposite edges of a cube. The chosen edges are marked with blue points. Points lying on opposite edges of the same face are connected and the direction of the sum of the corresponding vectors is indicated with an arrow. The corresponding encoding point is shown in red. The red points obtained from all possible choices of the same kind are the vertices of a cube, a truncated octahedron, and an octahedron, respectively (see Fig.~\ref{fig:Polyhedra6}).}{fig:Cubes}

\begin{table}[!hb]
  \centering
  \begin{tabular}{c|c}
    Truncated  & \multirow{2}{*}{Octahedron} \\
    octahedron & \\
    \hline
      $**1110$ & $**1101$ \\
      $**0001$ & $**0010$ \\
      $10**11$ & $01**11$ \\
      $01**00$ & $10**00$ \\
      $1110**$ & $1101**$ \\
      $0001**$ & $0010**$ \\
  \end{tabular}
  \caption[Patterns of strings corresponding to the vertices of truncated octahedron and octahedron]{Patterns of strings corresponding to the vertices of truncated octahedron and octahedron (``$*$'' stands for any value).}
  \label{tab:Patterns}
\end{table}

\begin{itemize}
  \item The \emph{cube} has $8$ vertices:
    \begin{equation}
      \frac{1}{\sqrt{3}}(\pm 1, \pm 1, \pm 1)
    \end{equation}
and there are $2$ strings encoded into each vertex. These $16$ strings are exactly those $x_1 x_2 \ldots x_6 \in \set{0,1}^6$ that satisfy
    \begin{equation}
      \abs{x_1 - x_2} + \abs{x_3 - x_4} + \abs{x_5 - x_6} \in \set{0,3}.
    \end{equation}
This condition ensures that the three arrows in Fig.~\ref{fig:Cubes} are orthogonal.
  \item The \emph{truncated octahedron} has $24$ vertices. Their coordinates are obtained by all permutations of the components of
    \begin{equation}
      \frac{1}{\sqrt{5}}(0, \pm 1, \pm 2).
    \end{equation}
There is just $1$ string encoded into each vertex. In this case there will be two pairs of chosen edges that belong to the same face (note the ``cross'' in the Fig.~\ref{fig:Cubes} formed by pairs whose arrows are pointing outwards of the page). The third pair (with the arrow pointing up) can be rotated around this face to any of the four possible positions. This corresponds to fixing four bits of the string and choosing the remaining two bits in an arbitrary way. Since the ``cross'' can be on any of the six faces of the cube, one can easily describe all $24$ strings of this type (they are listed in the first column of Table~\ref{tab:Patterns}).
  \item The \emph{octahedron} has $6$ vertices:
    \begin{equation}
      (\pm 1, 0, 0) \cup
      (0, \pm 1, 0) \cup
      (0, 0, \pm 1)
    \end{equation}
and there are $4$ strings encoded into each vertex. In this case two arrows in Fig.~\ref{fig:Cubes} are pointing to opposite directions (up and down). If we fix these arrows, we can rotate the third one (pointing outwards) in any of four directions. Hence we can describe all $24$ strings of this type in a similar way (see the second column of Table~\ref{tab:Patterns}).
\end{itemize}

\image{1.0}{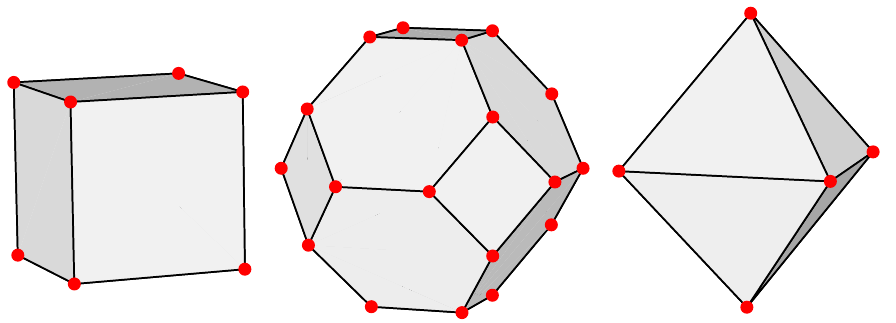}{Polyhedra corresponding to three different types of strings for \rac{6} QRAC with SR}{Three polyhedra (cube, truncated octahedron, and octahedron) corresponding to three different types of strings for \rac{6} QRAC with SR. The red points in Fig.~\ref{fig:QRAC6} are obtained by superimposing these three polyhedra.}{fig:Polyhedra6}

The coefficients $\beta$ of the encoding states are the $64$ roots of the polynomial
\begin{gather}
  \beta^4 (\beta-1)^4 (4 \beta^4-1)^4 (36 \beta^8+24 \beta^4+1)^2 \nonumber \\
  (25 \beta^8-15 \beta^4+1) (400 \beta^8-360 \beta^4+1) (400 \beta^8+56 \beta^4+25).
\end{gather}

The obtained success probability using input randomization is
\begin{equation}
  p = \frac{1}{2} + \frac{2+\sqrt{3}+\sqrt{15}}{16\sqrt{6}} \approx \p{6}.
\end{equation}

\subsubsection{The \rac{9} QRAC with SR} \label{sect:QRAC9}

\image{0.4}{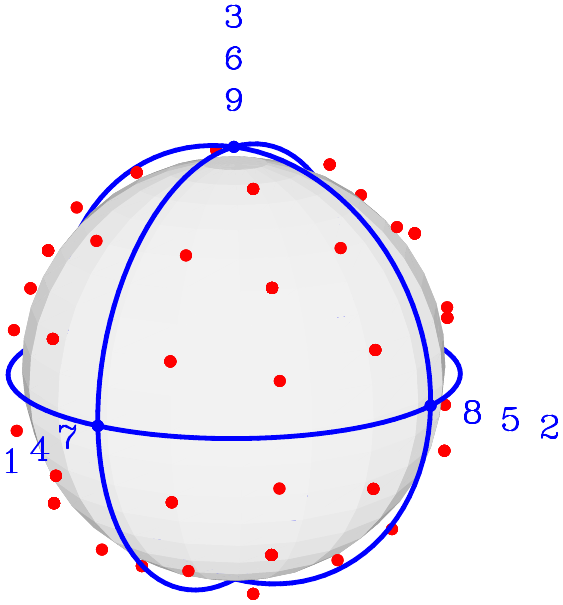}{The \rac{9} QRAC with SR}{The \rac{9} QRAC with SR.}{fig:QRAC9}

This QRAC is a combination of three \rac{3} QRACs described in Sect.~\ref{sect:KnownQRAC3}. It has three measurements along each axis:
\begin{equation}
  \begin{aligned}
    \vc{v}_1 &= \vc{v}_4 = \vc{v}_7 = \pm (1, 0, 0), \\
    \vc{v}_2 &= \vc{v}_5 = \vc{v}_8 = \pm (0, 1, 0), \\
    \vc{v}_3 &= \vc{v}_6 = \vc{v}_9 = \pm (0, 0, 1).
  \end{aligned}
\end{equation}
The measurement bases  $M_1$, $M_2$, and $M_3$ corresponding to the Bloch vectors $\vc{v}_1$, $\vc{v}_2$, and $\vc{v}_3$ are given by (\ref{eq:M1}), (\ref{eq:M2}), and (\ref{eq:M3}), respectively.

The encoding points can be characterized as a $4 \times 4 \times 4$ \emph{cubic lattice} formed by vectors (\ref{eq:vx}) projected on the surface of the Bloch ball. Note that this lattice consists of vertices of $8$ equal cubes each lying in a different octant. Then the $7$ points inside of each spherical triangle in Fig.~\ref{fig:QRAC9} are the projection of the vertices of the corresponding cube.\footnote{We get $7$ points instead of $8$ since the projections of two diagonally opposite vertices coincide.}

All $2^9=512$ strings can be classified into $3$ types. First consider a string $a_1 a_2 a_3 \in \set{0,1}^3$ and define
\begin{equation}
  s(a_1, a_2, a_3) =
    \frac{
      \abs{a_1 - a_2} +
      \abs{a_2 - a_3} +
      \abs{a_3 - a_1}
    }{2}.
\end{equation}
Notice that $s(a_1, a_2, a_3) \in \set{0,1}$. Now for $x = x_1 x_2 \ldots x_9 \in \set{0,1}^9$ define
\begin{equation}
  t(x) =
    s(x_1, x_4, x_7) +
    s(x_2, x_5, x_8) +
    s(x_3, x_6, x_9).
\end{equation}
Then the type of the string $x$ can be determined as follows:
\begin{equation}
  t(x) =
  \begin{cases}
    0, 3 & \text{cube}, \\
    1    & \text{truncated cube}, \\
    2    & \text{small rhombicuboctahedron}.
  \end{cases}
\end{equation}
These types are named after polyhedra, since each type of string is encoded into the vertices of the corresponding polyhedron (see Fig.~\ref{fig:Polyhedra9}):
\begin{itemize}
  \item The \emph{cube} has $8$ vertices and there are $28$ strings encoded into each vertex. These vertices are:
    \begin{equation}
      \frac{1}{\sqrt{3}}(\pm 1, \pm 1, \pm 1).
    \end{equation}
  \item The deformed\footnote{The edges of the \emph{truncated cube} are of the same length. In our case the eges forming triangles are $\sqrt{2}$ times longer than the other edges.} \emph{truncated cube} has $24$ vertices and there are $3$ strings encoded into each vertex. These vertices are:
    \begin{equation}
      \frac{1}{\sqrt{19}}(\pm 1, \pm 3, \pm 3) \cup
      \frac{1}{\sqrt{19}}(\pm 3, \pm 1, \pm 3) \cup
      \frac{1}{\sqrt{19}}(\pm 3, \pm 3, \pm 1).
    \end{equation}
  \item The deformed\footnote{The edges of the \emph{small rhombicuboctahedron} are also of the same length, but in our case the edges forming triangles again are $\sqrt{2}$ times longer.} \emph{small rhombicuboctahedron} also has $24$ vertices and there are $9$ strings encoded into each vertex. These vertices are:
    \begin{equation}
      \frac{1}{\sqrt{11}}(\pm 3, \pm 1, \pm 1) \cup
      \frac{1}{\sqrt{11}}(\pm 1, \pm 3, \pm 1) \cup
      \frac{1}{\sqrt{11}}(\pm 1, \pm 1, \pm 3).
    \end{equation}
\end{itemize}

\image{1.0}{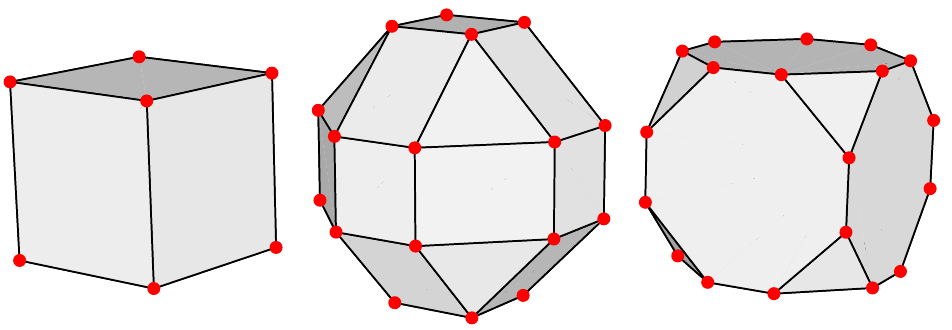}{Polyhedra corresponding to three different types of strings for \rac{9} QRAC with SR}{Three polyhedra (cube, small rhombicuboctahedron, and truncated cube) corresponding to three different types of strings for \rac{9} QRAC with SR. The red points in Fig.~\ref{fig:QRAC9} are obtained by superimposing these three polyhedra.}{fig:Polyhedra9}

The coefficients $\beta$ for the corresponding qubit states $\alpha\ket{0}+\beta\ket{1}$ are the roots of the following polynomial:
\begin{gather}
  (36 \beta^8 + 24 \beta^4 + 1)^{28} (1444 \beta^8 + 760 \beta^4 + 81)^3 (484 \beta^8 + 440 \beta^4 + 1)^9 \nonumber \\
  (52128400 \beta^{16} - 21509824 \beta^{12} + 26780424 \beta^8 - 372400 \beta^4 + 15625)^3 \nonumber \\
  (5856400 \beta^{16} - 1788864 \beta^{12} + 1232264 \beta^8 - 92400 \beta^4 + 15625)^9.
\end{gather}

Using input randomization we get success probability
\begin{equation}
  p = \frac{1}{2} + \frac{192+10\sqrt{3}+9\sqrt{11}+3\sqrt{19}}{384} \approx \p{9}.
\end{equation}

\subsection{Symmetric constructions} \label{sect:Symmetric}


In Sect.~\ref{sect:Numerical} we have discussed in great detail \rac{n} quantum random access codes with shared randomness for some particular values of $n$. Since these codes were obtained using numerical optimization, there are still some questions left open. Most importantly, are the codes for $n \geq 4$ discussed in Sect.~\ref{sect:Numerical} optimal? If this is the case, do these codes (see Figs.~\ref{fig:QRAC2}, \ref{fig:QRAC3}, \ref{fig:QRAC4}, \ref{fig:QRAC5}, \ref{fig:QRAC6}, and \ref{fig:QRAC9}) have anything in common that makes them so good?

The purpose of this section is to shed some light on these two questions. We will explore the possibility that \emph{symmetry} is the property that makes QRACs with SR good. In Sect.~\ref{sect:GreatCircles} we will explore what symmetries the codes found by numerical optimization have and what other symmetries are possible. In several subsequent sections we will use these symmetries to construct new codes and compare them with the numerical ones (the success probabilities of the obtained codes are summarized in Table~\ref{tab:SymmetricProbabilities}). In Sect.~\ref{sect:Discussion} we will conclude that symmetric codes are \emph{not} necessarily optimal and speculate about what else could potentially be used to construct good QRACs.

\begin{table}[!hb]
  \centering
  \begin{tabular}{r|c|c}
    $n$ & Section                & Probability   \\
    \hline
    $4$ & \ref{sect:QRAC4[sym]}  & $\p{4[sym]}$  \\
    $6$ & \ref{sect:QRAC6[sym]}  & $\p{6[sym]}$  \\
    $9$ & \ref{sect:QRAC9[sym]}  & $\p{9[sym]}$  \\
   $15$ & \ref{sect:QRAC15[sym]} & $\p{15[sym]}$
  \end{tabular}
  \caption[The success probabilities of symmetric \rac{n} QRACs with SR]{The success probabilities of symmetric \rac{n} QRACs with SR. See Table~\ref{tab:ComparisonOfProbabilities} for the comparison with numerically obtained codes.}
  \label{tab:SymmetricProbabilities}
\end{table}

\subsubsection{Symmetric great circle arrangements} \label{sect:GreatCircles}

If we want to construct a QRAC with SR that has some sort of symmetry, we have to choose the directions of measurements in a symmetric way. In other words, we have to symmetrically arrange the great circles that are orthogonal to the measurement directions.

In this section we will discuss two ways that great circles can be arranged on a sphere in a symmetric way. These arrangements come from \emph{quasiregular polyhedra} and \emph{triangular symmetry groups}, respectively. The first kind of arrangement is not directly observed in numerically obtained examples, despite its high symmetry. However, the second one is observed in almost all numerically obtained codes. Since our approach is empirical, we will not justify when an arrangement is ``symmetric enough''\footnote{Several possible criteria are: (a) any great circle can be mapped to any other by a rotation from the symmetry group of the arrangement, (b) the sphere is cut into pieces that are regular polygons, (c) the sphere is cut into pieces of the same form. However, not all examples we will give satisfy these three conditions. In fact, each condition is violated by at least one of the examples we will consider.} to be of interest. We will use the term \emph{symmetric codes} to refer to the codes constructed below. This is just to distinguish them from numerically obtained codes in Sect.~\ref{sect:Numerical}, not because they satisfy some formal criterion of ``being symmetric''.

\subsubsubsection{Quasiregular polyhedra}

A (convex) \emph{quasiregular polyhedron} is the intersection of a \emph{Platonic solid} with its dual. There are only three possibilities:
\begin{align}
  \text{octahedron} &= \text{tetrahedron} \cap \text{tetrahedron}, \\
  \text{cuboctahedron} &= \text{cube} \cap \text{octahedron}, \\
  \text{icosidodecahedron} &= \text{icosahedron} \cap \text{dodecahedron}.
\end{align}
The tetrahedron is self-dual thus the octahedron, which is the intersection of two tetrahedrons, has slightly different properties than the other two polyhedra (e.g., its all faces are equal). For this reason octahedron may be considered as a degenerate quasiregular polyhedron or not be considered quasiregular at all since it is Platonic. Thus there are only two (non-degenerate) convex quasiregular polyhedra (see Fig.~\ref{fig:Quasiregular}).

\image{0.6}{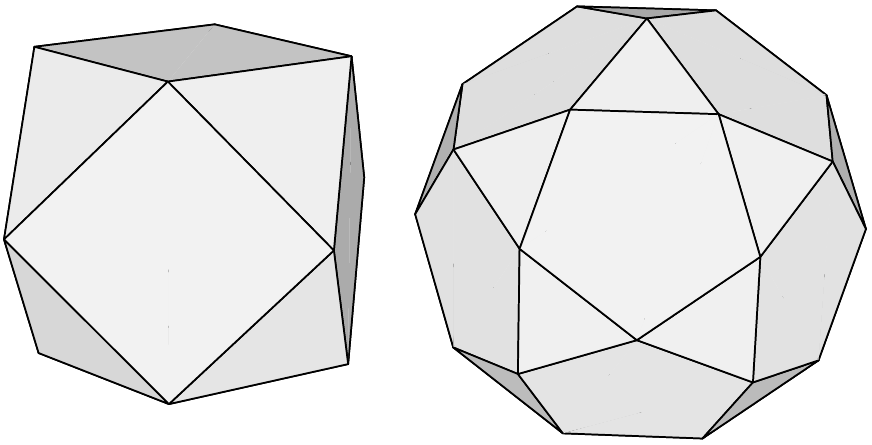}{Quasiregular polyhedra}{Quasiregular polyhedra: \emph{cuboctahedron} and \emph{icosidodecahedron}.}{fig:Quasiregular}

These polyhedra have several nice properties. For example, all their edges are equivalent and there are exactly two types of faces (both regular polygons), each completely surrounded by the faces of the other type. The most relevant property for us is that their edges form great circles. Since the arrangements of great circles formed by the edges of cuboctahedron and icosidodecahedron do not appear in the numerical codes, we will use them in Sects.~\ref{sect:QRAC4[sym]} and \ref{sect:QRAC6[sym]} to construct new (symmetric) \rac{4} and \rac{6} QRACs with SR, respectively.

\subsubsubsection{Triangular symmetry groups}

Consider a spherical triangle---it is enclosed by three planes that pass through its edges and the center of the sphere. Let us imagine that these planes are mirrors that reflect our triangle. These three reflections generate a \emph{reflection group} \cite{FiniteReflectionGroups,TriangularSymmetryGroups}. For some specific choices of the triangle this group is finite and the images of the triangle under different group operations do not overlap. Hence they form a \emph{tiling} of the sphere. This tiling can also be seen as several (most likely more than three) great circles cutting the sphere into equal triangles.

We can choose any of the triangles in the tiling and repeatedly reflect it along its edges so that it moves around one of its vertices. This means that the angles of the corners that meet at any vertex of the tiling must be equal. Moreover, we do not want the triangle to intersect with any of the mirrors, so only an even number of triangles can meet at a vertex.\footnote{Fore example, if we project the edges of an icosahedron on the sphere, we obtain arcs that form a tiling with five triangles meeting at each vertex. We cannot use these arcs as mirrors, since they do not form great circles (we cannot extend any of them to a great circle, without intersecting other triangles).}

\begin{figure}
  \centering
  \begin{minipage}[t]{0.8\textwidth}
    \centering
    \includegraphics[width=0.3\textwidth]{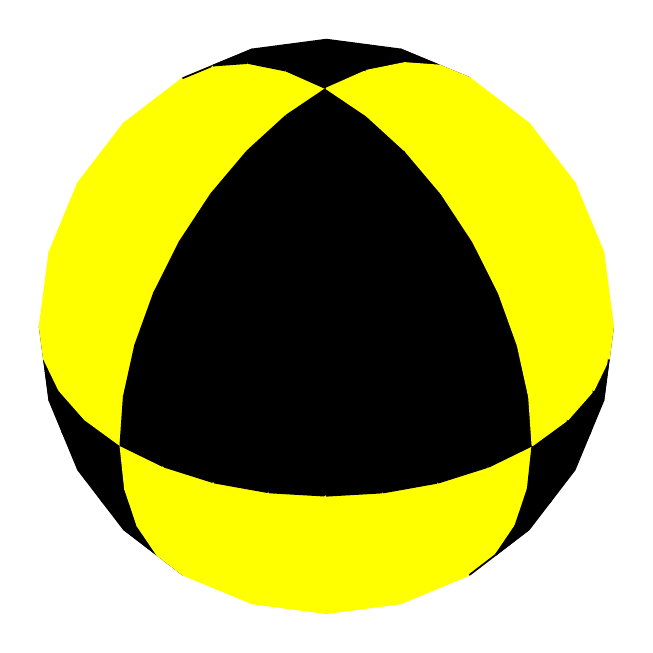}
    \includegraphics[width=0.3\textwidth]{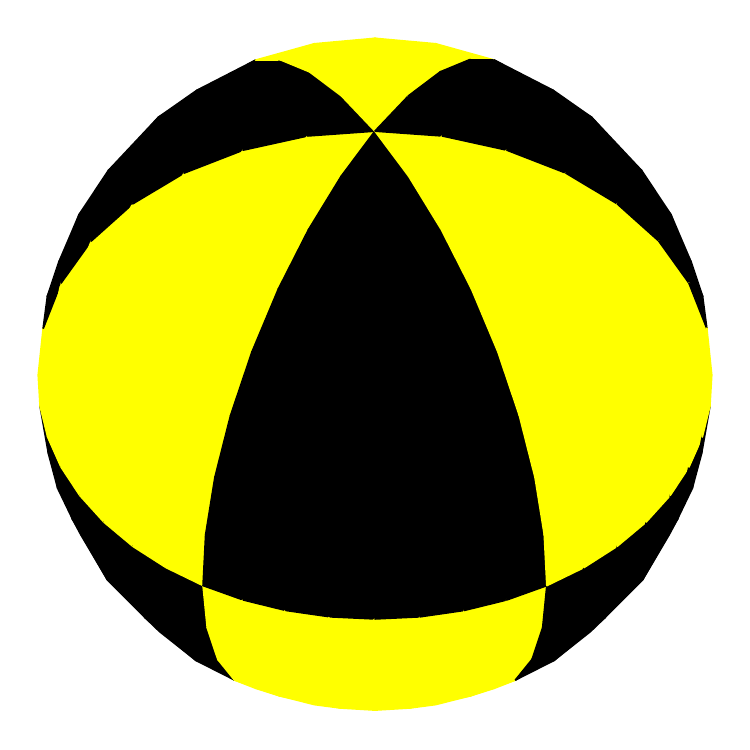}
    \includegraphics[width=0.3\textwidth]{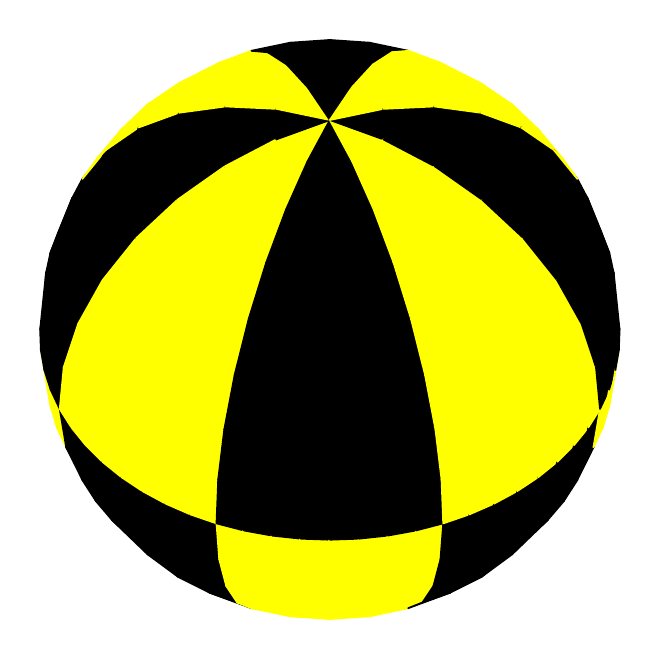}
  \end{minipage}
  \begin{minipage}[t]{0.8\textwidth}
    \centering
    \includegraphics[width=0.3\textwidth]{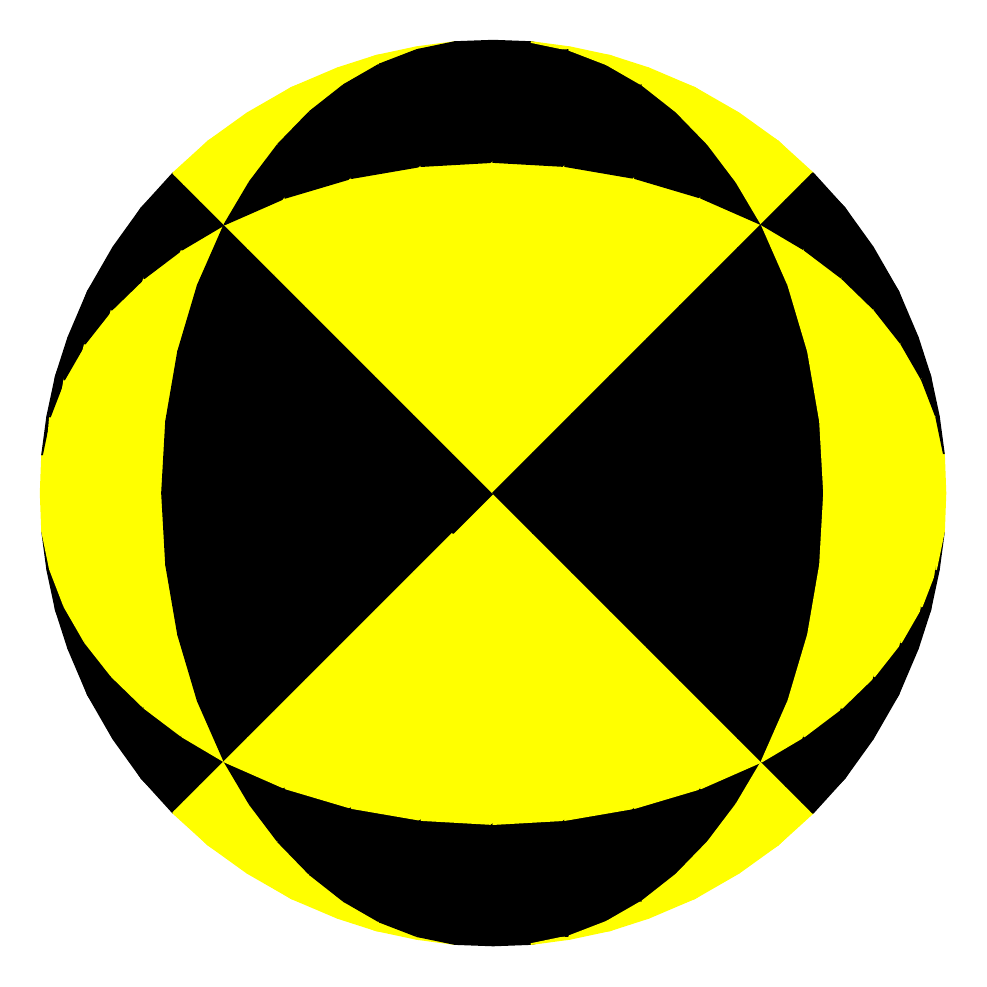}
    \includegraphics[width=0.3\textwidth]{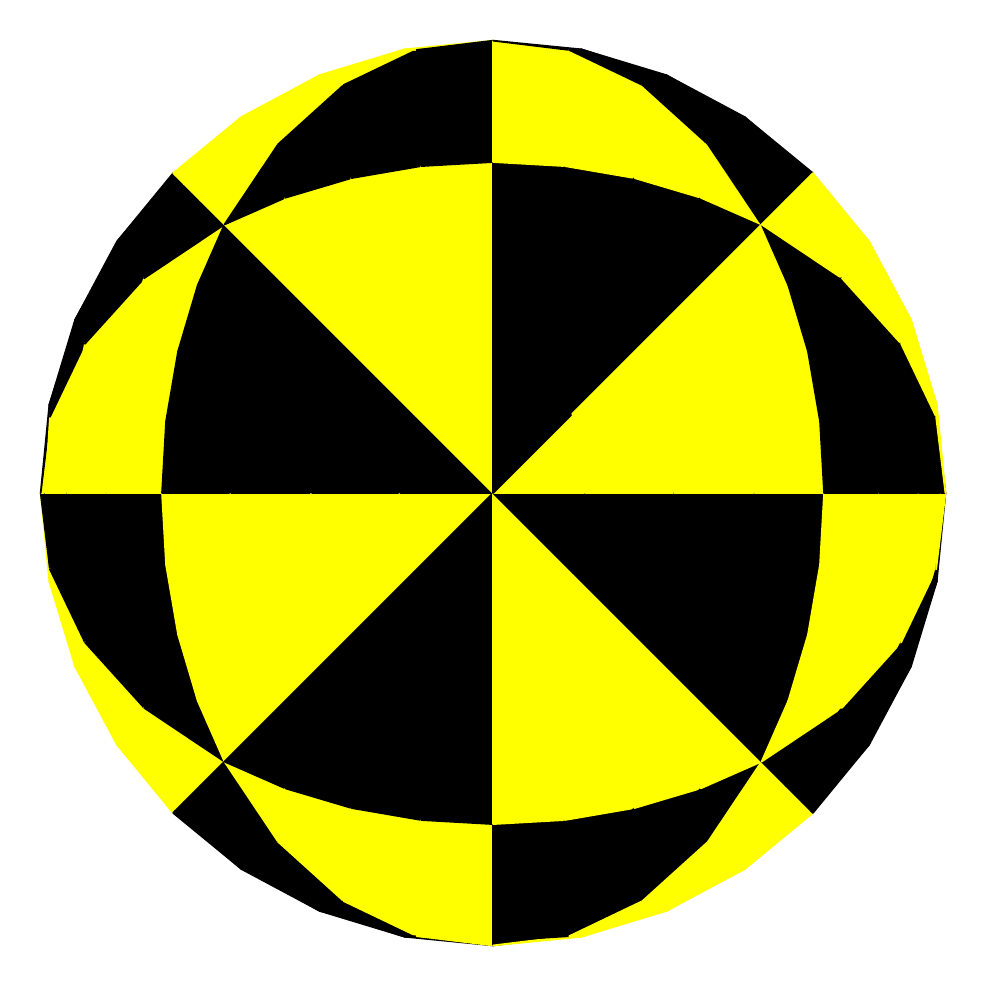}
    \includegraphics[width=0.3\textwidth]{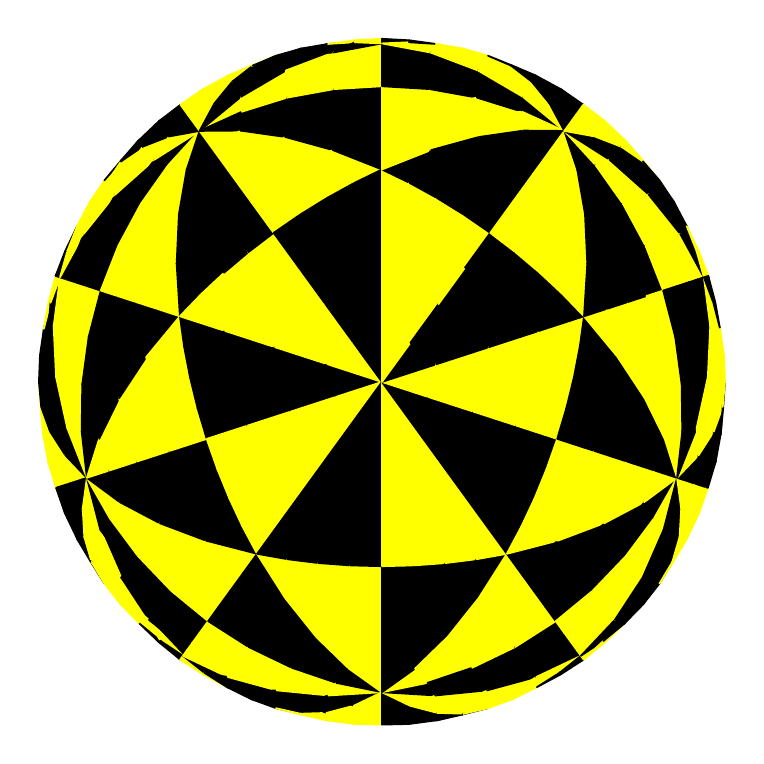}
  \end{minipage}
  \caption[Triangular symmetry groups]{Triangular symmetry groups. First row: $(2,2,2)$, $(2,2,3)$, $(2,2,4)$. Second row: $(2,3,3)$, $(2,3,4)$, $(2,3,5)$.}
  \label{fig:TriangularSymmetryGroups}
\end{figure}

Hence the angles of the spherical triangle must be $(\frac{\pi}{p}, \frac{\pi}{q}, \frac{\pi}{r})$ for some integers $p,q,r \geq 2$. The sum of the angles of a spherical triangle is at least $\pi$, so the numbers $p,q,r$ must satisfy
\begin{equation}
  \frac{1}{p} + \frac{1}{q} + \frac{1}{r} > 1.
\end{equation}
If $p \leq q \leq r$, the only solutions are: $(2,2,k)$ for any $k \geq 2$, $(2,3,3)$, $(2,3,4)$, and $(2,3,5)$. The  tilings corresponding to these solutions are shown in Fig.~\ref{fig:TriangularSymmetryGroups}. The symmetry group of such tiling is called \emph{triangular symmetry group} \cite[pp.~158]{TriangularSymmetryGroups} and is denoted by $(p,q,r)$.

We can observe these tilings in almost all numerically obtained QRACs discussed in Sect.~\ref{sect:Numerical}. They are formed when the great circles corresponding to measurements partition the Bloch sphere into equal triangles. All such cases are summarized in Table~\ref{tab:TriangularSymmetryGroups}. Tilings appearing in \rac{2} and \rac{4} QRACs that are not mentioned in the table can be seen as degenerate cases.

\begin{table}[!ht]
  \centering
  \begin{tabular}{c|c|l|l}
    $n$ & $(p,q,r)$ & Polyhedron                     & Section and figure                                 \\
    \hline
    $3$ & $(2,2,2)$ & octahedron                     & Sect.~\ref{sect:QRAC2,QRAC3}, Fig.~\ref{fig:QRAC3} \\
    $5$ & $(2,2,4)$ & normalized octagonal dipyramid & Sect.~\ref{sect:QRAC5},       Fig.~\ref{fig:QRAC5} \\
    $6$ & $(2,3,3)$ & normalized tetrakis hexahedron & Sect.~\ref{sect:QRAC6},       Fig.~\ref{fig:QRAC6} \\
    $9$ & $(2,2,2)$ & octahedron                     & Sect.~\ref{sect:QRAC9},       Fig.~\ref{fig:QRAC9} \\
  \end{tabular}
  \caption[Triangular symmetry groups of numerical \rac{n} QRACs]{Triangular symmetry groups of numerical \rac{n} QRACs.}
  \label{tab:TriangularSymmetryGroups}
\end{table}

The tilings corresponding to triangular symmetry groups $(2,3,4)$ and $(2,3,5)$ do not appear in numerically obtained codes. Thus we will use them to construct new (symmetric) \rac{9} and \rac{15} QRACs with SR in Sects.~\ref{sect:QRAC9[sym]} and \ref{sect:QRAC15[sym]}, respectively. To each tiling one can associate a corresponding polyhedron with equal triangular faces. The polyhedra corresponding to tilings $(2,3,4)$ and $(2,3,5)$ are called the normalized\footnote{\emph{Normalized} means that all vectors pointing from the origin to the vertices of the polyhedron are rescaled to have unit norm.} \emph{disdyakis dodecahedron} and the normalized \emph{disdyakis triacontahedron}, respectively.

\vertskip

Polyhedra arising from both types of symmetric great circle arrangements (qua\-si\-re\-gu\-lar polyhedra and triangular symmetry groups) are summarized in Table~\ref{tab:GreatCircleSolids}. The great circle arrangements corresponding to the four marked polyhedra do not appear in numerically obtained codes, so we will use them to construct new (symmetric) QRACs with SR.

\begin{table}[!ht]
  \centering
  \begin{tabular}{r|r|r|c|l}
    $n$ & \multicolumn{2}{|c|}{Faces} & $(p,q,r)$ & Polyhedron \\
    \hline
    $3$ &   $8$ &   $8$ & $(2,2,2)$ & octahedron                              \\
    $4$ &  $14$ &  $14$ &    QR     & cuboctahedron                        \V \\
    $6$ &  $32$ &  $32$ &    QR     & icosidodecahedron                    \V \\
    $6$ &  $24$ &  $32$ & $(2,3,3)$ & normalized tetrakis hexahedron          \\
    $9$ &  $48$ &  $74$ & $(2,3,4)$ & normalized disdyakis dodecahedron    \V \\
   $15$ & $120$ & $212$ & $(2,3,5)$ & normalized disdyakis triacontahedron \V
  \end{tabular}
  \caption[Polyhedra whose edges form great circles]{Polyhedra whose edges form great circles. The first column indicates the number of great circles. The next two indicate, respectively, the number of faces of the polyhedron and the maximal number of pieces achievable by cutting the sphere with $n$ great circles (see Sect.~\ref{sect:noQRAC4}). The fourth column indicates the triangular symmetry group (QR means quasiregular). The name of the polyhedron is given in the last column. Four marked polyhedra will be used in subsequent sections to construct symmetric QRACs with SR.}
  \label{tab:GreatCircleSolids}
\end{table}

\subsubsection{Symmetric \rac{4} QRAC with SR} \label{sect:QRAC4[sym]}

\image{0.4}{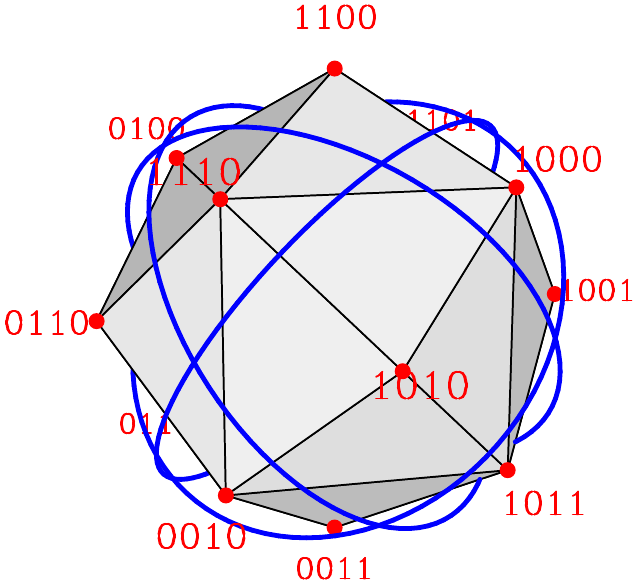}{Symmetric \rac{4} QRAC with SR}{Symmetric \rac{4} QRAC with SR.}{fig:QRAC4[sym]}

\image{0.3}{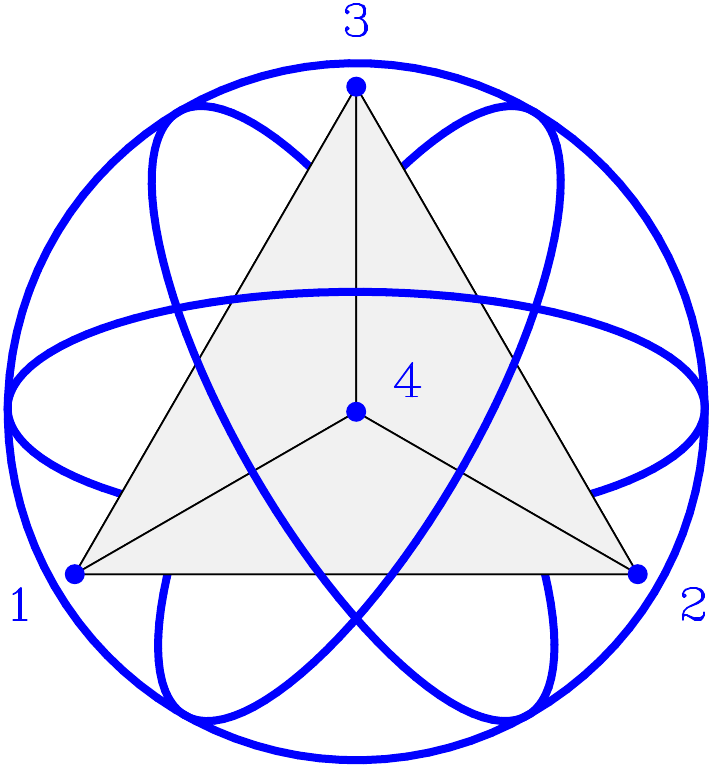}{A regular tetrahedron and four great circles parallel to its faces}{A regular tetrahedron and four great circles parallel to its faces. The circles are determined by the measurements in the direction of the vertices of the tetrahedron. The numbers at the vertices indicate the Bloch vectors of basis states $\ket{\psi_0}$ of the measurements for the \rac{4} QRAC shown in Fig.~\ref{fig:QRAC4[sym]}.}{fig:Tetrahedron}

Recall that in Sect.~\ref{sect:noQRAC4} we proved that four planes passing through the center of the Bloch sphere partition its surface into at most $14$ parts. The most symmetric way to obtain $14$ parts is to use the four planes parallel to the four faces of a regular \emph{tetrahedron}. The measurements are along the four directions given by the vertices (see Fig.~\ref{fig:Tetrahedron}).

The simplest way to construct a regular tetrahedron is to choose four specific vertices of the cube, i.e., from the set $\frac{1}{\sqrt{3}} (\pm 1, \pm 1, \pm 1)$. For example, we could choose the ones with an odd number of positive coordinates. They provide us with the following pairs of antipodal Bloch vectors as the measurement bases:
\begin{equation}
  \begin{aligned}
    \vc{v}_1 &= \pm (+1,-1,-1) / \sqrt{3}, \\
    \vc{v}_2 &= \pm (-1,+1,-1) / \sqrt{3}, \\
    \vc{v}_3 &= \pm (-1,-1,+1) / \sqrt{3}, \\
    \vc{v}_4 &= \pm (+1,+1,+1) / \sqrt{3}.
  \end{aligned}
  \label{eq:SICPOVM}
\end{equation}
The qubit states corresponding to these Bloch vectors are as follows:
\begin{equation}
  \begin{aligned}
    M_1 &= M(+1,+1), \\
    M_2 &= M(+1,-1), \\
    M_3 &= M(-1,+1), \\
    M_4 &= M(-1,-1),
  \end{aligned}
\end{equation}
where
\begin{equation}
  M(s_1,s_2) = \set{
    \frac{1}{2} \sqrt{1+\frac{s_1}{\sqrt{3}}} \mx{ \sqrt{3}-s_1 \\ s_2 (s_1 - i) },
    \frac{1}{2} \sqrt{1-\frac{s_1}{\sqrt{3}}} \mx{ \sqrt{3}+s_1 \\ s_2 (i - s_1) }
  }.
\end{equation}

The great circles determined by these measurements partition the Bloch ball into $14$ parts. In fact, the grid formed by these circles is a projection of the edges of a \emph{cuboctahedron} (see the part on quasireglar polyhedra in Sect.~\ref{sect:GreatCircles}) on the surface of the Bloch ball (see Figs.~\ref{fig:QRAC4[sym]} and \ref{fig:Tetrahedron}). 

In each of the $14$ parts of the Bloch sphere a definite string can be encoded so that each bit can be recovered with a probability greater than $\frac{1}{2}$. Strange as it may seem, the remaining $2$ strings ($x=0000$ and $x=1111$) can be encoded anywhere without affecting the success probability of this QRAC. This is not a surprise if we recall from Sect.~\ref{sect:OptimalEncoding} that the optimal encoding $\vc{r}_x$ of the string $x$ is a unit vector in the direction of $\vc{v}_x$ given by equation (\ref{eq:vx}). In our case the Bloch vectors of the measurement bases point to the vertices of a regular tetrahedron centered at the origin. They clearly sum to zero, so $\vc{v}_{0000} = \vc{v}_{1111} = 0$. Thus the scalar product $\vc{r}_x \cdot \vc{v}_x$ in (\ref{eq:rxvx}) is also zero and the success probability does not depend on the vectors $\vc{r}_{0000}$ and $\vc{r}_{1111}$. Therefore, we will ignore these two strings in the following discussion.

\image{0.3}{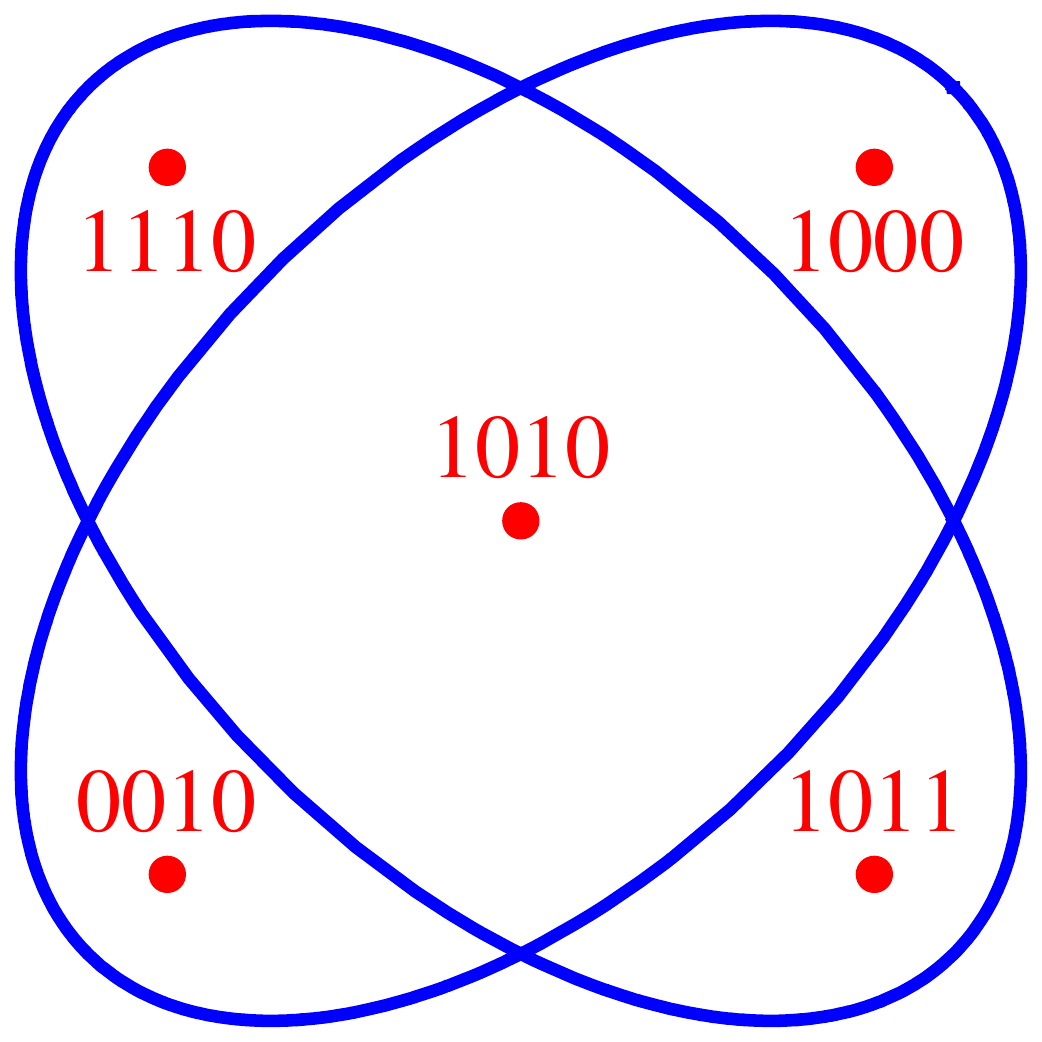}{Strings encoded into the spherical square and the adjacent spherical triangles of the \rac{4} QRAC}{The relationship between the strings encoded into the spherical square and the adjacent spherical triangles according to the \rac{4} QRAC shown in Fig.~\ref{fig:QRAC4[sym]}.}{fig:Ellipses}

The other $14$ strings are encoded into the vertices of a normalized \emph{tetrakis hexahedron} (the \emph{convex hull} of the \emph{cube} and \emph{octahedron}). The string $x = x_1 x_2 x_3 x_4$ is encoded into the Bloch vector $\vc{r}(x) = \vc{r}_{w}(x)$, where
\begin{equation}
  w = x_1 \XOR x_2 \XOR x_3 \XOR x_4 \in \set{0,1}
\end{equation}
is the parity of the input. In the case $w=0$ the encoding points are the vertices $(\pm 1, 0, 0) \cup (0, \pm 1, 0) \cup (0, 0, \pm 1)$ of an \emph{octahedron}:
\begin{equation}
  \vc{r}_0(x) = (-1)^{x_4}
  \mx{
    1 - \abs{x_1 - x_4} \\
    1 - \abs{x_2 - x_4} \\
    1 - \abs{x_3 - x_4}
  }.
\end{equation}
But for $w=1$ we get the vertices $(\pm 1, \pm 1, \pm 1)/\sqrt{3}$ of a \emph{cube}:
\begin{equation}
  \vc{r}_1(x) = \frac{(-1)^{x_1 x_2 + x_3 x_4}}{\sqrt{3}}
  \mx{
    (-1)^{x_1 + x_4} \\
    (-1)^{x_2 + x_4} \\
    (-1)^{x_3 + x_4}
  }.
\end{equation}
Note that the Bloch vectors $\vc{r}_1(x)$ are the vertices of the same cube as the Bloch vectors of the \rac{3} QRAC discussed in Sect.~\ref{sect:KnownQRAC3}.

One can observe the following properties of this encoding. The surface of the Bloch ball is partitioned into $6$ \emph{spherical squares} and $8$ \emph{spherical triangles}. Strings with $w=0$ and $w=1$ are encoded into squares and triangles, respectively. If $w=1$ ($x=1000$ or $x=0111$ and their permutations), the string has one bit that differs from the other three. Such a string is encoded into the basis state of the corresponding measurement so that this bit can be recovered with certainty. If $w=0$, the string is encoded into a square and has the following property: each of its bits takes the value that occurs more frequently at the same position in the strings of the four neighboring triangles (see Fig.~\ref{fig:Ellipses} as an example).

The corresponding encoding function is $E(x)=\alpha_w\ket{0}+\beta_w\ket{1}$ with coefficients $\alpha_0$, $\beta_0$ and $\alpha_1$, $\beta_1$ explicitly given by
\begin{equation}
  \left\{
  \begin{aligned}
    \alpha_0 &= \sqrt{\frac{1}{2} + (-1)^{x_4} \frac{1 - \abs{x_3 - x_4}}{2}}, \\
    \beta_0  &= x_3 x_4 + (-1)^{x_4} \frac{1 - \abs{x_1 - x_4} + i \bigl(1 - \abs{x_2 - x_4}\bigr)}{\sqrt{2}},
  \end{aligned}
  \right.
\end{equation}
and
\begin{equation}
  \left\{
  \begin{aligned}
    \alpha_1 &= \sqrt{\frac{1}{2} + \frac{s(x)}{2 \sqrt{3}}}, \\
    \beta_1  &= (-1)^{x_3} s(x) \frac{(-1)^{x_1} + i (-1)^{x_2}}{\sqrt{6 + s(x) 2 \sqrt{3}}},
  \end{aligned}
  \right.
\end{equation}
where $s(x) \in \set{-1,1}$ is given by
\begin{equation}
  s(x) = (-1)^{x_1 x_2 + x_3 x_4 + x_3 + x_4}.
\end{equation}
The $14$ coefficients $\beta_0$ and $\beta_1$ are the roots of the polynomial
\begin{equation}
  \beta (\beta - 1) (4 \beta^4 - 1) (36 \beta^8 + 24 \beta^4 + 1).
\end{equation}

Using input randomization we get the same success probability for any input:
\begin{equation}
  p = \frac{1}{2} + \frac{2+\sqrt{3}}{16} \approx \p{4[sym]}.
\end{equation}
It is surprising that despite higher symmetry (compare Figs.~\ref{fig:QRAC4} and \ref{fig:QRAC4[sym]}) this QRAC has a lower success probability than the \rac{4} QRAC discussed in Sect.~\ref{sect:QRAC4}.

\subsubsection{Symmetric \rac{6} QRAC with SR} \label{sect:QRAC6[sym]}

\image{0.4}{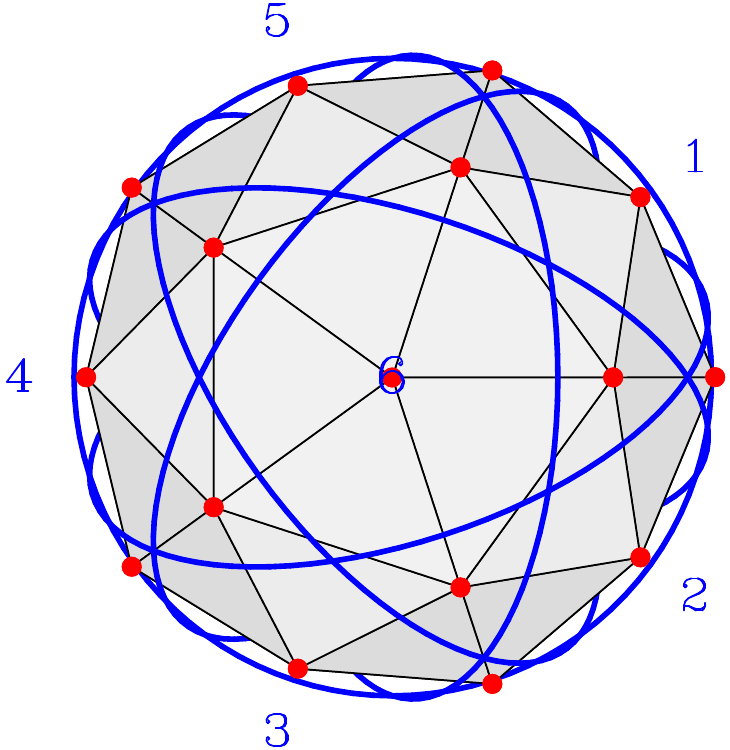}{Symmetric \rac{6} QRAC with SR}{Symmetric \rac{6} QRAC with SR.}{fig:QRAC6[sym]}

According to the discussion in Sect.~\ref{sect:noQRAC4}, six great circles can cut the sphere into at most $32$ parts. It turns out that there is a very symmetric arrangement that achieves this maximum. Observe that the \emph{dodecahedron} has $12$ faces and diametrically opposite ones are parallel. For each pair of parallel faces we can draw a plane through the origin parallel to both faces. These six planes intersect the sphere in six great circles that define our measurements. They are the projections of the edges of the \emph{icosidodecahedron} (see Fig.~\ref{fig:Quasiregular}), which is one of the quasiregular polyhedra discussed in Sect~\ref{sect:GreatCircles}.

There is another way to describe these measurements. Notice that the \emph{icosahedron} (the dual of the dodecahedron) has $12$ vertices that consist of six antipodal pairs. Our measurements are along the six directions defined by these pairs. The coordinates of the vertices of the icosahedron are as follows:
\begin{equation}
  \frac{1}{\sqrt{1+\tau^2}} (0, \pm \tau, \pm 1) \cup
  \frac{1}{\sqrt{1+\tau^2}} (\pm 1, 0, \pm \tau) \cup
  \frac{1}{\sqrt{1+\tau^2}} (\pm \tau, \pm 1, 0),
  \label{eq:Icosahedron}
\end{equation}
where $\tau = \frac{1+\sqrt{5}}{2}$ is the \emph{golden ratio} (the positive root of $x^2 = x + 1$).

Each of the $64$ strings is encoded either in a vertex of an icosahedron or dodecahedron. They have $12$ and $20$ vertices, respectively, so there are two strings encoded in each vertex. The union of the icosahedron and the dodecahedron is called the \emph{pentakis dodecahedron} (see the polyhedron in Fig.~\ref{fig:QRAC6[sym]}).

The success probability of this code is
\begin{equation}
  p = \frac{1}{2} + \frac{\sqrt{5}}{32} + \frac{1}{96} \sqrt{75 + 30 \sqrt{5}} \approx \p{6[sym]}.
\end{equation}

\subsubsection{Symmetric \rac{9} QRAC with SR} \label{sect:QRAC9[sym]}

\doubleimage
{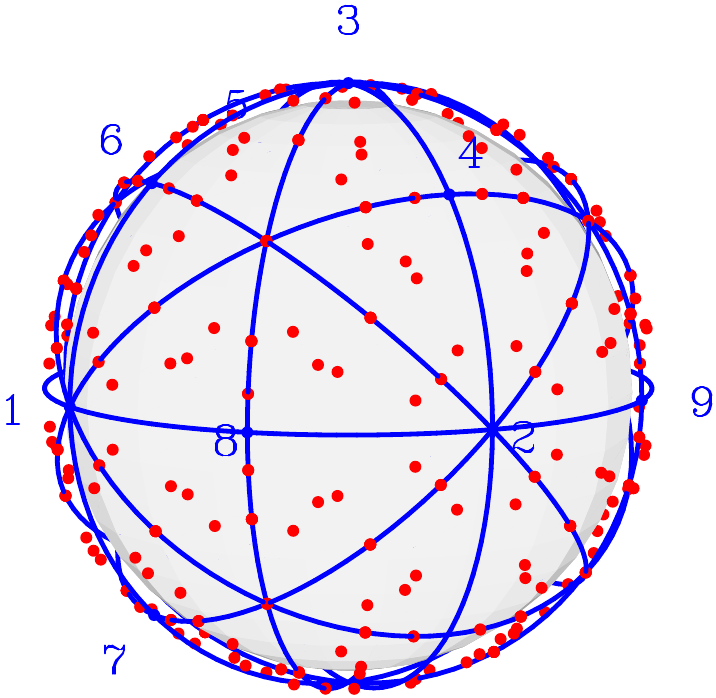}{Symmetric \rac{9} QRAC with SR}{Symmetric \rac{9} QRAC with SR.}{fig:QRAC9[sym]}
{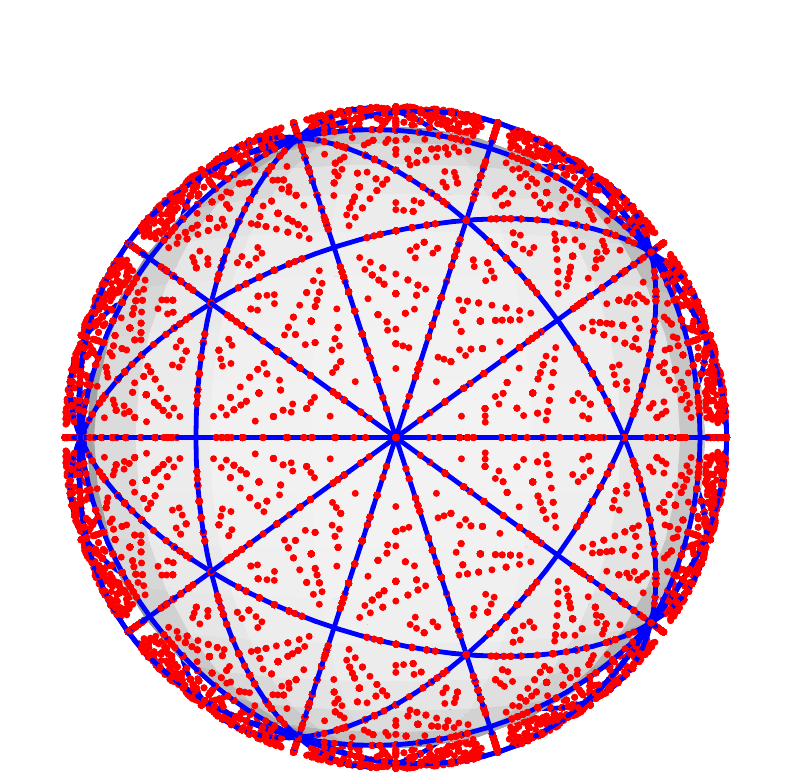}{Symmetric \rac{15} QRAC with SR}{Symmetric \rac{15} QRAC with SR.}{fig:QRAC15[sym]}

This code is based on the triangular tiling of the sphere whose symmetry group is $(2,3,4)$. The great circles corresponding to measurements coincide with the projection of the edges of the \emph{normalized disdyakis dodecahedron}. We can think of this QRAC as the union of \rac{3} and \rac{6} codes. The first three measurements are along the coordinate axis as in the \rac{3} QRAC discussed in Sect.~\ref{sect:KnownQRAC3}. The remaining six measurements are exactly the same as for the \rac{6} code discussed in Sect.~\ref{sect:QRAC6} (see Figs.~\ref{fig:Cube} and \ref{fig:QRAC6}), i.e., they are along the six antipodal pairs of $12$ vertices of the cuboctahedron shown in Fig.~\ref{fig:Quasiregular}. Note that a great circle of the first kind cannot be transformed to a great circle of the second kind via an operation from the symmetry group of the code.\footnote{For the other three symmetric codes we can transform any circle to any other in this way, i.e., the symmetry group acts transitively on the circles.}

The resulting QRAC is shown in Fig.~\ref{fig:QRAC9[sym]} and its success probability is
\begin{equation}
  p \approx \p{9[sym]}.
\end{equation}

\subsubsection{Symmetric \rac{15} QRAC with SR} \label{sect:QRAC15[sym]}

The triangular symmetry group of this code is $(2,3,5)$ and the great circles coincide with the projection of the edges of the \emph{normalized disdyakis triacontahedron}. To understand what the measurements are in this case, note that the \emph{icosidodecahedron} (see Fig.~\ref{fig:Quasiregular}) has $30$ vertices. Their coordinates are:
\begin{gather}
    (\pm 1, 0, 0) \cup
    (0, \pm 1, 0) \cup
    (0, 0, \pm 1), \\
    \frac{1}{2 \tau} (\pm 1, \pm \tau, \pm \tau^2) \cup
    \frac{1}{2 \tau} (\pm \tau^2, \pm 1, \pm \tau) \cup
    \frac{1}{2 \tau} (\pm \tau, \pm \tau^2, \pm 1).
  \label{eq:Icosidodecahedron}
\end{gather}
The measurement directions are given by $15$ antipodal pairs of these vertices.

The obtained QRAC is shown in Fig.~\ref{fig:QRAC15[sym]}. Its success probability is
\begin{equation}
  p \approx \p{15[sym]}.
\end{equation}

\subsection{Discussion} \label{sect:Discussion}

In this section we will compare and analyze the numerical and symmetric QRACs with SR described in Sects.~\ref{sect:Numerical} and \ref{sect:Symmetric}, respectively. Hopefully these observations can be used to find new \rac{n} QRACs with SR or to generalize the existing ones (see Sect.~\ref{sect:Generalizations} for possible generalizations).


The success probabilities of numerical and symmetric QRACs with SR are given in Tables~\ref{tab:NumericalProbabilities} and \ref{tab:SymmetricProbabilities}, respectively (see Table~\ref{tab:ComparisonOfProbabilities} for the comparison). We see that none of the symmetric codes discussed in Sect.~\ref{sect:Symmetric} is optimal. However, the success probabilities of numerical and symmetric codes do not differ much. Moreover, recall that there are two more symmetric codes (\rac{3} and \rac{6}) that coincide with the numerically obtained ones (see Table~\ref{tab:GreatCircleSolids}). Concerning these two codes we can reach more optimistic conclusions: the \rac{3} QRAC is optimal (see Sect.~\ref{sect:UpperBound}) and possibly the \rac{6} QRAC (see Sect.~\ref{sect:QRAC6}) is as well, since we did not manage to improve it in Sect.~\ref{sect:QRAC6[sym]}.

\begin{table}[!ht]
  \centering
  \begin{tabular}{r|c|r}
    $n$                   & Section                & \multicolumn{1}{c}{Probability} \\
    \hline
    \multirow{2}{*}{$4$}  & \ref{sect:QRAC4}       &  $\p{4}$ \\
                          & \ref{sect:QRAC4[sym]}  & $>\p{4[sym]}$ \\
    \hline
    \multirow{2}{*}{$6$}  & \ref{sect:QRAC6}       &  $\p{6}$ \\
                          & \ref{sect:QRAC6[sym]}  & $>\p{6[sym]}$ \\
    \hline
    \multirow{2}{*}{$9$}  & \ref{sect:QRAC9}       &  $\p{9}$ \\
                          & \ref{sect:QRAC9[sym]}  & $>\p{9[sym]}$ \\
    \hline
    \multirow{2}{*}{$15$} &                        &  $\p{15}$ \\
                          & \ref{sect:QRAC15[sym]} & $>\p{15[sym]}$
  \end{tabular}
  \caption[Comparison of the success probabilities of \rac{n} QRACs with SR]{Comparison of the success probabilities of \rac{n} QRACs with SR. For each $n$ the first probability corresponds to a numerical code, but the second one to a symmetric code. For $n=15$ we do not have numerical results, so we just use five measurements along each coordinate. In fact, the numerical \rac{4} and \rac{9} QRACs also use measurements only along coordinate axis. The \rac{6} QRAC with two measurements along each coordinate axis has success probability $\p{6[ort]}$.}
  \label{tab:ComparisonOfProbabilities}
\end{table}


We just saw that symmetric QRACs are not necessarily optimal. One could ask if there are other heuristic methods that potentially could be used to construct good QRACs with SR. We will give a few speculations in the remainder of this section. In particular, we will discuss some special kinds of measurements that could be useful. To make the discussion more general, we will not restrict ourselves to the case of a single qubit.


\begin{definition}
Two orthonormal bases $\base_1$ and $\base_2$ of $\C^d$ are called \emph{mutually unbiased bases} (MUBs) if $\abs{\braket{\psi_1}{\psi_2}}^2 = \frac{1}{d}$ for all $\ket{\psi_1} \in \base_1$ and $\ket{\psi_2} \in \base_2$. The maximal number of pairwise mutually unbiased bases in $\C^d$ is $d+1$. \cite{MUBs}
\end{definition}

When $d=2$, equation (\ref{eq:ScalarProduct}) implies that Bloch vectors corresponding to basis vectors of \emph{different} mutually unbiased bases are orthogonal\footnote{The notion of the Bloch vector can be generalized for $d \geq 2$ (see \cite{Kimura}). Then a similar duality holds as well (see equation (\ref{eq:GeneralizedScalarProduct}) in Sect.~\ref{sect:Generalizations}): mutually unbiased quantum states correspond to orthogonal Bloch vectors, but orthogonal quantum states correspond to ``mutually unbiased'' Bloch vectors, i.e., equiangular vectors pointing to the vertices of a regular simplex.}. There are three such bases in $\C^2$ and their Bloch vectors correspond to the vertices of an octahedron. For example, the bases $M_1$, $M_2$, and $M_3$ defined in Sects.~\ref{sect:KnownQRAC2} and \ref{sect:KnownQRAC3} are MUBs (they correspond to measuring along $x$, $y$, and $z$ axis).

Note that the measurements for numerical \rac{2}, \rac{3}, \rac{4}, and \rac{9} QRACs are performed entirely using MUBs and three out of five measurement bases for numerical \rac{5} QRAC are also MUBs.


There is another very special measurement that appears in our QRACs.

\begin{definition}
A set of $d^2$ unit vectors $\ket{\psi_i} \in \C^d$ is called \emph{symmetric, informationally complete POVM} (\mbox{SIC-POVM}) if $\abs{\braket{\psi_i}{\psi_j}}^2 = \frac{1}{d+1}$ for any $i,j$. \cite{SIC-POVMs}
\end{definition}

For $d=2$ there are four such quantum states. Again, from equation (\ref{eq:ScalarProduct}) we see that the inner product between any two Bloch vectors corresponding to these states is $-\frac{1}{3}$. Such equiangular Bloch vectors are exactly the vertices of a tetrahedron, e.g., $\vc{v}_1$, $\vc{v}_2$, $\vc{v}_3$, $\vc{v}_4$ defined in (\ref{eq:SICPOVM}). They were used in Sect.~\ref{sect:QRAC4[sym]} to construct a symmetric \rac{4} QRAC.


Let us compare numerical and symmetric \rac{4} QRACs from Sects.~\ref{sect:QRAC4} and \ref{sect:QRAC4[sym]}, respectively. The first one is based on MUBs and is not very symmetric. Moreover, it looks like we are wasting one out of four bits, since two measurements are along the same direction. However, all measurement directions in the Bloch sphere are mutually orthogonal, except the ones that coincide. The second \rac{4} code is based on a \mbox{SIC-POVM} and is very symmetric. However, it appears that in this case we are wasting two out of $16$ strings, since the way we encode them does not influence the success probability.

Now, if we compare the success probabilities of both \rac{4} codes (see Table~\ref{tab:ComparisonOfProbabilities}), we see that the first one is clearly better. Hence we conclude that
\begin{center}
  \emph{orthogonality} of the measurement Bloch vectors \\
  seems to be more important than \emph{symmetry}.
\end{center}
One can come to a similar conclusion when comparing \rac{9} and \rac{15} codes. Thus it looks like using roughly $n/3$ measurements along each coordinate axis is quite a good heuristic for constructing \rac{n} QRAC with SR (see Sect.~\ref{sect:Lower2}).

\section{Conclusion} \label{sect:Conclusion}

\subsection{Summary} \label{sect:Summary}

We study the \emph{worst} case success probability of random access codes with shared randomness. Yao's principle (see equation (\ref{eq:Yao}) in Sect.~\ref{sect:Yao}) and input randomization (see Theorem~\ref{thm:SRtoPure}) is applied to consider the \emph{average} case success probability instead (this works in both classical and quantum cases).

In Sect.~\ref{sect:ClassicalBound} we construct an optimal \emph{classical} \rac{n} RAC with SR as follows (see Theorem~\ref{thm:OptimalClassical}): Alice XORs the input string with $n$ random bits she shares with Bob, computes the majority and sends it to Bob; if the \mbox{$i$th} bit is requested, Bob outputs the \mbox{$i$th} bit of the shared random string XORed with the received bit. The asymptotic success probability of this code is given by equation (\ref{eq:Approx}) in Sect.~\ref{sect:ClassicalBound}:
\begin{equation}
  p(n) \approx \frac{1}{2} + \frac{1}{\sqrt{2 \pi n}}.
\end{equation}

The worst case success probability of an optimal \emph{quantum} RAC with SR satisfies the following inequalities:
\begin{equation}
  \frac{1}{2} + \sqrt{\frac{2}{3 \pi n}} \leq p(n) \leq \frac{1}{2} + \frac{1}{2 \sqrt{n}}.
\end{equation}
These upper and lower bounds are obtained in Sects.~\ref{sect:UpperBound} and \ref{sect:Lower1}, respectively.

The success probabilities of classical and quantum RACs with SR are compared in Fig.~\ref{fig:Comparison}.

\image{0.9}{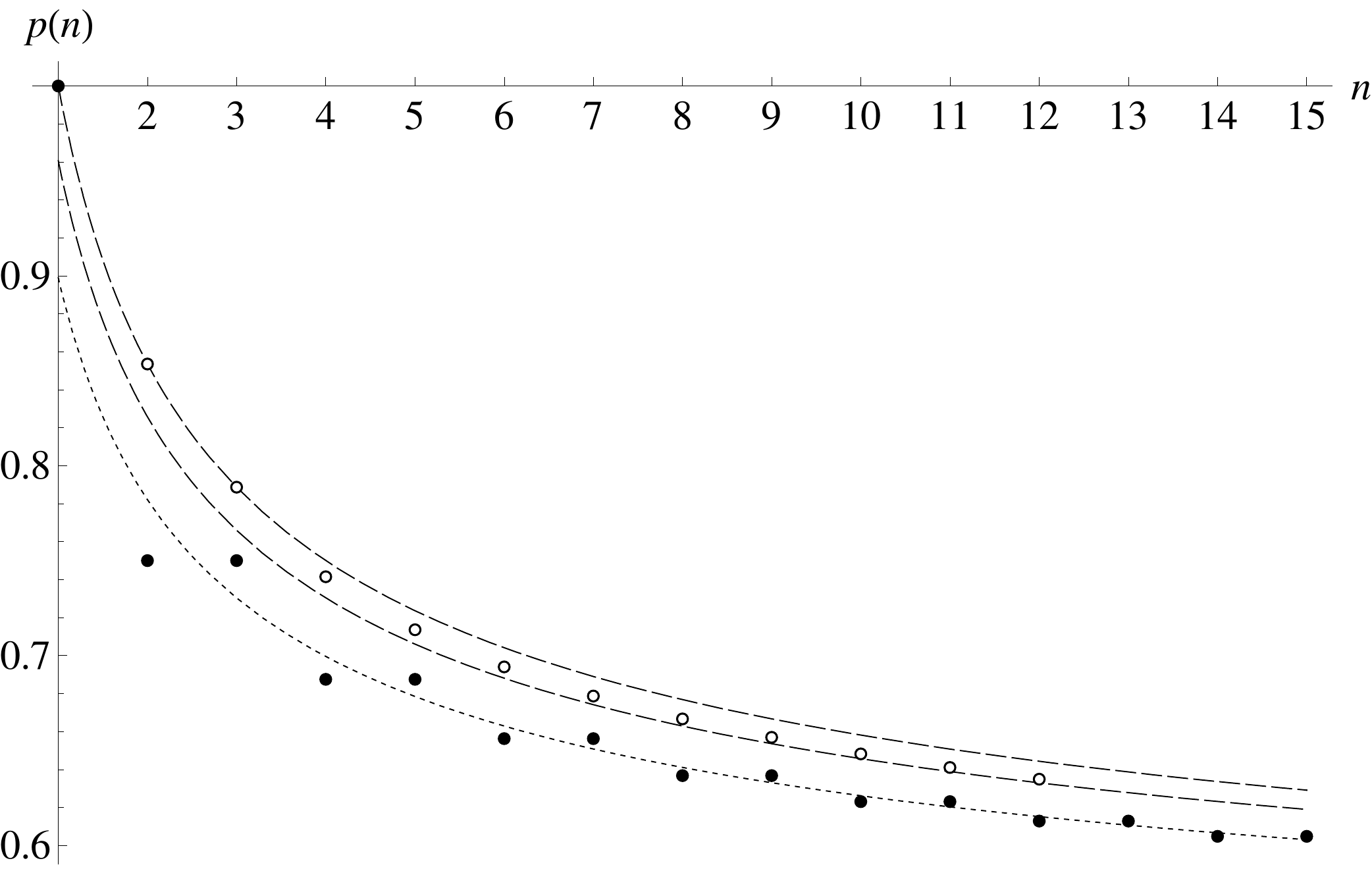}{Comparison of success probabilities of classical and quantum RACs}{Comparison of success probabilities of classical and quantum RACs. Black dots correspond to optimal classical RACs and the dotted line shows the asymptotic behavior. Circles correspond to numerical QRACs and dashed lines to quantum upper and lower bounds, respectively.}{fig:Comparison}

\subsection{Open problems for \rac{n} QRACs} \label{sect:OpenProblems}

\emph{Lower bound by orthogonal measurements.} The known \rac{2} and \rac{3} QRACs (see Sect.~\ref{sect:KnownQRACs}) and our numerical \rac{4} and \rac{9} QRACs with SR (see Sects.~\ref{sect:QRAC4} and \ref{sect:QRAC9}) suggest that MUBs can be used to obtain good QRACs (see Sect.~\ref{sect:Discussion}). Indeed, \rac{n} QRAC with orthogonal measurements (see Sect.~\ref{sect:Lower2}) is better than the one with random measurements (see Sect.~\ref{sect:Lower1}). However, we were not able to obtain an asymptotic expression for its success probability. This is equivalent to obtaining an asymptotic expression for (\ref{eq:CubicDistance}), i.e., the average distance traveled by a random walk with roughly $n/3$ steps along each coordinate axis.

In Fig.~\ref{fig:Optimality} we show how close both lower bounds and the success probabilities of numerical QRACs are relative to the upper bound from Sect.~\ref{sect:UpperBound}. Assume that Alice and Bob are given a point in the light gray region in Fig.~\ref{fig:Optimality} and asked to construct a QRAC with SR whose success probability is at least as good. Then they can use measurements along coordinate axis as in Sect.~\ref{sect:Lower2}. If the point is in the dark gray region, they can use one of the numerical codes from Sect.~\ref{sect:Numerical}. However, if it is in the white region, they have to solve the next open problem.

\emph{Optimality of numerical codes.} Prove the optimality of any of the numerically obtained \rac{n} QRACs with SR for $n \geq 4$ discussed in Sect.~\ref{sect:Numerical}. Are the optimal constructions unique (up to isomorphism)?

\emph{Prove the ``Homer conjecture''} that quantum RACs with SR are at least as good as their classical counterparts in the sense discussed at the end of Sect.~\ref{sect:OptimalEncoding}.

\image{1.0}{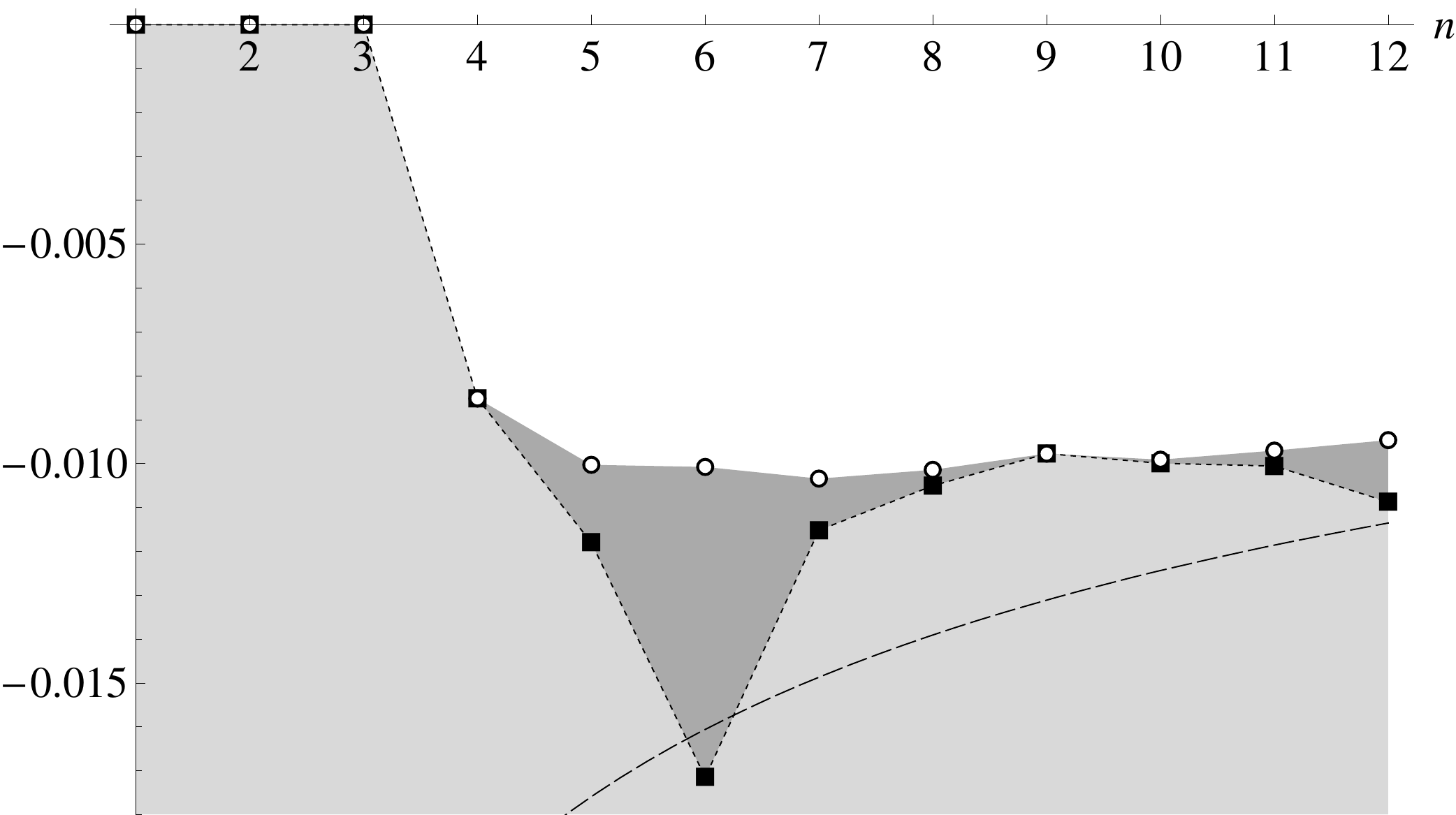}{\mbox{Close-up} of the region between the quantum upper and lower bound}{\mbox{Close-up} of the narrow region in Fig.~\ref{fig:Comparison} between the quantum upper and lower bound (everything is shown relative to the upper bound that corresponds to the horizontal axis). Circles indicate the gap between the upper bound and numerical QRACs with SR. Black squares show the gap between the upper bound and the lower bound by measurements along coordinate axes (see Fig.~\ref{fig:OrthBound}). Dashed line corresponds to the gap between the quantum upper bound and the lower bound by random measurements.}{fig:Optimality}

\subsection{Possible generalizations} \label{sect:Generalizations}

There are several ways that random access codes with SR can be generalized, both classically and quantumly. In particular, one can consider
\begin{itemize}
  \item \racp{n} codes in base $d$, $d > 2$ (called \emph{qudits} in the quantum case),
  \item \racm{n}{m} codes with $m > 1$,
  \item \racm{n}{m} codes where any $k > 1$ bits (qubits) must be recovered.
\end{itemize}
Of course, one can consider several of these generalizations simultaneously. In the setting without shared randomness such generalizations have already appeared in the literature (see Sect.~\ref{sect:History}). We will briefly introduce the notion of the generalized Bloch vector which we believe can be useful to study such generalizations (it has been explicitly used in \cite{No41} to prove the impossibility of \racm{2^m}{m} QRAC with $p>1/2$, when SR is not allowed).

The notion of the Bloch vector introduced in Sect.~\ref{sect:BlochSphere} can be generalized for $d>2$. For example, to write down the density matrix for $d=3$ one uses eight \emph{\mbox{Gell-Mann} matrices} denoted by $\lambda_i$ instead of three Pauli matrices $\sigma_i$ defined in equation (\ref{eq:Pauli}). In general one needs $d^2-1$ matrices $\lambda_i$ that span the set of all traceless $d \times d$ Hermitian matrices. A convenient choice of $\lambda_i$ are the so called \emph{generalized Gell-Mann matrices}, also known as the \emph{generators of the Lie algebra of $SU(d)$}, given in \cite{SUGenerators}. We can use them to generalize equation (\ref{eq:Rho}):
\begin{equation}
  \rho = \frac{1}{d} \left( I + \sqrt{\frac{d(d-1)}{2}} \; \vc{r} \cdot \vc{\lambda} \right),
  \label{eq:GeneralRho}
\end{equation}
where $\vc{\lambda} = (\lambda_1, \dotsc, \lambda_{d^2-1})$ and $\vc{r} \in \R^{d^2-1}$ is the \emph{generalized Bloch vector}\footnote{Our normalization follows \cite{Positivity}, where the generalized Bloch sphere has radius $1$. Another widely used convention is to assume radius $\sqrt{2(d-1)/d}$, e.g., see \cite{Kimura, Kimura2}.\label{foo:Norm}} or \emph{coherence vector} \cite{Kimura, Positivity}. Since the $\lambda_i$ are chosen so that $\tr \lambda_i = 0$ and $\tr (\lambda_i \lambda_j) = 2 \delta_{ij}$, equation (\ref{eq:ScalarProduct}) generalizes to
\begin{equation}
  \abs{\braket{\psi_1}{\psi_2}}^2 = \tr(\rho_1 \rho_2) =
  \frac{1}{d} \bigl( 1 + (d-1) \; \vc{r}_1 \cdot \vc{r}_2 \bigr).
  \label{eq:GeneralizedScalarProduct}
\end{equation}

To recover a base $d$ digit, we perform a measurement in an orthonormal basis $\set{\ket{\psi_1}, \dotsc, \ket{\psi_d}}$ of $\C^d$. Since $\abs{\braket{\psi_i}{\psi_j}}^2 = 0$ for any pair $i \neq j$, the corresponding Bloch vectors must satisfy $\vc{r}_i \cdot \vc{r}_j = -\frac{1}{d-1}$. This means that they are the vertices of a regular simplex that belongs to a \mbox{$(d-1)$-dimensional} subspace and is centered at the origin (for $d=2$ this is just a line segment).

On the other hand, in Sect.~\ref{sect:Discussion} we observed that it might be advantageous to perform measurements along orthogonal directions in the Bloch sphere to recover different bits. Let $\vc{r_i} \perp \vc{s_j}$ be two orthogonal Bloch vectors. Then the corresponding quantum states $\ket{\psi_i}$ and $\ket{\varphi_j}$ must satisfy $\abs{\braket{\psi_i}{\varphi_j}}^2 = \frac{1}{d}$. This is exactly the case when $\ket{\psi_i}$ and $\ket{\varphi_j}$ belong to \emph{different} mutually unbiased bases (see Sect.~\ref{sect:Discussion}). This suggests that distinct bits should be recovered using mutually unbiased measurements. Note that the Bloch vectors of the states from two MUBs correspond to the vertices of two regular simplices in mutually orthogonal subspaces. In general, the Bloch vectors of the states from all $d+1$ MUBs are the vectices of the so called \emph{complementarity polytope} \cite{ComplementarityPolytope}, which is just the octahedron when $d=2$.

The conclusion of Sect.~\ref{sect:Discussion} and our discussion above suggests the use of MUBs to construct QRACs also for $d>2$. Such attempts have already been made \cite{Galvao,Severini}. Galv\~ao \cite{Galvao} gives an example of \racx{2}{0.79} QRAC for qutrits ($d=3$) and Casaccino et al. \cite{Severini} numerically investigate \rac{(d+1)} QRACs based on MUBs for \mbox{$d$-level} quantum systems. However, there is a significant difference between the qubit and qudit case. Recall that for $d=2$ the optimal way to encode the message $x$ is to use a unit vector in the direction of $\vc{v}_x$ (see equation (\ref{eq:vx}) in Sect.~\ref{sect:OptimalEncoding}). A similar expression for $\vc{v}_x$ can be obtained when $d>2$, but then the matrix $\rho$ assigned to $\vc{r} = \vc{v}_x/\norm{\vc{v}_x}$ according to equation (\ref{eq:GeneralRho}) is not necessarily \emph{positive semidefinite} and hence may not be a valid density matrix. However, it is known that for small enough values of $\norm{\vc{r}}$ (in our \mbox{case$^\text{\ref{foo:Norm}}$} $\norm{\vc{r}} \leq \frac{1}{d-1}$), \emph{all} Bloch vectors correspond to valid density matrices \cite{Kimura2}. Hence, if we cannot use the pure state corresponding to $\vc{v}_x/\norm{\vc{v}_x}$, we can always use the mixed state corresponding to $\frac{1}{d-1} \vc{v}_x/\norm{\vc{v}_x}$. If one knows more about the shape of the region corresponding to valid quantum states, one can make a better choice and use a longer vector, possibly in a slightly different direction. Unfortunately, apart from convexity, not much is known about this shape. Already for $d=3$ it is rather involved \cite{Kimura, Kimura2}. In general the conditions (in terms of the coordinates of the generalized Bloch vector $\vc{r}$) for $\rho$ to have \mbox{non-negative} eigenvalues are given in \cite{Positivity, Kimura}.

However, for proving only an upper bound, one can ignore all such details. Thus we believe it might be possible to generalize our upper bound (see Sect.~\ref{sect:UpperBound} and \ref{sect:GeneralUpperBound}) using generalized Bloch vectors. It would be interesting to compare such a result with the upper bound (\ref{eq:HypercontractiveUpperBound}) that was obtained by \mbox{Ben-Aroya} et al. in \cite{Hypercontractive}.

Finally, another way of generalizing QRACs with SR is to add other resources. A good candidate is \emph{shared entanglement}.

\subsection{Acknowledgments}

Most of these results were obtained in the summer of 2006 during the undergraduate summer research program at the University of Waterloo. We would like to thank University of Waterloo and IQC for their hospitality. We would also like to thank Andrew Childs and Ashwin Nayak for many useful comments on the preliminary version of this manuscript. In particular, the simplification of the proof of Lemma~\ref{lm:Sv} is due to Ashwin.

         \appendix

\section{Combinatorial interpretation of sums} \label{app:MagicFormulas}

In this appendix we give a combinatorial interpretation of the sums in equations (\ref{eq:Odd}) and (\ref{eq:Even}) from Sect.~\ref{sect:ClassicalBound}. This interpretation is formalized in the form of equations (\ref{eq:Magic1}) and (\ref{eq:Magic2}). We referred to these equations in Sect.~\ref{sect:ClassicalBound} to obtain an exact formula (\ref{eq:ApproxEvenOdd}) for the average success probability of an optimal classical RAC.

Let us consider a set of $n$ distinct elements and count \emph{the number of ways to choose more than half of $n$ elements and mark one of them as special}. There are two approaches: first choose the elements and then mark the special one or first choose the special one and then choose the others.

In the first scenario there are $i\binom{n}{i}$ ways to choose exactly $i$ elements and mark one of them as special. If we have to choose more than half, we obtain the sum $\sum_{i=m+1}^n i\binom{n}{i}$ where $m=\floor{\frac{n}{2}}$.

In the second scenario there are $n$ ways to choose the special element. Then there are $l=n-1$ elements left and at least $m$ of them must be taken to have more than half of $n$ elements in total. The number of ways to do this corresponds to the number of subsets of size at least $m$ of a set of $l$ distinct elements. Let us consider the cases when $l$ is odd and even separately.

If $n=2m$ then $l=2m-1$ is odd. To each ``large'' subset of size $i$ \mbox{($m \leq i \leq l$)} we can assign a unique ``small'' subset (the complement set) of size $l-i$ ($0 \leq l-i \leq m-1$), and vice versa. Each subset is either ``large'' or ``small'', so the number of ``large'' and ``small'' subsets is the same---it is half of the number of all subsets, i.e., $2^l/2 = 2^{2m-2}$.

If $n=2m+1$ then $l=2m$ is even. The ``large'' subsets have $m+1 \leq i \leq l$ elements, but the ``small'' ones: $0 \leq l-i \leq m-1$. Let us call the remaining $\binom{2m}{m}$ subsets of size $m$ ``balanced''. In this case the bijection between the ``large'' and ``small'' subsets holds as well, but it maps the ``balanced'' subsets to themselves. Thus the total number of all subsets is $\text{``large''} + \text{``small''} + \binom{2m}{m} = 2^l$. The number of ``large'' subsets is $\left(2^{l} + \binom{2m}{m}\right)/2 = 2^{2m-1} + \frac{1}{2} \binom{2m}{m}$.

Both counting methods must give the same results, so for odd and even $n$ we obtain, respectively:
\begin{align}
  &\sum_{i=m+1}^{2m+1} i \binom{2m+1}{i} = (2m+1) \cdot \left(2^{2m-1} + \frac{1}{2} \binom{2m}{m}\right)
  \label{eq:Magic1}, \\
  &\sum_{i=m+1}^{2m} i \binom{2m}{i} = 2m \cdot 2^{2m-2}
  \label{eq:Magic2}.
\end{align}
We would like to acknowledge Juris Smotrovs for providing this interpretation.

\section{POVMs versus orthogonal measurements} \label{app:POVMs}

An orthogonal (or \emph{von Neumann}) measurement is not the most general type of measurement of a quantum system. In general a POVM measurement \cite{NielsenChuang,Peres} may extract more information. In this appendix we show that in the qubit case POVMs can be simulated using a probabilistic combination of \emph{enhanced orthogonal measurements}, as defined in Sect.~\ref{sect:GeneralUpperBound} (such a measurement is either an orthogonal measurement or a constant function). To define a POVM we have to introduce the notion of a positive semidefinite matrix \cite{MatrixAnalysis}.

\begin{definition}
A complex square matrix $E$ is called \emph{positive semidefinite} if $\bra{\psi}E\ket{\psi} \geq 0$ for all $\ket{\psi}$. 
\end{definition}

An equivalent definition is that $E$ is diagonalizable and all eigenvalues of $E$ are real and non-negative. Thus $E$ is Hermitian.

\begin{definition}
A \emph{positive operator-valued measure} (POVM) is a set $\set{E_1,\dotsc,E_m}$ of positive semidefinite matrices such that $\sum_{i=1}^m E_i = I$. \cite{NielsenChuang,Peres}
\end{definition}

POVM measurements can have an arbitrary number of outcomes, but in the case of \rac{n} QRACs we have to consider only two-outcome single-qubit POVMs. Such a POVM can be specified by $\set{E_0,E_1}$, where $E_0$ is positive semidefinite and $E_1 = I - E_0$. Since $E_0$ is also Hermitian, we can find a basis $\base = \set{\ket{\psi_0}, \ket{\psi_1}}$ in which $E_0$ is diagonal, i.e., $E_0 = \smx{a&0\\0&b}$. In this basis $E_1 = \smx{1-a&0\\0&1-b}$. Since both $E_0$ and $E_1$ are positive semidefinite, $0 \leq a \leq 1$ and $0 \leq b \leq 1$. An arbitrary pure qubit state $\ket{\psi}$ in the basis $\base$ can be specified by (\ref{eq:Psi}). When $\ket{\psi}$ is measured, the probabilities of the outcomes are
\begin{equation}
  \left\{
  \begin{aligned}
    P_0 &= \bra{\psi} E_0 \ket{\psi} =   a   \cos^2 \frac{\theta}{2}  +   b   \sin^2 \frac{\theta}{2}, \\
    P_1 &= \bra{\psi} E_1 \ket{\psi} = (1-a) \cos^2 \frac{\theta}{2}  + (1-b) \sin^2 \frac{\theta}{2}.
  \end{aligned}
  \right.
  \label{eq:POVM}
\end{equation}

\image{0.8}{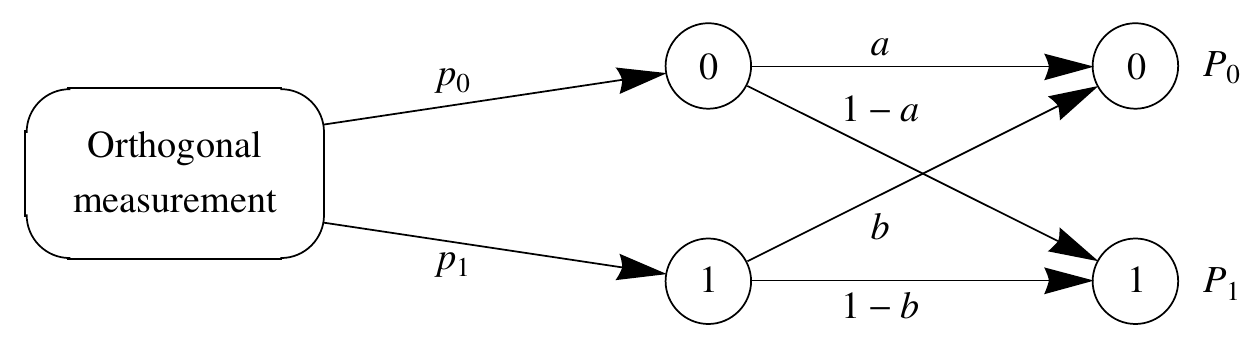}{A simulation of the POVM measurement on a qubit using an orthogonal measurement and a \mbox{post-processing}}{A simulation of the POVM measurement $\set{E_0,E_1}$ on a qubit using an orthogonal measurement and a \mbox{post-processing} of the measurement result.}{fig:Simulation}

Let us consider the following process (see Fig.~\ref{fig:Simulation}) that simulates the POVM measurement $\set{E_0,E_1}$:
\begin{enumerate}
  \item perform an orthogonal measurement in the basis $\base = \set{\ket{\psi_0},\ket{\psi_1}}$,
  \item perform the following \mbox{post-processing} of the outcome of the measurement:
    \begin{itemize}
      \item \emph{if the outcome was $0$}: output $0$ with probability $a$ and output $1$ with probability $1-a$,
      \item \emph{if the outcome was $1$}: output $0$ with probability $b$ and output $1$ with probability $1-b$.
    \end{itemize}
\end{enumerate}
To see why this process is equivalent to the POVM measurement $\set{E_0,E_1}$, consider a pure qubit state $\ket{\psi}$ given by (\ref{eq:Psi}) in the basis $\base$. When $\ket{\psi}$ is measured in the basis $\base = \set{\ket{\psi_0},\ket{\psi_1}}$, the probabilities of the outcomes $0$ and $1$ are as follows (see also equation (\ref{eq:Projections}) in Sect.~\ref{sect:BlochSphere}):
\begin{equation}
  \left\{
  \begin{aligned}
    p_0 &= \abs{\braket{\psi_0}{\psi}}^2 = \cos^2{\frac{\theta}{2}}, \\
    p_1 &= \abs{\braket{\psi_1}{\psi}}^2 = \sin^2{\frac{\theta}{2}}.
  \end{aligned}
  \right.
\end{equation}
Now it is simple to verify that the process shown in Fig.~\ref{fig:Simulation} has the same outcome probabilities (\ref{eq:POVM}) as the POVM measurement. However, this process is not a probabilistic combination of enhanced orthogonal measurements, since it involves a probabilistic \mbox{post-processing} of the measurement result. To obtain the desired result, we have to modify it. The key idea is that with a certain probability the output can be produced without performing an actual measurement.

Let $\mu = \min \set{a,b}$. Whatever state is input to the process shown in Fig.~\ref{fig:Simulation}, the probability $P_0$ to output $0$ is at least $\mu$, because
\begin{equation}
  P_0 = p_0 a + p_1 b \geq (p_0 + p_1) \mu = \mu.
  \label{eq:IneqP0}
\end{equation}
Note that $\mu$ does not depend on the state being measured. This means that one can output $0$ with probability $\mu$ without performing an actual measurement. A similar lower bound holds for $P_1$ as well:
\begin{equation}
  P_1 = p_0 (1 - a) + p_1 (1 - b) \geq (p_0 + p_1) (1 - M) = 1 - M,
  \label{eq:IneqP1}
\end{equation}
where $M = \max \set{a,b} = a + b - \mu$. Let us consider the following probabilistic combination of four decoding strategies:
\begin{itemize}
  \item \emph{with probability $c_{0}$}: output $0$ without performing a measurement,
  \item \emph{with probability $c_{1}$}: output $1$ without performing a measurement,
  \item \emph{with probability $c_{01}$}: measure in the basis $\set{\ket{\psi_0}, \ket{\psi_1}}$,
  \item \emph{with probability $c_{10}$}: measure in the opposite basis $\set{\ket{\psi_1}, \ket{\psi_0}}$.
\end{itemize}
The resulting probabilities of outcomes for this process are
\begin{equation}
  \left\{
  \begin{aligned}
    P_0 &= c_0 + c_{01} p_0 + c_{10} p_1, \\
    P_1 &= c_1 + c_{01} p_1 + c_{10} p_0.
  \end{aligned}
  \right.
  \label{eq:ProbabilisticCombination}
\end{equation}
We can use the lower bounds (\ref{eq:IneqP0}) and (\ref{eq:IneqP1}) for $P_0$ and $P_1$, respectively, to assign the probabilities $c_0$, $c_1$, $c_{01}$, and $c_{10}$ in the following way:
\begin{equation}
  \left\{
  \begin{aligned}
    c_0    &= \mu, \\
    c_1    &= 1 - (a + b) + \mu, \\
    c_{01} &= a - \mu, \\
    c_{10} &= b - \mu
  \end{aligned}
  \right.
  \label{eq:Constants}
\end{equation}
(note that at least one of the probabilities $c_{01}$ or $c_{10}$ will be zero). It is not hard to verify that after the assignment (\ref{eq:Constants}) the probabilities $P_0$ and $P_1$ in (\ref{eq:ProbabilisticCombination}) will match the probabilities of outcomes (\ref{eq:POVM}) of the POVM measurement.

Thus for each qubit POVM given by $a$ and $b$ one can find a probabilistic combination of enhanced orthogonal measurements given by $c_0$, $c_1$, $c_{01}$, and $c_{10}$, such that in both cases the probabilities of outcomes are the same.

\begin{example}
For $a = b = 1/2$ we have $c_0 = c_1 = 1/2$ and $c_{01} = c_{10} = 0$, corresponding to random guessing (observe that $E_0 = E_1$ in this case).
\end{example}

\begin{example}
However, $a = 1$ and $b = 0$ corresponds to a projective measurement in basis $\set{\ket{\psi_0}, \ket{\psi_1}}$, because $c_{01} = 1$ and $c_{10} = c_0 = c_1 = 0$.
\end{example}

\begin{example}
Finally, $a = 1$ and $b = 1$ corresponds to a constant function $0$, because $c_0 = 1$ and $c_{01} = c_{10} = c_1 = 0$.
\end{example}


\end{document}

%% file: FrontPagePaper.tex
\begin{titlepage}
\thispagestyle{plain}
\newcommand{\aff}[2]{\mbox{#1$^\text{#2}$}}

\begin{center}
  \vspace*{0.5in}
  {\LARGE Quantum Random Access Codes\\with Shared Randomness}\\[0.5cm]
  \vspace*{0.2in}

  {\large
    \aff{Andris Ambainis}{\hyperlink{Riga}{a}},
    \aff{Debbie Leung}{\hyperlink{Waterloo}{b}},\\[0.1cm]
    \aff{Laura Mancinska}{\hyperlink{Riga}{a},\hyperlink{Waterloo}{b}},
    \aff{Maris Ozols}{\hyperlink{Riga}{a},\hyperlink{Waterloo}{b}}}\\[0.5cm]
  \hypertarget{Riga}{$^a$}
    \textit{Department of Computing, University of Latvia,\\
    Raina bulv. 19, Riga, LV-1586, Latvia}\\[0.2cm]
  \hypertarget{Waterloo}{$^b$}
    \textit{Department of Combinatorics and Optimization, and Institute for Quantum\\ Computing,
    University of Waterloo, Waterloo, Ontario, N2L 3G1, Canada}\\[0.5cm]
  {\large June 12, 2009}\\[1.0cm]
\end{center}

\makeatletter
\def\support{\xdef\@thefnmark{}\@footnotetext}
\makeatother
\support{AA was supported by University of Latvia Grant \mbox{ZP01-100} (at University of Latvia) and CIAR, NSERC, ARO, MITACS and IQC University Professorship (at University of Waterloo). DWL was funded by the CRC, \mbox{ORF-RI}, NSERC, CIFAR, MITACS, and \mbox{Quantum-Works}.}

\begin{abstract}
\input{Abstract}
\end{abstract}

\begin{center}
  Supplementary materials are available on-line at \\
  \href{http://home.lanet.lv/~sd20008/racs}{\texttt{http://home.lanet.lv/$\sim$sd20008/racs}}
\end{center}

\end{titlepage}

\newpage
\setcounter{page}{2}
\tableofcontents
\newpage

%% file: Abstract.tex
We consider a communication method, where the sender encodes $n$ classical bits into $1$ qubit and sends it to the receiver who performs a certain measurement depending on which of the initial bits must be recovered. This procedure is called \racm{n}{1} quantum random access code (QRAC) where $p>1/2$ is its success probability. It is known that \racx{2}{0.85} and \racx{3}{0.79} QRACs (with no classical counterparts) exist and that \racp{4} QRAC with $p>1/2$ is not possible.

We extend this model with shared randomness (SR) that is accessible to both parties. Then \racp{n} QRAC with SR and $p>1/2$ exists for any $n \geq 1$. We give an upper bound on its success probability (the known \racx{2}{0.85} and \racx{3}{0.79} QRACs match this upper bound). We discuss some particular constructions for several small values of $n$.

We also study the classical counterpart of this model where $n$ bits are encoded into $1$ bit instead of $1$ qubit and SR is used. We give an optimal construction for such codes and find their success probability exactly---it is less than in the quantum case.

%% file: PaperQRACs.bbl
\begin{thebibliography}{99}

\bibitem{Wiesner}
Stephen Wiesner,
``Conjugate coding,''
\textit{SIGACT News}, vol.~15, issue~1, pp.~78--88, 1983.


\bibitem{DenseCoding1}
Andris Ambainis, Ashwin Nayak, Amnon Ta-Shma, Umesh Vazirani,
``Dense quantum coding and a lower bound for 1-way quantum automata,''
Proceedings of the 31st Annual ACM Symposium on Theory of Computing (STOC'99), pp.~376--383, 1999.
\href{http://arxiv.org/abs/quant-ph/9804043}{\texttt{arXiv:quant-ph/9804043v2}}

\bibitem{DenseCoding2}
Andris Ambainis, Ashwin Nayak, Amnon Ta-Shma, Umesh Vazirani,
``Dense Quantum Coding and Quantum Finite Automata,''
\textit{Journal of the ACM}, vol.~49, no.~4, pp.~496--511, 2002.

\bibitem{Nayak}
Ashwin Nayak,
``Optimal lower bounds for quantum automata and random access codes,''
Proceedings of the 40th IEEE Symposium on Foundations of Computer Science (FOCS'99), pp.~369--376, 1999.
\href{http://arxiv.org/abs/quant-ph/9904093}{\texttt{arXiv:quant-ph/9904093v3}}

\bibitem{No41}
Masahito Hayashi, Kazuo Iwama, Harumichi Nishimura, Rudy Raymond, Shigeru Yamashita,
``\mbox{$(4,1)$-Quantum} Random Access Coding Does Not Exist,''
\textit{New J. Phys.}, vol.~8, no.~8, pp.~129, 2006.
\href{http://arxiv.org/abs/quant-ph/0604061}{\texttt{arXiv:quant-ph/0604061v1}}


\bibitem{Galvao}
Ernesto F. Galv\~ao,
``Foundations of quantum theory and quantum information applications,''
PhD thesis, University of Oxford, 2002.

\bibitem{Spekkens}
Robert W. Spekkens, Daniel H. Buzacott, Anthony J. Keehn, Ben Toner, Geoff J. Pryde,
``Preparation contextuality powers parity-oblivious multiplexing,''
\textit{Phys. Rev. Lett.}, vol.~102, no.~1, 010401, 2009.
\href{http://arxiv.org/abs/0805.1463}{\texttt{arXiv:0805.1463v2}}

\bibitem{Severini}
Andrea Casaccino, Ernesto F. Galv\~ao, Simone Severini,
``Extrema of discrete Wigner functions and applications,''
\textit{Phys. Rev. A} \textbf{78}, 022310, 2008.
\href{http://arxiv.org/abs/0805.3466}{\texttt{arXiv:0805.3466v2}}


\bibitem{Klauck}
Hartmut Klauck,
``Lower bounds for quantum communication complexity,''
Proceedings of the 42nd IEEE Symposium on Foundations of Computer Science (FOCS'01), pp.~288, 2001.\\
\href{http://arxiv.org/abs/quant-ph/0106160}{\texttt{arXiv:quant-ph/0106160v3}}

\bibitem{Aaronson}
Scott Aaronson,
``Limitations of Quantum Advice and One-Way Communication,''
Proceedings of the 19th Annual IEEE Conference on Computational Complexity (CCC'04), pp.~320--332, 2004.\\
\href{http://arxiv.org/abs/quant-ph/0402095}{\texttt{arXiv:quant-ph/0402095v4}}

\bibitem{Gavinsky}
Dmitry Gavinsky, Julia Kempe, Oded Regev, Ronald de Wolf,
``Bounded-Error Quantum State Identification and Exponential Separations in Communication Complexity,''
Proceedings of the 38th Annual ACM Symposium on Theory of Computing (STOC'06), pp.~594--603, 2006.\\
\href{http://arxiv.org/abs/quant-ph/0511013}{\texttt{arXiv:quant-ph/0511013v1}}

\bibitem{NetworkCoding}
Masahito Hayashi, Kazuo Iwama, Harumichi Nishimura, Rudy Raymond, Shigeru Yamashita,
``Quantum Network Coding,''
Proceedings of the 24th International Symposium on Theoretical Aspects of Computer Science (STACS'07), pp.~610--621, 2007.
\href{http://arxiv.org/abs/quant-ph/0601088}{\texttt{arXiv:quant-ph/0601088v2}}


\bibitem{Kerenidis}
Iordanis Kerenidis,
``Quantum Encodings and Applications to Locally Decodable Codes and Communication Complexity,''
PhD thesis, University of California at Berkeley, 2004.

\bibitem{KerenidisDeWolf}
Iordanis Kerenidis, Ronald de Wolf,
``Exponential Lower Bound for 2-Query Locally Decodable Codes via a Quantum Argument,''
\textit{J. Comput. Syst. Sci.}, vol.~69, 3, pp.~395--420, 2004.\\
\href{http://arxiv.org/abs/quant-ph/0208062}{\texttt{arXiv:quant-ph/0208062v2}}

\bibitem{Wehner}
Stephanie Wehner, Ronald de Wolf,
``Improved Lower Bounds for Locally Decodable Codes and Private Information Retrieval,''
Automata, Languages and Programming, pp.~1424--1436, 2005.\\
\href{http://arxiv.org/abs/quant-ph/0403140}{\texttt{arXiv:quant-ph/0403140v2}}

\bibitem{Hypercontractive}
Avraham Ben-Aroya, Oded Regev, Ronald de Wolf,
``A Hypercontractive Inequality for Matrix-Valued Functions with Applications to Quantum Computing and LDCs,''
Proceedings of the 49th Annual IEEE Symposium on Foundations of Computer Science (FOCS'08), pp.~477--486, 2008.\\
\href{http://arxiv.org/abs/0705.3806}{\texttt{arXiv:0705.3806v2}}


\bibitem{AaronsonLearnability}
Scott Aaronson,
``The Learnability of Quantum States'',
\textit{Proc. Roy. Soc. London Ser. A}, vol.~463, no.~2088, pp.~3089--3114, 2007.
\href{http://arxiv.org/abs/quant-ph/0608142}{\texttt{arXiv:quant-ph/0608142v3}}


\bibitem{Yao}
Andrew Chi-Chin Yao,
``Probabilistic computations: Toward a unified measure of complexity,''
Proceedings of the 18th Annual IEEE Symposium on Foundations of Computer Science (SFCS'77), pp.~222--227, 1977.

\bibitem{PrinciplesOfQ}
Giuliano Benenti, Giulio Casati, Giuliano Strini,
``Principles of Quantum Computation and Information,''
vol.~1, World Scientific, 2004.

\bibitem{Stirling}
Eric W. Weisstein, ``Stirling's Approximation,'' MathWorld.\\
\href{http://mathworld.wolfram.com/StirlingsApproximation.html}
{\texttt{http://mathworld.wolfram.com/StirlingsApproximation.html}}

\bibitem{SpherePointPicking}
Eric W. Weisstein, ``Sphere Point Picking,'' MathWorld.\\
\href{http://mathworld.wolfram.com/SpherePointPicking.html}
{\texttt{http://mathworld.wolfram.com/SpherePointPicking.html}}

\bibitem{Chandrasekhar}
Subrahmanyan Chandrasekhar,
``Stochastic Problems in Physics and Astronomy,''
\textit{Reviews of Modern Physics}, vol.~15, no.~1, 1943.

\bibitem{Hughes}
Barry D. Hughes,
``Random Walks and Random Environments,'' vol.~1,
Clarendon Press, 1995.

\bibitem{FiniteReflectionGroups}
Larry C. Grove, Clark T. Benson,
``Finite Reflection Groups'', 2nd Ed.,
Springer, 1985.

\bibitem{TriangularSymmetryGroups}
Walter W. Rouse Ball, H.S.M. Coxeter,
``Mathematical Recreations and Essays,'' 13th Ed.,
Courier Dover Publications, 1987.



\bibitem{MUBs}
William K. Wootters, Brian D. Fields,
``Optimal State-Determination by Mutually Unbiased Measurements,''
\textit{Annals of Physics}, vol.~191, Issue~2, pp.~363--381, 1989.

\bibitem{SIC-POVMs}
Joseph M. Renes, Robin Blume-Kohout, Andrew J. Scott, Carlton M. Caves,
``Symmetric Informationally Complete Quantum Measurements,''
\textit{J. Math. Phys.}, vol.~45, pp.~2171--2180, 2004.\\
\href{http://arxiv.org/abs/quant-ph/0310075}{\texttt{arXiv:quant-ph/0310075v1}}


\bibitem{Kimura}
Gen Kimura,
``The Bloch vector for N-level systems,''
\textit{Physics Letters A}, vol.~314, Issues~5--6, pp.~339--349, 2003.
\href{http://arxiv.org/abs/quant-ph/0301152}{\texttt{arXiv:quant-ph/0301152v2}}

\bibitem{Kimura2}
Gen Kimura, Andrzej Kossakowski,
``The Bloch-Vector Space for N-Level Systems: the Spherical-Coordinate Point of View,''
\textit{Open Syst. Inf. Dyn.}, vol.~12, no.~3, pp.~207--229, 2005.\\
\href{http://arxiv.org/abs/quant-ph/0408014}{\texttt{arXiv:quant-ph/0408014v2}}

\bibitem{Positivity}
Mark S. Byrd, Navin Khaneja,
``Characterization of the Positivity of the Density Matrix in Terms of the Coherence Vector Representation,''
\textit{Phys. Rev. A}, vol.~68, Issue~6, 062322, 2003.\\
\href{http://arxiv.org/abs/quant-ph/0302024}{\texttt{arXiv:quant-ph/0302024v2}}

\bibitem{SUGenerators}
Foek T. Hioe, Joseph H. Eberly,
``N-Level Coherence Vector and Higher Conservation Laws in Quantum Optics and Quantum Mechanics,''
\textit{Phys. Rev. Lett.}, vol.~47, pp.~838--841, 1981.

\bibitem{ComplementarityPolytope}
Ingemar Bengtsson, \AA sa Ericsson,
``Mutually Unbiased Bases and the Complementarity Polytope,''
\textit{Open Syst. Inf. Dyn.}, vol.~12, no.~2, pp.~107--120, 2005.
\href{http://arxiv.org/abs/quant-ph/0410120}{\texttt{arXiv:quant-ph/0410120v1}}


\bibitem{NielsenChuang}
Michael A. Nielsen, Isaac L. Chuang,
``Quantum Computation and Quantum Information,''
Cambridge University Press, 2000.

\bibitem{Peres}
Asher Peres,
``Quantum Theory: Concepts and Methods,''
Kluwer Academic Publishers, 2002.

\bibitem{MatrixAnalysis}
Roger A. Horn, Charles R. Johnson,
``Matrix Analysis,''
Cambridge University Press, 1985.

\end{thebibliography}
